\documentclass[11pt]{article}
\usepackage[T1]{fontenc}
\usepackage[latin9]{inputenc}
\usepackage{geometry}
\usepackage{color}
\usepackage{prettyref}
\usepackage{amsmath}
\usepackage{amsthm}
\usepackage{amssymb}
\usepackage{setspace}
\usepackage[authoryear]{natbib}
\makeatletter
\theoremstyle{plain}
\newtheorem{assumption}{\protect\assumptionname}
\theoremstyle{plain}
\newtheorem{lem}{\protect\lemmaname}[section]
\theoremstyle{plain}
\newtheorem{thm}{\protect\theoremname}[section]
\theoremstyle{remark}
\newtheorem{rem}{\protect\remarkname}[section]
\theoremstyle{definition}
\newtheorem{condition}{\protect\conditionname}

\usepackage{lmodern}
\usepackage{authblk}
\usepackage{threeparttable} 
\usepackage{tabu}
\usepackage{tabularx}
\usepackage{bbm}
\usepackage{amsmath}
\usepackage{lscape}
\usepackage{float}

\usepackage{array}
\newcolumntype{C}{>{\centering\let\newline\\\arraybackslash\hspace{0pt}}X}
\newcolumntype{R}{>{\raggedleft\let\newline\\\arraybackslash\hspace{0pt}}X}

\usepackage[titletoc,title]{appendix}

\makeatother

\providecommand{\assumptionname}{Assumption}
\providecommand{\conditionname}{Condition}
\providecommand{\lemmaname}{Lemma}
\providecommand{\remarkname}{Remark}
\providecommand{\theoremname}{Theorem}

\begin{document}
\title{Efficient Peer Effects Estimators with Group Effects}
\author{Guido M. Kuersteiner, Ingmar R. Prucha, and Ying Zeng\thanks{Kuersteiner: Department of Economics, University of Maryland, College
Park, MD 20742, United States (gkuerste@umd.edu); Prucha: Department
of Economics, University of Maryland, College Park, MD 20742, United
States (prucha@umd.edu); Zeng: Department of Public Finance, School
of Economics, Xiamen University, Xiamen 361005, China (zengying17@xmu.edu.cn).}}
\maketitle
\begin{abstract}
We study linear peer effects models where peers interact in groups
and individual's outcomes are linear in the group mean outcome and
characteristics. We allow for unobserved random group effects as well
as observed fixed group effects. The specification is in part motivated
by the moment conditions imposed in \citet{graham_identifying_2008}.
We show that these moment conditions can be cast in terms of a linear
random group effects model and that they lead to a class of GMM estimators
with parameters generally identified as long as there is sufficient
variation in group size or group types. We also show that our class
of GMM estimators contains a Quasi Maximum Likelihood estimator (QMLE)
for the random group effects model, as well as the Wald estimator
of \citet{graham_identifying_2008} and the within estimator of \citet{lee_identification_2007}
as special cases. Our identification results extend insights in \citet{graham_identifying_2008}
that show how assumptions about random group effects, variation in
group size and certain forms of heteroscedasticity can be used to
overcome the reflection problem in identifying peer effects. Our QMLE
and GMM estimators accommodate additional covariates and are valid
in situations with a large but finite number of different group sizes
or types. Because our estimators are general moment based procedures,
using instruments other than binary group indicators in estimation
is straight forward. Our QMLE estimator accommodates group level covariates
in the spirit of Mundlak and Chamberlain and offers an alternative
to fixed effects specifications. This model feature significantly
extends the applicability of Graham's identification strategy to situations
where group assignment may not be random but correlation of group
level effects with peer effects can be controlled for with observable
group level characteristics. Monte-Carlo simulations show that the
bias of the QMLE estimator decreases with the number of groups and
the variation in group size, and increases with group size. We also
prove the consistency and asymptotic normality of the estimator under
reasonable assumptions.

\newpage{}
\end{abstract}

\section{Introduction}

\global\long\def\diag{\operatorname{diag}}%
\global\long\def\tr{\operatorname{tr}}%
\global\long\def\vecd{\operatorname{vec_{D}}}%
\global\long\def\plim{\operatorname{plim}}%
\global\long\def\cov{\operatorname{Cov}}%
\global\long\def\var{\operatorname{Var}}%
Peer effects are of great interest to empirical researchers and policy
makers. The idea that individuals are affected by their peers motivates
policies that try to manipulate peer composition for better outcomes.
Peer effects are often confounded by group level effects. An example
are teacher effects in a class room setting. Identifying peer effects
is notoriously challenging due to the reflection problem \citep{manski_identification_1993,angrist_perils_2014}
as well as due to spurious peer effects originating from group level
effects. Random group allocation may be one way to overcome these
identification problems. With groups formed at random, a random effects
specification for group level characteristics can be adopted. An alternative
approach consists in postulating that, conditional on observed group
level characteristics, group level effects can be viewed as randomly
assigned. Regression control techniques based on observed group characteristics
then lead to a similar random effects specification, but without the
need to appeal to random group assignment. We propose estimators that
can accommodate both scenarios.

Random group assignment plays a prominent role in the empirical peer
effects literature in a number of fields including education, labor,
firm, finance and development studies. Recent examples from this literature
include \textcolor{black}{\citet{sacerdote_peer_2001,duflo_role_2003,zimmerman_peer_2003,stinebrickner_what_2006,kang_classroom_2007,graham_identifying_2008,guryan_peer_2009,carrell_does_2009,carrell_natural_2013,duflo_peer_2011,sojourner_identification_2013,booij_ability_2017,garlick_academic_2018,fafchamps_networks_2018,cai_interfirm_2018,frijters_heterogeneity_2019}.}
Assuming group effects to be independent of observed individual and
group characteristics is plausible when groups are formed at random.
Ignoring group effects or assuming fixed group effects \citep{lee_identification_2007}
leads to consistent but less efficient estimators. Random group effects
themselves have important empirical interpretations. For example,
researchers in education policy often treat random class effects as
unobserved teacher effects (e.g., \citealt{nye_how_2004,rivkin_teachers_2005,chetty_how_2011}).
Absent random group assignment, the estimators we propose can accommodate
observed group level effects that can come from information about
group characteristics such as the training and experience of teachers,
or averages of individual group member characteristics. Group level
characteristics can be interpreted as parametrizations of group effects
in the spirit of \citet{mundlak_pooling_1978} and \citet{chamberlain_analysis_1980}.
The choice between a random effects or fixed effects estimator then
depends less on random group assignment but more on whether group
specific effects are believed to be observable or not. In some cases
there may be independent interest in the effects of group specific
covariates. An example is the effect of teacher training on student
performance. In such cases a random effects estimator is the preferred
choice because fixed effects estimators are often unable to identify
these types of group level effects. 

Our analysis extends insights in \citet{graham_identifying_2008}
that show how assumptions about random group effects, variation in
group size and certain forms of heteroscedasticity can be used to
overcome the reflection problem in identifying peer effects. We give
an interpretation of the conditional variance estimator (CVE) of \citet{graham_identifying_2008}
in terms of a GMM estimator based on moment conditions for the within-group
variance and between-group variance. We show that the moment conditions
underlying \citet{graham_identifying_2008} are the score function
of a quasi maximum likelihood estimator (QMLE) for a random group
effects model. The QMLE can be shown to be the best GMM estimator
in the class of estimators using moment conditions for the within
and between variances of outcomes individually, rather than combining
them into a single moment condition as is the case for the CV estimator,
or focusing only on the within variation as is the case of the CMLE
of \citet{lee_identification_2007}. 

One limitation of the conditional variance estimator proposed by Graham
is the fact that it amounts to a difference in difference identification
strategy for the variances that requires groups to fall into two size
categories. As shown in \citet{graham_identifying_2008} the resulting
procedure takes the form of a Wald estimator for a set of binary instruments.
This setting is restrictive in applications where groups may not be
easily separated into two categories or where a more general set of
instruments needs to be considered. The estimators that we propose
are general GMM based procedures that accommodate additional covariates
as well as offer flexibility in terms of the instruments and the number
of moment conditions that are being used. We illustrate these points
by explicitly considering moment based estimators that exploit exogenous
variation in group size as well as general group level heteroscedasticity
as instruments. In contrast to the CMLE of \citet{lee_identification_2007},
which is a member of the class of GMM estimators we consider, our
QMLE uses both the within and between variance. This leads to efficiency
gains under correct specification but comes at the cost of potential
miss-specification bias if the assumption of observed fixed and unobserved
random group effects is incorrect. The trade-offs are similar to related
results for fixed and random effects in the panel literature.

Our work is also related to the literature in spatial econometrics
started by the work of \citet{cliff_spatial_1973,cliff_spatial_1981}
and \citet{anselin_spatial_1988}.\footnote{\citet{anselin_thirty_2010} offers a brief review of the development
of spatial econometrics literature over the past thirty years.} Recently, there is a growing number of studies using spatial methods
to model social network effects, e.g., \citet{lee_identification_2007},
\citet{bramoulle_identification_2009}, and \citet{kuersteiner_dynamic_2020}.
The strength of social links can be characterized by proximity in
the social network space. We extend \citet{kelejian_estimation_2006}
and \citet{lee_identification_2007} by considering a random group
effects specification. Spatial models were traditionally estimated
with maximum likelihood (ML), e.g., \citet{ord_estimation_1975}.
\citet{kelejian_generalized_1998,kelejian_generalized_1999} develop
generalized method of moments (GMM) estimators based on linear and
quadratic moments. While this paper utilizes a quasi-maximum likelihood
estimation method, the score function depends on linear quadratic
forms of the error terms. Properties of quadratic moment conditions
were introduced by \citet{kelejian_generalized_1998,kelejian_generalized_1999}
in the cross section case, and \citet{kapoor_panel_2007} and \citet{kuersteiner_dynamic_2020}
in a panel setting. Moreover, \citet{kelejian_asymptotic_2001} and
\citet{kelejian_specification_2010} develop a central limit theorem
for linear quadratic forms, which is the basis for the asymptotic
analysis in this paper.

The linear-in-means peer effect model in \citet{manski_identification_1993}
is a special case of a spatial model with group-wise equal dependence,
see \citet{kelejian_2sls_2002} and \citet{kelejian_estimation_2006}.
\citet{kelejian_2sls_2002} were the first to study the group-wise
equal dependence spatial model. They show that if there is one group
in a single cross section and the model has equal spatial weights,
two-stage least squares (2SLS), GMM and QMLE methods all yield inconsistent
estimators, although consistent estimation with 2SLS and GMM is possible
for panel data. However, \citet{kelejian_estimation_2006} point out
that if group fixed effects are incorporated and the panel is balanced,
the estimators are inconsistent. The results in \citet{kelejian_estimation_2006}
show the importance of variation in group size in identification of
spatial models with blocks of equal weights. The QMLE developed in
this paper and the conditional maximum likelihood estimator in \citet{lee_identification_2007}
both rely on group size variation for identification although we show
that identification exploiting heteroscedastic errors is also possible.
Extensions include \citet{lee_specification_2010} who allow for specific
social structure within each group and \citet{liu_gmm_2010} and \citet{liu_endogenous_2014}
who allow for non-row normalized weight matrices. The linear spatial
model has also been applied to the empirical evaluation of peer effects
by \citet{lin_identifying_2010} and \citet{boucher_peers_2014}.
\citet{bramoulle_identification_2009} study a broader range of social
interaction models and give conditions for identification.

The paper is organized as follows. In Section \ref{sec:graham} we
consider identification of endogenous peer effects in a simple setting
without covariates for the CV, CML and QML estimators. Section \ref{sec:General-Model}
presents the full model that allows for covariates and general variation
in group size. Section \ref{sec:Theoretical-Results} summarizes the
technical conditions we impose and presents theoretical results for
the QMLE. Section \ref{sec:Monte_Carlo} contains a small Monte Carlo
experiment. Proofs are collected in an appendix.

\section{Peer Effects with Random Group Effects\label{sec:graham}}

We start the discussion by presenting a simple model without covariates,
to introduce and discuss basic features of our new quasi-maximum likelihood
estimator (QMLE), and connect it to the conditional variance (CV)
estimator in \citet{graham_identifying_2008} and the conditional
maximum likelihood (CMLE) estimator in \citet{lee_identification_2007}.
The model decomposes variation in outcomes of a cross-section of individuals
into idiosyncratic noise, group level random effects and correlation
that is due to group level interaction. Quadratic moment conditions
implied by this random effects specification lead to efficient GMM,
quasi maximum likelihood, and under additional distributional assumptions,
maximum likelihood estimators. Estimators based on these moment conditions
include the CV estimator of \citet{graham_identifying_2008}, the
QMLE as well as the CMLE of \citet{lee_identification_2007} as special
cases.

Let $y_{ir}$ be an observed outcome of individual $i$ in group $r$
which has $m_{r}$ members, let $\alpha_{r}$ be an unobserved group
level effect and let $\epsilon_{ir}$ be unobserved individual specific
characteristics. We observe data for $R$ groups as well as a categorical
variable $D_{r}$ which determines group type. An example is when
there are three group sizes such that $D_{r}\in\left\{ 'small','medium','large'\right\} .$
However, $D_{r}$ could be a characteristic that is not necessarily
related to group size. An example is when groups are defined by classrooms
of schools in urban, suburban or rural districts and $D_{r}$ is used
to denote urbanicity. Classes could also be categorized by sociodemographic
composition such as whether English or other languages are the native
language spoken by students in the class. We allow for type-dependent
heteroscedasticity. Types add flexibility to the specification by
relaxing the constraints the model imposes on the relationship between
group variance and group size. In some cases type specific heteroscedasticity
provides identifying variation that is separate from group size variation.

The peer effects model is stated in terms of a structural equation
\begin{equation}
y_{ir}=\lambda\bar{y}_{\left(-i\right)r}+\alpha_{r}+\epsilon_{ir},\label{eq:Graham_Mod_L-M-O}
\end{equation}
where $\bar{y}_{\left(-i\right)r}=\frac{1}{m_{r}-1}\sum_{j\neq i}^{m_{r}}y_{jr}$
is the leave-out-mean of the outcome variable. The parameter $\lambda$
captures the endogenous peer effects, see \citet{manski_identification_1993}.
The structural form emphasizes the decomposition of $y_{ir}$ into
a social interaction term $\lambda\bar{y}_{(-i)r}$, a group level
effect $\alpha_{r}$ and an idiosyncratic error term $\epsilon_{ir}.$
For example, when $y_{ir}$ is a measure of student performance and
$r$ is a class-room index then $\alpha_{r}$ can be interpreted as
a class-room or teacher effect while $\epsilon_{ir}$ are unobserved
student characteristics for student $i$ in classroom $r.$ Cross-sectional
independence of $\epsilon_{ir}$ can be justified by random group
assignment such as in the application of Graham (2008). The assumptions
we impose on $\epsilon_{ir}$ and $\alpha_{r}$ are in line with the
random effects panel literature where group level dependence of unobservables
is modeled with the common factor $\alpha_{r}.$ We leave possible
generalizations of this framework to cases where $\epsilon_{ir}$
is allowed to be dependent for future work.

Following Graham (2008) who emphasizes random assignments of individuals
to groups, we assume that $\alpha_{r}$ is a random effect independent
of $\epsilon_{ir}.$ As shown by Graham (2008) for a slightly different
model based on full rather than leave-out means, the random effects
nature of the model leads to a set of quadratic moment conditions
that can be exploited for identification. We expand on these ideas
by showing that the implied moment conditions are related to the moment
conditions of a random effects pseudo likelihood estimator. Transformations
of these moments turn out to coincide with moments used by \citet{graham_identifying_2008}
as well as \citet{lee_identification_2007} who considers a fixed
effects version of the model. \citet{lee_identification_2007} focuses
on identification of $\lambda$ based on group size variation. Here
we emphasize a random effects specification where identification is
driven by heterogeneity at the group level that could result from
sources including but not limited to class size variation. A literature
on linear instrumental variables methods gives conditions under which
$\lambda$ can be identified in models that have additional exogenous
covariates $Z_{r}$, e.g., \citet{angrist_perils_2014} or \citet{bramoulle_identification_2009}.\footnote{The leave-out-mean $\bar{y}_{\left(-i\right)r}$ can be viewed as
a special case of a Cliff-Ord-type (\citealt{cliff_spatial_1973,cliff_spatial_1981})
spatial lag. \citet{kelejian_generalized_1998} give an early basic
condition for identification by IV.} Besides the conventional instrumental variables strategies, alternative
strategies are also available, see \citet{lee_identification_2007},
\citet{graham_identifying_2008} for a modified model or \citet{kuersteiner_dynamic_2020}.
Letting $Y_{r}=\left(y_{1r},...,y_{m_{r}r}\right)^{'}$, $\epsilon_{r}=\left(\epsilon_{1r},...,\epsilon_{m_{r}r}\right)^{'}$,
$\iota_{m_{r}}=\left(1,...,1\right)^{\prime}$ and $W_{m_{r}}=\frac{1}{m_{r}-1}(\iota_{m_{r}}\iota_{m_{r}}^{\prime}-I_{m_{r}})$,
the model can be written in matrix notation as 
\begin{equation}
Y_{r}=\lambda W_{m_{r}}Y_{r}+\alpha_{r}\iota_{m_{r}}+\epsilon_{r}.\label{eq:Model_Graham_Modified}
\end{equation}
To isolate or identify the social interaction effect, we impose the
following restrictions on unobservables. 
\begin{assumption}
\label{assume:epsilon} For $r=1,\ldots,R$ the $r$-th group is associated
with a categorical variable $D_{r}\in\{1,2,...,J\}$ with $J\geqslant1$
being fixed and finite, and for each category $j\in\{1,2,...,J\}$
there is at least one group $r$ with $D_{r}=j$. For $r=1,...,R$
and $i=1,...,m_{r}$ the disturbance terms $\epsilon_{ir}$ are independently
distributed across all $i$ and $r$, with $E\left[\epsilon_{ir}|D_{r},m_{r}\right]=0$
and $E\left[\epsilon_{ir}^{2}|D_{r},m_{r}\right]=\sigma_{\epsilon0,D_{r}}^{2}$
, $0<\underbar{\ensuremath{a}}_{\epsilon}\leqslant\sigma_{\epsilon0,D_{r}}^{2}\leqslant\overline{\text{\ensuremath{a}}}_{\epsilon}<\infty$
and where $\sigma_{\epsilon0,D_{r}}^{2}$ is a function only of $D_{r}$.
There exists some $\eta_{\epsilon}>0$ such that $E[|\epsilon_{ir}|^{4+\eta_{\epsilon}}]<\infty$.
\end{assumption}
Note that the variance $E\left[\epsilon_{ir}^{2}|D_{r},m_{r}\right]=\sigma_{\epsilon0,D_{r}}^{2}$
has the representation $\sigma_{\epsilon0,D_{r}}^{2}=\sigma_{\epsilon0,1}^{2}1\left\{ D_{r}=1\right\} +...+\sigma_{\epsilon0,J}^{2}1\left\{ D_{r}=J\right\} $
where $\sigma_{\epsilon0,1}^{2},....,\sigma_{\epsilon0,J}^{2}$ are
fixed parameters to be estimated.
\begin{assumption}
\label{assume:alpha}For $r=1,...,R$, the group effects $\alpha_{r}$
are independently and identically distributed, with $E\left[\alpha_{r}|D_{r},m_{r}\right]=0$
and $E\left[\alpha_{r}^{2}|D_{r},m_{r}\right]=\sigma_{\alpha0}^{2}$,
where $\text{\ensuremath{0\leq}}\sigma_{\alpha0}^{2}\leq\overline{\text{\ensuremath{a}}}_{\alpha}<\infty$.
There exists some $\eta_{\alpha}>0$ such that $E\left[|\alpha_{r}|^{4+\eta_{\alpha}}\right]<\infty$.
Also, $\{\alpha_{r}:r=1,...,R\}$ are independent of $\{\epsilon_{ir}:i=1,...,m_{r};r=1,...,R\}$.
\end{assumption}
Assumption \ref{assume:epsilon} implies in particular that individuals
do not self select into groups based on unobserved characteristics
and Assumption \ref{assume:alpha} suggests that there is no matching
between group characteristics and individual characteristics. This
no sorting or matching assumption can sometimes be motivated by specific
empirical designs. For example, in the Project STAR experiment that
\citet{graham_identifying_2008} considers, kindergarten students
and teachers are randomly assigned to classrooms. This random assignment
mechanism justifies interpreting $\alpha_{r}$ as the classroom or
teacher effect. It also justifies assuming that $\alpha_{r}$ and
$\epsilon_{ir}$ are mutually independent random variables, see Graham
(2008) Assumption 1.1. Assumption \ref{assume:epsilon} allows $\epsilon_{ir}$
to be homoscedastic across all groups when $J=1$ or heteroscedastic
across different categories of $D_{r}$ when $J\geqslant2$. This
formulation contains the case considered by \citet{graham_identifying_2008}
where $J=2$ as a special case.

Assumptions \ref{assume:epsilon} and \ref{assume:alpha} above imply
moment conditions. These moment conditions take the form of restrictions
on the within and between group variance. As discussed in more detail
below, these moment conditions are fundamental to the ML estimator.
In particular, we show that the score of the ML estimator is a weighted
average of those fundamental moment conditions.

To derive the moment conditions, define the composite error term $U_{r}=\alpha_{r}\iota_{m_{r}}+\epsilon_{r}$
where $U_{r}$ is an $m_{r}\times1$ vector with elements $u_{ir}=\alpha_{r}+\epsilon_{ir}.$
Let $\bar{u}_{r}$ and $\bar{\epsilon}_{r}$ be the mean of $u_{ir}$
and $\epsilon_{ir}$ in group $r$. Let $\ddot{U}_{r}=U_{r}-\bar{u}_{r}\iota_{m_{r}}$
be the vector of within-group deviations from the mean of $U_{r}$
and let $\ddot{Y}_{r}$ and $\ddot{\epsilon}_{r}$ be defined in a
similar manner. It can be shown that $\bar{y}_{r}=\bar{u}_{r}/(1-\lambda)=(\alpha_{r}+\bar{\epsilon}_{r})/(1-\lambda)$
with $\bar{u}_{r}=\alpha_{r}+\bar{\epsilon}_{r}$, and $\ddot{Y}_{r}=\frac{m_{r}-1}{m_{r}-1+\lambda}\ddot{U}_{r}=\frac{m_{r}-1}{m_{r}-1+\lambda}\ddot{\epsilon}_{r}$.
Two conditional moment conditions, one for the within-group variance,
the other for the between group variance, arise from the model in
(\ref{eq:Model_Graham_Modified}) under Assumptions \ref{assume:epsilon}
and \ref{assume:alpha}. The expected value of the within-group and
between-group squares of group $r$ are

\begin{equation}
var_{r}^{w}=E\left[\frac{\ddot{Y}_{r}^{\prime}\ddot{Y}_{r}}{m_{r}-1}|m_{r},D_{r}\right]=E\left[\frac{(m_{r}-1)\ddot{U}_{r}^{\prime}\ddot{U}_{r}}{(m_{r}-1+\lambda)^{2}}|m_{r},D_{r}\right]=\frac{\left(m_{r}-1\right)^{2}}{\left(m_{r}-1+\lambda\right)^{2}}\sigma_{\epsilon,D_{r}}^{2},\label{eq:varw}
\end{equation}
\begin{equation}
var_{r}^{b}=E\left[\bar{y}_{r}^{2}|m_{r},D_{r}\right]=E\left[\frac{\bar{u}_{r}^{2}}{(1-\lambda)^{2}}|m_{r},D_{r}\right]=\frac{1}{(1-\lambda)^{2}}\left(\sigma_{\alpha}^{2}+\frac{\sigma_{\epsilon,D_{r}}^{2}}{m_{r}}\right),\label{eq:varb}
\end{equation}
where $\sigma_{\epsilon,D_{r}}^{2}=\sigma_{\epsilon,1}^{2}1\left\{ D_{r}=1\right\} +...+\sigma_{\epsilon,J}^{2}1\left\{ D_{r}=J\right\} .$

To see how these moment conditions can achieve the identification
of $\lambda$ consider the case where $\sigma_{\epsilon,D_{r}}^{2}=\sigma_{\epsilon,D_{s}}^{2}$
but $m_{r}\neq m_{s}$. Then, Equation \eqref{eq:varw} implies that
\begin{equation}
(\frac{m_{r}-1+\lambda}{m_{s}-1+\lambda})^{2}=\frac{(m_{r}-1)^{3}}{(m_{s}-1)^{3}}\frac{E\left[\ddot{Y}_{s}^{\prime}\ddot{Y}_{s}|m_{s},D_{s}\right]}{E\left[\ddot{Y}_{r}^{\prime}\ddot{Y}_{r}|m_{r},D_{r}\right]}.\label{eq:Wald_within}
\end{equation}
Alternatively consider the case where $m_{r}=m_{s}=m$ and $\sigma_{\epsilon,D_{r}}^{2}\neq\sigma_{\epsilon,D_{s}}^{2}$,
then combining \eqref{eq:varw} and \eqref{eq:varb} gives
\begin{equation}
\frac{(m-1+\lambda)^{2}}{(1-\lambda)^{2}}=\frac{E\left[\bar{y}_{r}^{2}|m,D_{r}\right]-E\left[\bar{y}_{s}^{2}|m,D_{s}\right]}{E\left[\ddot{Y}_{r}^{\prime}\ddot{Y}_{r}/[m(m-1)^{3}]|m,D_{r}\right]-E\left[\ddot{Y}_{s}^{\prime}\ddot{Y}_{s}/[m(m-1)^{3}]|m,D_{s}\right]}.\label{eq:Wald_L-O-M}
\end{equation}
Expressions on the left hand side of both \eqref{eq:Wald_within}
and \eqref{eq:Wald_L-O-M} in principle can be solved for $\lambda$
if we restrict $\lambda\in(-1,1)$ and $m_{r}\geqslant2$, as both
expressions are monotonic functions of $\lambda$. Equation \eqref{eq:Wald_L-O-M}
is a modified version of Equation (9) in \citet{graham_identifying_2008}
that accounts for the leave-out-mean specification we consider. The
numerator differences out the variance of $\alpha_{r}$ which is assumed
constant across types. This restriction is also imposed by \citet{graham_identifying_2008}
in his Assumption 1.2. In Lemma \ref{lem:ID_1} below we outline the
exact conditions under which identification is possible.

The discussion above shows that under additional assumptions on $\lambda$
and group size, identification of $\lambda$ is possible through moment
conditions related to within and between variance when there is variation
in either group size $m_{r}$ or idiosyncratic error variance $\sigma_{\epsilon,D_{r}}^{2}$.
We now formalize the discussion into Lemma \ref{lem:ID_1} below.
Let the parameter vector be $\theta=\left(\lambda,\sigma_{\alpha}^{2},\sigma_{\epsilon,1}^{2},...,\sigma_{\epsilon,J}^{2}\right)^{\prime}$
and, for clarity, let the true parameter vector be denoted by $\theta_{0}=\left(\lambda_{0},\sigma_{\alpha0}^{2},\sigma_{\epsilon0,1}^{2},...,\sigma_{\epsilon0,J}^{2}\right)^{\prime}$.
For identification, we further assume that group size $m_{r}\geqslant2$
and impose the following assumption on $\lambda$.
\begin{assumption}
\label{assume:lambda} The parameter of the endogenous peer effects
$\lambda_{0}\in\Lambda$, where $\Lambda$ is a compact subset of
$(-1,1)$. Assume that $\theta_{0}\in\Theta$ with $\Theta=\Lambda\times\text{\ensuremath{[0,\overline{\text{\ensuremath{a}}}_{\alpha}]\times[\underbar{\ensuremath{a}}_{\epsilon},\overline{\text{\ensuremath{a}}}_{\epsilon}]\times\ldots\times[\underbar{\ensuremath{a}}_{\epsilon},\overline{\text{\ensuremath{a}}}_{\epsilon}]}}$
compact. 
\end{assumption}
The estimation procedures we propose in this paper can be implemented
with the availability of a general set of valid instruments and are
valid for cases where $J\geqslant1$ as long as $J$ is fixed and
finite. In the simple model without covariates the available instruments
are group size $m_{r}$ and categorical variable $D_{r}$. These instruments
are valid if assignment to groups is random in a way that generates
random variation in group size or category. Utilizing Equation \eqref{eq:varw}
and \eqref{eq:varb}, and using group size $m_{r}$ and the categorical
variable $D_{r}$ as instruments yields the following conditional
moment restriction $E[\chi_{r}(\theta_{0})|m_{r},D_{r}]=0$ with 
\begin{equation}
\chi_{r}(\theta)=\left[\begin{array}{c}
\chi_{r}^{w}(\theta)\\
\chi_{r}^{b}(\theta)
\end{array}\right]=\left[\begin{array}{c}
\frac{(m_{r}-1+\lambda)^{2}\ddot{Y}_{r}^{\prime}\ddot{Y}_{r}}{(m_{r}-1)^{2}}-(m_{r}-1)\sigma_{\epsilon,D_{r}}^{2}\\
(1-\lambda)^{2}\bar{y}_{r}^{2}-\sigma_{\alpha}^{2}-\frac{\sigma_{\epsilon,D_{r}}^{2}}{m_{r}}.
\end{array}\right].\label{eq:mwb}
\end{equation}

Identification of the parameter $\theta$ is possible with variation
in group size for a given category or variation in the idiosyncratic
variance over categories for the same group size. This is summarized
in the following lemma. The proof of the lemma is given in Appendix
\ref{sec:prooflemma2}.
\begin{lem}
\label{lem:ID_1}Suppose Assumptions \ref{assume:epsilon}-\ref{assume:lambda}
hold. Then the parameter $\theta_{0}$ is identified under the following
two scenarios:

(i) There are two groups $r$ and $s$ such that $m_{r}\neq m_{s}$
and $D_{r}=D_{s}$, and therefore $\sigma_{\epsilon0,D_{r}}^{2}=\sigma_{\epsilon0,D_{s}}^{2}$.
Then the parameter $\theta_{0}$ is identified in $\Theta.$ In particular,
the moment conditions $E[\chi_{q}^{w}(\theta)|m_{q},D_{q}]=0$ and
$E[\chi_{q}^{b}(\theta)|m_{q},D_{q}]=0$ for $q=r,s$ with $\chi_{r}^{w}(\theta)$
and $\chi_{r}^{b}(\theta)$ defined in (\ref{eq:mwb}) identify $\lambda_{0}$,
$\sigma_{\epsilon0,D_{r}}^{2}$,and $\sigma_{\alpha0}^{2}$. The remaining
parameters $\sigma_{\epsilon0,j}^{2}$ are identified by $E(\chi_{q}^{w}(\theta)|m_{q},D_{q})=0$
for $q\neq r$ or $s.$ 

(ii) There are two groups $r$ and $s$, such that $m_{r}=m_{s}$
and $\sigma_{\epsilon0,D_{r}}^{2}\neq\sigma_{\epsilon0,D_{s}}^{2}$.
Then the parameter $\theta_{0}$ is identified in $\Theta.$ In particular,
the moment condition $E[\nu_{q}(\theta)|m_{q},D_{q}]=0$, $q=r,s$
uniquely identifies $\lambda_{0}$ and $\sigma_{\alpha0}^{2}$, where
\begin{align}
\nu_{q}(\theta) & =\chi_{q}^{b}(\theta)-\frac{\chi_{q}^{w}(\theta)}{m_{q}(m_{q}-1)}=(1-\lambda)^{2}\bar{y}_{q}^{2}-\sigma_{\alpha}^{2}-\frac{(m_{q}-1+\lambda)^{2}\ddot{Y}_{q}^{\prime}\ddot{Y}_{q}}{m_{q}(m_{q}-1)^{3}}\label{eq:nu}
\end{align}
with $\chi_{q}^{w}(\theta)$ and $\chi_{q}^{b}(\theta)$ defined in
(\ref{eq:mwb}). The remaining parameters $\sigma_{\epsilon0,j}^{2}$
are identified by $E(\chi_{q}^{w}(\theta)|m_{q},D_{q})=0$.
\end{lem}
Full identification is achieved in Scenario (i) with group size variation
in at least one category. As an example, consider types that describe
urbanicity such that $D_{r}=D_{s}=1$ denotes two classrooms $r$
and $s$ that are both located in an urban school but where $m_{r}\neq m_{s}$
such that the classrooms differ in size, while the remaining categories
$d=2,...,J$ may have the same group sizes. In this setting $\theta_{0}$
is identified without any further constraints on the variances $\sigma_{\epsilon,j}^{2}$.
If the number of distinct group sizes exceeds the number of categories
$J$ then it automatically must be the case that there exist some
category that is associated with at least two distinct group sizes.
Note that the result holds irrespective of whether the constraint
of homoscedastic errors $\sigma_{\epsilon0,D_{r}}^{2}=\sigma_{\epsilon0,D_{s}}^{2}$
is imposed on the model or not. From Scenario (i) we see that variation
in group size alone can provide variation that is sufficient for identification.
Furthermore, in the homoscedastic case where only a common variance
parameter $\sigma_{\epsilon}^{2}$ is specified, two distinct group
sizes are sufficient for identification by the result in Scenario
(i). This corresponds to the identification result of the conditional
maximum likelihood estimator (CMLE) in \citet{lee_identification_2007},
the score function of which can be written as $\varphi(m_{r})\chi_{r}^{w}(\theta)$,
where $\varphi(m_{r})$ is a function of $m_{r}$.

While variation in group size serves as the source of identification
in Scenario (i), identification based on the moment condition $E[\chi_{r}(\theta)|m_{r},D_{r}]=0$
is also possible without group size variation as long as there is
some other form of group heterogeneity. As is shown in the proof for
Scenario (ii) of Lemma \ref{lem:ID_1}, utilizing $m_{q}=m$ and $E\left[\nu_{q}(\theta)|m_{q},D_{q}\right]=0$
for $q=r,s$ yields \eqref{eq:Wald_L-O-M}. From \eqref{eq:Wald_L-O-M}
we see that the endogenous peer effect parameter $\lambda$ is identified
if there is heteroscedasticity across groups of the same size for
at least one size, and that $\lambda$ can be estimated from the sample
analog of \eqref{eq:Wald_L-O-M}. The intuition of identification
in Scenario (ii) echoes that of the conditional variance (CV) estimator
of \citet{graham_identifying_2008}. Similar to Graham (2008), \eqref{eq:nu}
is based on the relationship between the within-group and between-group
variance as captured by $\nu_{r}(\theta)$ , and can be used to construct
a Wald type moment condition like in \eqref{eq:Wald_L-O-M} using
the categorical variable as the instrument.

The above discussion focused on identification based on the moment
vector $\chi_{r}(\theta)$. We next discuss the importance of these
moment conditions for efficient estimation, and their relationship
to the score of the Gaussian ML estimator. The optimal moment function
corresponding to $\chi_{r}(\theta)$ is given by $\chi_{r}^{\ast}(\theta)=\varphi^{*}(m_{r},D_{r})\chi_{r}(\theta)$
where, focusing on the case with $J=2$ for exposition, \footnote{See our Online Appendix for details. The derivation uses Lemma \ref{lemma:moment}
and the special properties of matrices $\Omega(\theta)$, $I-\lambda W$
and $W$ described in Appendix \ref{subsec:matrix properties}. In
the Online Appendix we also give an explicit expression for the variance
covariance matrix of $\chi_{r}(\delta).$}
\begin{align}
\varphi^{*}\left(m_{r},D_{r}\right) & =E[\frac{\partial}{\partial\theta^{\prime}}\chi_{r}(\theta_{0})|m_{r},D_{r}]\left(E[\chi_{r}(\theta_{0})\chi_{r}(\theta_{0})^{\prime}|m_{r},D_{r}]\right)^{-1}\nonumber \\
 & =\left(\begin{array}{cc}
\frac{1}{\left(m_{r}-1+\lambda\right)\sigma_{\epsilon0,D_{r}}^{2}} & -\frac{m_{r}}{(1-\lambda)(\sigma_{\epsilon0,D_{r}}^{2}+m_{r}\sigma_{\alpha0}^{2})}\\
0 & -\frac{m_{r}^{2}}{2(\sigma_{\epsilon0,D_{r}}^{2}+m_{r}\sigma_{\alpha0}^{2})^{2}}\\
-\frac{1\left\{ D_{r}=1\right\} }{2\sigma_{\epsilon0,1}^{4}} & -\frac{1\left\{ D_{r}=1\right\} m_{r}}{2(\sigma_{\epsilon0,1}^{2}+m_{r}\sigma_{\alpha0}^{2})^{2}}\\
-\frac{1\left\{ D_{r}=2\right\} }{2\sigma_{\epsilon0,2}^{4}} & -\frac{1\left\{ D_{r}=2\right\} m_{r}}{2(\sigma_{\epsilon0,2}^{2}+m_{r}\sigma_{\alpha0}^{2})^{2}}
\end{array}\right).\label{eq:varphi-1}
\end{align}
Clearly, it follows that $E[\chi_{r}^{\ast}(\theta_{0})]=0$ by iterated
expectations. We note that the moment condition in (\ref{eq:Wald_L-O-M})
underlying the CV estimator is based on a linear transformation of
$\varphi^{*}\left(m_{r},D_{r}\right).$ Furthermore, as we shall see
in the next section, under the additional assumption that $\alpha$
and $\epsilon$ follow a Gaussian distribution, the score function
of the log likelihood for group $r$ is exactly the negative of $\chi_{r}^{\ast}(\theta)$,
that is 
\[
\partial lnL_{r}(\theta_{0})/\partial\theta=-\chi_{r}^{\ast}(\theta_{0}),
\]
where $\ln L_{r}(\theta)$ denotes the log likelihood function for
group $r$ conditional on $\left(m_{1},...,m_{R},D_{1},...,D_{R}\right)$.
From these observations we see that the matrices $\varphi^{\ast}\left(m_{r},D_{r}\right)$
can be viewed to provide the optimal weighting for the basic moment
functions $\chi_{r}(\delta)$; compare also the corresponding discussion
for the general model for more details.

The result that $\partial\ln L_{r}(\theta_{0})/\partial\theta=-\chi_{r}^{\ast}(\theta_{0})$
for the score function under Gaussianity establishes the asymptotic
efficiency of the GMM estimator based on $E\left[\chi_{r}(\theta)|m_{r},D_{r}\right]=0$
under the assumption of Gaussian distributions for the unobservables.
When the unobservables are not Gaussian then the GMM estimator has
the interpretation of a quasi maximum likelihood estimator (QMLE).
Similarly, in \citet{lee_identification_2007} the score function
of the conditional maximum likelihood estimator (CMLE) for group $r$
is the optimal moment function corresponding to $E\left[\chi_{r}^{w}(\theta)|m_{r}\right]=0$
under the assumption of homoscedastic and normally distributed errors
$\epsilon_{ir}$. While the CMLE of \citet{lee_identification_2007}
is not efficient under the assumptions we postulate in this paper,
it shares robustness properties of within group panel estimators in
cases where the group effects are possibly correlated with covariates
in the model. Under those circumstances, random effects quasi maximum
likelihood estimators are generally not expected to be consistent.

Our discussion so far highlights variance as the source of identification,
with variation in either size $m_{r}$ or variance of the idiosyncratic
error terms $\sigma_{\epsilon,D_{r}}^{2}$ across groups as conditions.
We show that variation in group size and error term variance is a
source of identification in the QMLE, CVE and CMLE. In all, the CMLE
utilizes how the within-group variance changes with $\lambda$ and
size when error terms are homoscedastic, while the CVE exploits the
relationship between the within-group variance and between-group variance
in relation to $\lambda$ and size when there is either variation
in group size or heteroscedasticity across groups. Our QMLE uses both
pieces of information. All three estimators remain valid without covariates,
and may achieve identification as long as there are at least two different
group sizes in the limit in the case of homoscedasticity. This complements
other results in the literature. For example, Proposition 4 in \citet{bramoulle_identification_2009}
states that in the setting of \citet{lee_identification_2007}, $\lambda$
is identified by instrumenting $(I-W)WY$ with $(I-W)W^{2}Z$, $(I-W)W^{3}Z$,
etc., in line with the spatial literature on the estimation of Cliff-Ord
type models. Their result is due to the fact that they only exploit
restrictions for the conditional mean of $\epsilon$. In \citet{graham_identifying_2008}
as well as in this paper additional constraints on the distribution
of $\alpha$ and $\epsilon$ are imposed and shown to be useful in
the identification of peer effects. Under these conditions including
$Z$ offers additional sources of variation, but identification is
possible with or without it.

Adding covariates is critically important in empirical applications.
Consider adding the covariate matrix $Z$. This leads to two additional
moment conditions $E\left[\ddot{Z}_{r}^{\prime}\ddot{U}_{r}|m_{r},D_{r}\right]=0$
and $E\left[\bar{z}_{r}^{\prime}\bar{u}_{r}|m_{r},D_{r}\right]=0$,
where $\bar{z}_{r}=\iota_{m_{r}}^{\prime}Z_{r}/m_{r}$ is the group
mean of $Z_{r}$ and $\ddot{Z}_{r}=Z_{r}-\iota_{m_{r}}\bar{z}_{r}$
is the deviation from group mean. Moreover, $\ddot{Y}_{r}$ and $\bar{y}_{r}$
now need to be replaced by $\ddot{Y}_{r}-\frac{m_{r}-1}{m_{r}-1+\lambda}\ddot{Z}_{r}\beta$
and $\bar{y}_{r}-\frac{\bar{z}_{r}\beta}{1-\lambda}$ respectively.
The score function of the QMLE then is the same as the moment conditions
of the best GMM corresponding to these two moment functions in addition
to the moments $E\left[\chi_{r}(\theta)|m_{r},D_{r}\right]=0$. In
the same way, in the presence of covariates and assuming homoscedasticity
of $\epsilon$, Lee's CMLE estimator is based on $E\left[\ddot{Z}_{r}^{\prime}\ddot{U}_{r}\right]=0$
in addition to $E\left[\chi_{r}^{w}(\theta)|m_{r}\right]=0$ and the
relative efficiency considerations discussed in this section continue
to apply to the situation with covariates.

\section{General Model\label{sec:General-Model}}

In this section we generalize the model to allow for individual characteristics,
average individual characteristics of peers and group level covariates.
We assume that we have access to observations on $R$ groups belonging
to $J$ categories, where $1\leqslant J<\infty$ is fixed. We consider
asymptotics where the number of groups $R$ tends to infinity and
where the number of group sizes is finite. For the asymptotic identification
of $\lambda_{0}$ and $\sigma_{\alpha0}^{2}$ this setup assumes that
in the limit we observe infinitely many groups for at least two group
sizes or two categories, echoing the requirement of variation in either
group sizes or categories for identification in Section \ref{sec:graham}.
In designs that allow for heteroscedasticity, we also need infinitely
many groups for each category $j\in\left\{ 1,...,J\right\} $ to identify
the remaining variance parameters $\sigma_{\epsilon0,j}^{2}$. Let
$r=1,...,R$ denote the group index, let $D_{r}$ denote the category
of group $r$, and let $m_{r}$ denote the size of group $r$. The
total sample size is then given by $N=\sum_{r=1}^{R}m_{r}$. Suppose
further that interactions occur within each group, but not across
groups, and that peer effects work through the mean outcome and mean
characteristics of peers in the same group. The linear-in-means peer
effects model that includes endogenous as well as exogenous peer effects
then is given by 
\begin{equation}
y_{ir}=\beta_{1}+\lambda\bar{y}_{(-i)r}+x_{1,ir}\beta_{2}+\bar{x}_{2,(-i)r}\beta_{3}+x_{3,r}\beta_{4}+\alpha_{r}+\epsilon_{ir},\label{eq:main_orig}
\end{equation}
where $y_{ir}$ is the outcome variable of individual $i$ in group
$r$, $\bar{y}_{(-i)r}=\frac{1}{m_{r}-1}\sum_{j\neq i}^{m_{r}}y_{jr}$
is the average outcome of $i$'s peers, $x_{1,ir}$ and $x_{2,ir}$
are both row vectors of predetermined characteristics of individual
$i$ in group $r$, $\bar{x}_{2,(-i)r}=\frac{1}{m_{r}-1}\sum_{j\neq i}^{m_{r}}x_{2,jr}$
is a vector of average characteristics of $i$'s peers, $x_{3,r}$
is a vector of observed group characteristics. The variables in $x_{1,ir}$
and $x_{2,ir}$ can be non-overlapping, partially overlapping or totally
overlapping. The error term consists of two components, the group
effect $\alpha_{r}$ and the disturbance term $\epsilon_{ir}$. We
treat $x_{1,ir}$, $x_{2,ir}$, $x_{3,r}$, $D_{r}$ and $m_{r}$
as non-stochastic, while noting that at the expense of more complex
notation we could also think of the analysis as being conditional
on these variables. In this model, peer effects work through the mean
peer outcome $\bar{y}_{(-i)r}$ and mean peer characteristics $\bar{x}_{2,(-i)r}$.
The two terms are also known as the leave-out-mean of $y$ and $x_{2}$,
as they are means of the group leaving out oneself. In Manski's terminology,
$\lambda\bar{y}_{(-i)r}$ in (\ref{eq:main_orig}) reflects endogenous
peer effects, and $\bar{x}_{2,(-i)r}\beta_{3}$ is the exogenous peer
effect, also referred to as contextual peer effects. The covariates
$\bar{x}_{2,(-i)r}$ and $x_{3,r}$ contain group level information
and can be interpreted as parametrizations of group level fixed effects
in the spirit of Mundlak (1978) and Chamberlain (1980). For example,
$x_{3,r}$ can contain full group averages of individual characteristics
or be composed of other characteristics that only vary at the group
level. The CMLE, as in the conventional panel case, cannot account
for this group level information. This can be a limitation in cases
where the effects of group level characteristics are of independent
interest in the analysis. An example is the effects of teacher education
and training on class test scores. 

Let $z_{ir}=(1,x_{1,ir},\bar{x}_{2,(-i)r},x_{3,r})$ be the row vector
of all exogenous variables, let $\beta=(\beta_{1},\beta_{2}^{\prime},\beta_{3}^{\prime},\beta_{4}^{\prime})^{\prime}$
be the corresponding coefficients vector, and let $k_{Z}$ denote
the number of columns in $z_{ir}$. A compact form of model (\ref{eq:main_orig})
is 
\begin{equation}
y_{ir}=\lambda\bar{y}_{(-i)r}+z_{ir}\beta+\alpha_{r}+\epsilon_{ir}.\label{eq:main}
\end{equation}
The model can be further written as a Cliff-Ord type spatial model.
To see this let $I_{m}$ denote the $m-$dimensional identity matrix,
let $\iota_{m}$ denote the $m-$dimensional column vector of ones,
and define the weight matrix $W_{m_{r}}$ for group $r$ as $W_{m_{r}}=\frac{1}{m_{r}-1}(\iota_{m_{r}}\iota_{m_{r}}^{\prime}-I_{m_{r}})$.
The off-diagonal elements of this matrix are all equal to $\frac{1}{m_{r}-1}$
and diagonal elements are 0. Let $Y_{r}=(y_{1r},...,y_{m_{r}r})^{\prime}$,
$Z_{r}=(z_{1r}^{\prime},...,z_{m_{r}r}^{\prime})^{\prime}$, $\epsilon_{r}=(\epsilon_{1r},...,\epsilon_{m_{r}r})^{\prime}$,
then the model for group $r$ can be expressed in matrix form as 
\begin{equation}
Y_{r}=\lambda W_{m_{r}}Y_{r}+Z_{r}\beta+U_{r},\label{eq:main_matrix}
\end{equation}
where $U_{r}=\alpha_{r}\iota_{m_{r}}+\epsilon_{r}.$ Let $Y=[Y_{1}^{\prime},Y_{2}^{\prime},...,Y_{R}^{\prime}]^{\prime}$,
$Z=[Z_{1}^{\prime},Z_{2}^{\prime},...,Z_{R}^{\prime}]^{\prime}$,
$U=[U_{1}^{\prime},U_{2}^{\prime},...,U_{R}^{\prime}]^{\prime}$,
and $W=diag_{r=1}^{R}\{W_{m_{r}}\}$ such that the model for the whole
sample is given by 
\begin{equation}
Y=\lambda WY+Z\beta+U.\label{eq:main_matrix_2}
\end{equation}
In the spatial literature $W$ is referred to as a spatial weight
matrix and $WY$ as a spatial lag. In analyzing the model in (\ref{eq:main_matrix_2})
we maintain the random effects specification detailed in Assumptions
1 and 2 of Section \ref{sec:graham}, which imply that $\alpha_{r}\sim(0,\sigma_{\alpha}^{2})$
and $\epsilon_{ir}\sim(0,\sigma_{\epsilon,D_{r}}^{2})$, where $D_{r}\in\{1,...,J\}$
with $J\geqslant1$ fixed and finite. The specification allows for
heteroscedasticity at the group level as long as there are only a
finite number of different parameters. For example, we could allow
for $\sigma_{\epsilon,D_{r}}^{2}$ to be different for small and large
groups, or more generally for all groups of a certain size $m_{r}$.
On the other hand we do not cover the case where $\sigma_{\epsilon,D_{r}}^{2}$
differs for each individual group $r,$ as this would lead to an infinite
dimensional parameter space.

The parameters of interest are $\lambda,\sigma_{\alpha}^{2},\sigma_{\epsilon,1}^{2},...,\sigma_{\epsilon,J}^{2}$
and $\beta$. Their respective true values are $\lambda_{0},\sigma_{\alpha0}^{2},\sigma_{\epsilon0,1}^{2},...,\sigma_{\epsilon0,J}^{2}$
and $\beta_{0}$. In analyzing the model it will be convenient to
concentrate the log-likelihood function with respect to $\beta$ for
given values of $\theta=(\lambda,\sigma_{\alpha}^{2},\sigma_{\epsilon1}^{2},...,\sigma_{\epsilon J}^{2})^{\prime}$.
Let $\Theta$ denote the parameter space for $\theta$ , and let $\delta=(\theta^{\prime},\beta^{\prime})^{\prime}$
denote the vector of all parameters. Under Assumptions 1 and 2 the
expression for the variance covariance matrix $\Omega_{0}$ of $U$
is
\[
\Omega_{0}=\Omega(\theta_{0})=\diag_{r=1}^{R}\Omega_{r}(\theta_{0})=\diag_{r=1}^{R}\{\sigma_{\epsilon0,D_{r}}^{2}I_{m_{r}}+\sigma_{\alpha0}^{2}\iota_{m_{r}}\iota_{m_{r}}^{\prime}\}.
\]

To define the quasi-maximum likelihood estimator (QMLE) for the peer
effects model in (\ref{eq:main}) note that solving $Y$ from (\ref{eq:main_matrix_2})
yields the reduced from:

\begin{equation}
Y=(I-\lambda W)^{-1}Z\beta+(I-\lambda W)^{-1}U.\label{eq:y_sol_2}
\end{equation}
If $\alpha_{r}$ and $\epsilon_{ir}$ follow normal distributions,
\begin{equation}
Y\sim N((I-\lambda W)^{-1}Z\beta,(I-\lambda W)^{-1}\Omega(\theta)(I-\lambda W^{\prime})^{-1}).\label{eq:y_dis}
\end{equation}
The corresponding log likelihood function is 
\begin{align}
\ln L_{N}(\theta,\beta) & =-\frac{N}{2}\ln(2\pi)+\frac{1}{2}\ln|(I-\lambda W)^{2}\Omega(\theta)^{-1}|\nonumber \\
 & -\frac{1}{2}(Y-\lambda WY-Z\beta)^{\prime}\Omega(\theta)^{-1}(Y-\lambda WY-Z\beta).\label{eq:logL}
\end{align}
and the corresponding QMLE is given by 
\begin{equation}
\hat{\delta}_{N}=(\hat{\theta}_{N}^{\prime},\hat{\beta}_{N}^{\prime})^{\prime}=\textrm{argmax}_{\theta,\beta}\ln L_{N}(\theta,\beta).\label{eq:Def_QMLE}
\end{equation}
It is convenient to concentrate out $\beta$ and to obtain the QMLE
for $\theta$ first. The first order condition for $\beta$ is 
\begin{equation}
\frac{\partial\ln L_{N}(\theta,\beta)}{\partial\beta}=(Y-\lambda WY-Z\beta)^{\prime}\Omega(\theta)^{-1}Z=0,\label{eq:focdelta}
\end{equation}
which leads to 
\begin{equation}
\hat{\beta}_{N}(\theta)=(Z^{\prime}\Omega(\theta)^{-1}Z)^{-1}Z^{\prime}\Omega(\theta)^{-1}(I-\lambda W)Y.\label{eq:gammahat}
\end{equation}
Plugging $\hat{\beta}_{N}(\theta)$ back into (\ref{eq:logL}) yields
the following concentrated log likelihood function,
\begin{align}
Q_{N}(\theta) & =\frac{1}{N}\ln L_{N}(\theta,\hat{\beta}_{N}(\theta))\nonumber \\
 & =-\frac{\ln(2\pi)}{2}+\frac{1}{2N}\ln|(I-\lambda W)^{2}\Omega(\theta)^{-1}|-\frac{1}{2N}Y^{\prime}(I-\lambda W)^{\prime}M_{Z}(\theta)(I-\lambda W)Y,\label{eq:QN}
\end{align}
where 
\begin{equation}
M_{Z}(\theta)=\Omega(\theta)^{-1}-\Omega(\theta)^{-1}Z(Z^{\prime}\Omega(\theta)^{-1}Z)^{-1}Z^{\prime}\Omega(\theta)^{-1}.\label{eq:Mz}
\end{equation}
Then the QMLE for $\theta,$ $\hat{\theta}_{N}=(\hat{\lambda}_{N},\hat{\sigma}_{\alpha,N}^{2},\hat{\sigma}_{\epsilon1,N}^{2},...,\hat{\sigma}_{\epsilon J,N}^{2})^{\prime}$
is given by
\begin{equation}
\hat{\theta}_{N}=\textrm{argmax}_{\theta}Q_{N}(\theta).\label{eq:varthetahat}
\end{equation}
Plugging $\hat{\theta}_{N}$ back into \eqref{eq:gammahat}, the QMLE
for $\beta$ is

\begin{equation}
\hat{\beta}_{N}=\hat{\beta}_{N}(\hat{\theta}_{N})=(Z^{\prime}\Omega(\hat{\theta}_{N})^{-1}Z)^{-1}Z^{\prime}\Omega(\hat{\theta}_{N})^{-1}(I-\hat{\lambda}_{N}W)Y.\label{eq:deltahathat1}
\end{equation}

A formal result regarding the asymptotic identification of the model
parameters is given in the next section. We next provide some intuition
for that result, by extending our earlier discussion of identification
for the canonical model without covariates to our model (\ref{eq:main_matrix_2})
with covariates. Let $\left\Vert .\right\Vert $ be the Euclidean
norm on $\mathbb{R^{\text{\ensuremath{k}}}}.$ Using the relationships
$\ddot{U}_{r}=\frac{\left(m_{r}-1+\lambda_{0}\right)}{m_{r}-1}\ddot{Y}_{r}-\ddot{Z}_{r}\beta_{0}$
and $\bar{u}_{r}=(1-\lambda_{0})\bar{y}_{r}-\bar{z}_{r}\beta_{0}$,
the moment functions related to the full model can be written as follows
\[
\chi_{r}(\theta)=\left[\begin{array}{c}
\chi_{r}^{w}(\theta)\\
\chi_{r}^{b}(\theta)\\
\chi_{r}^{zw}\left(\delta\right)\\
\chi_{r}^{zb}\left(\delta\right)
\end{array}\right]=\left[\begin{array}{c}
\left\Vert \frac{m_{r}-1+\lambda}{m_{r}-1}\ddot{Y}_{r}-\ddot{Z}_{r}\beta\right\Vert ^{2}-(m_{r}-1)\sigma_{\epsilon,D_{r}}^{2}\\{}
[(1-\lambda)\bar{y}_{r}-\bar{z}_{r}\beta]^{2}-\sigma_{\alpha}^{2}-\frac{\sigma_{\epsilon,D_{r}}^{2}}{m_{r}}\\
\ddot{Z}_{r}^{\prime}\left(\frac{m_{r}-1+\lambda}{m_{r}-1}\ddot{Y}_{r}-\ddot{Z}_{r}\beta\right)\\
\bar{z}_{r}^{\prime}\left((1-\lambda)\bar{y}_{r}-\bar{z}_{r}\beta\right)
\end{array}\right]
\]
where $\chi_{r}^{w}(\delta)$ and $\chi_{r}^{b}(\delta)$ summarize
the restrictions on the unobservables, and are natural extensions
of the moment conditions considered before in \eqref{eq:mwb} for
the model without covariates. The additional moment restrictions $\chi_{r}^{zw}\left(\delta\right)$
and $\chi_{r}^{zb}\left(\delta\right)$ relate to the exogeneity of
$Z_{r}$ relative to $\epsilon_{r}$ and $\alpha_{r}$. A formal asymptotic
identification result will be given in the next section. Intuitively,
for given $\lambda$ the last two moment conditions identify $\beta$,
while the first two identify $\lambda,\sigma_{\alpha}^{2},\sigma_{\varepsilon,j}^{2},j=1,...,J$
in an analogous manner as described in the discussion of Lemma \ref{lem:ID_1}
for the model without covariates.

As for the model without covariates there is a representation of the
score of the log-likelihood in terms of the fundamental moment conditions.
To describe the relationship between moments and the score we define
the matrix

\begin{align}
\varphi(m_{r},D_{r}) & =\left(\begin{array}{cccc}
\frac{1}{(m_{r}-1+\lambda_{0})\sigma_{\epsilon0,D_{r}}^{2}} & -\frac{m_{r}}{(1-\lambda_{0})(\sigma_{\epsilon0,D_{r}}^{2}+m_{r}\sigma_{\alpha0}^{2})} & \frac{1}{(m_{r}-1+\lambda_{0})\sigma_{\epsilon0,D_{r}}^{2}}\beta_{0}^{\prime} & -\frac{m_{r}}{(1-\lambda_{0})(\sigma_{\epsilon0,D_{r}}^{2}+m_{r}\sigma_{\alpha0}^{2})}\beta_{0}^{\prime}\\
0 & -\frac{m_{r}^{2}}{2(\sigma_{\epsilon0,D_{r}}^{2}+m_{r}\sigma_{\alpha0}^{2})^{2}} & 0 & 0\\
-\frac{1(D_{r}=1)}{2\sigma_{\epsilon0,1}^{4}} & -\frac{m_{r}1(D_{r}=1)}{2(\sigma_{\epsilon0,1}^{2}+m_{r}\sigma_{\alpha0}^{2})^{2}} & 0 & 0\\
\vdots & \vdots & \vdots & \vdots\\
-\frac{1(D_{r}=J)}{2\sigma_{\epsilon0,J}^{4}} & -\frac{m_{r}1(D_{r}=J)}{2(\sigma_{\epsilon0,J}^{2}+m_{r}\sigma_{\alpha0}^{2})^{2}} & 0 & 0\\
0 & 0 & -\frac{1}{\sigma_{\epsilon0,D_{r}}^{2}}I_{k_{Z}} & -\frac{m_{r}}{\sigma_{\epsilon0,D_{r}}^{2}+m_{r}\sigma_{\alpha0}^{2}}I_{k_{Z}}
\end{array}\right).\label{eq:varphim}
\end{align}
Furthermore observe that the log-likelihood function can be written
as $\ln L_{N}(\delta)=-\frac{N}{2}\ln(2\pi)+\sum_{r=1}^{R}\ln L_{r}(\delta)$
where 
\begin{align*}
\ln L_{r}(\delta) & =\frac{1}{2}\ln|(I_{m_{r}}-\lambda W_{m_{r}})^{2}\Omega_{r}(\theta)^{-1}|\\
 & -\frac{1}{2}(Y_{r}-\lambda W_{m_{r}}Y_{r}-Z_{r}\beta)^{\prime}\Omega_{r}(\theta)^{-1}(Y_{r}-\lambda W_{m_{r}}Y-Z_{r}\beta)
\end{align*}
is the log-likelihood function for group $r$. Then it can be shown
that\footnote{See our Online Appendix for details. The derivation uses Lemma \ref{lemma:moment}
and the special properties of matrices $\Omega(\theta)$, $I-\lambda W$
and $W$ described in Appendix \ref{subsec:matrix properties}. In
the Online Appendix we also give an explicit expression for the variance
covariance matrix of $\chi_{r}(\delta).$}
\begin{eqnarray*}
\frac{\partial\ln L_{r}(\delta_{0})}{\partial\delta} & = & -\chi_{r}^{\ast}(\delta_{0})=-\varphi(m_{r},D_{r})\chi_{r}\left(\delta_{0}\right).
\end{eqnarray*}
 As is well known, the score of the log-likelihood function, $S(\delta)=-\sum_{r=1}^{R}\frac{\partial\ln L_{r}(\delta)}{\partial\delta}$
can be interpreted as a moment function corresponding to the moments
$E\left[S(\delta_{0})\right]=-\sum_{r=1}^{R}E\left[\frac{\partial\ln L_{r}(\delta_{0})}{\partial\delta}\right]=0$.
Furthermore, under a Gaussian assumption the score is an optimal moment
function.\footnote{Observe that 
\begin{align*}
E\left[\frac{\partial S(\delta_{0})}{\partial\delta'}\right]\left[Var\left(S(\delta_{0})\right)\right]^{-1}S(\delta_{0}) & =\left(-\sum_{r=1}^{R}E\left[\frac{\partial^{2}lnL_{r}(\delta_{0})}{\partial\delta\partial\delta'}\right]\right)\left(\sum_{r=1}^{R}E\left[\frac{\partial lnL_{r}(\delta_{0})}{\partial\delta}\frac{\partial lnL_{r}(\delta_{0})}{\partial\delta'}\right]\right)^{-1}S(\delta_{0})\\
 & =S(\delta_{0}).
\end{align*}
in light of the information matrix equality.} From this we see that the matrices $\varphi(m_{r},D_{r})$ can be
viewed to provide the optimal weighting for the basic moment functions
$\chi_{r}(\delta)$. Under Gaussian assumptions the optimal GMM estimator
coincides with the maximum likelihood estimator and is asymptotically
efficient under the stated assumptions.

\section{Theoretical Results\label{sec:Theoretical-Results}}

We next state our assumptions for the general model. We maintain Assumptions
\ref{assume:epsilon}-\ref{assume:lambda} on $\epsilon$, $\alpha$
and $\lambda$. In the following we add assumptions regarding the
exogenous variables, and the sizes and relative magnitudes of groups
in the sample. Let $\mathcal{I}_{m,j}\subset\{1,...,R\}$ be the index
set of all groups in category $j$ with size equal to $m$. Thus if
$r\in\mathcal{I}_{m,j}$, then $D_{r}=j$ and $m_{r}=m$. Let $R_{m,j}$
be the cardinality of $\mathcal{I}_{m,j}$, in other words $R_{m,j}$
is the number of groups in category $j$ with size equal to $m$,
and let $R_{j}$ be the number of groups in category $j$, that is
$R_{j}=\sum_{r=1}^{R}1(D_{r}=j)=\sum_{m=2}^{\bar{M}}R_{m,j}$, where
the upper bound\textcolor{red}{{} }\textcolor{black}{$\bar{M}$} on
the group size is specified in the next assumption below. Furthermore
let $\omega_{m,j}=R_{m,j}/R$ denote the share of groups in category
$j$ with size equal to $m$, and let $\omega_{j}=R_{j}/R=\sum_{m=2}^{\bar{M}}\omega_{m,j}$
be the share of groups in category $j$. Below we maintain the following
assumption regarding the group sizes and their relative magnitudes.
\begin{assumption}
\textcolor{black}{\label{assume:n}(a) The sample size $N$ goes to
infinity; (b) The group size is bounded in the sense that there exists
some positive constant $\bar{M}$ such that $2\leqslant m_{r}\leqslant\bar{M}<\infty$
}\textup{for}\textcolor{black}{{} $r=1,2,...,R$; (c) The limit $\omega_{m,j}^{\ast}=\lim_{N\rightarrow\infty}\omega_{m,j}$
exists and $\omega_{m,j}^{\ast}<1$ for all $2\leqslant m\leqslant\bar{M}$
and $j$, and $\omega_{j}^{\ast}=\lim_{N\rightarrow\infty}\omega_{j}=\sum_{m=2}^{\bar{M}}\omega_{m,j}^{*}>0$
for all $j$.}
\end{assumption}
The restriction that the minimal group size is 2 rules out singleton
groups. A member of such a group has no peers. Assumption \ref{assume:n}(b)
imposes a fixed upper bound on group size. In many applications this
is not a serious constraint. The assumption is more restrictive than
Lee (2007) who allows for group size to grow with sample size. It
is worth pointing out that increasing group sizes generally reduce
the convergence rates for estimators of peer effects parameters, and
as demonstrated by \citet{kelejian_2sls_2002} in some cases lead
to inconsistency of these estimators. 

Assumption \ref{assume:n}(c) states that asymptotically, no single
type-group size combination can dominate the sample by requiring that
\textcolor{black}{$\omega_{m,j}^{\ast}<1$ for all $2\leqslant m\leqslant\bar{M}$
and $j$}. In addition, all types $j$ occur in the sample in an asymptotically
non-negligible way because \textcolor{black}{$\omega_{j}^{\ast}>0$
for all $j$. On the other hand, we do allow that for certain combinations
of $j$ and $m$ the limit }$\omega_{m,j}^{\ast}$ is zero, allowing
for some group sizes of type $j$ to occur infrequently or not at
all in the sample.

Observe that $N=\sum_{r=1}^{R}m_{r}=\sum_{j=1}^{J}\sum_{m=2}^{\bar{M}}mR_{m,j}$.
Since group size is bounded, the number of groups $R$ goes to infinity
as $N$ goes to infinity. Since $\sum_{j=1}^{J}\sum_{m=2}^{\bar{M}}R_{m,j}=R$,
we have $\sum_{j=1}^{J}\sum_{m=2}^{\bar{M}}\omega_{m,j}=\sum_{j=1}^{J}\omega_{j}=1$
and thus $\sum_{j=1}^{J}\sum_{m=2}^{\bar{M}}\omega_{m,j}^{\ast}=\sum_{j=1}^{J}\omega_{j}^{\ast}=1$.
Since $\omega_{j}^{\ast}>0$ \textcolor{black}{by As}sumption \ref{assume:n}(c)
it follows that also $R_{j}$ goes to infinity, which is needed to
facilitate the consistent estimation of \textcolor{black}{$\sigma_{\epsilon,j}^{2}$}.
Assumption \ref{assume:n}(c) implies that the limit of the average
group size is given by
\begin{align}
m^{\ast} & =\lim_{N\rightarrow\infty}\frac{N}{R}=\lim_{N\rightarrow\infty}\sum_{j=1}^{J}\sum_{m=2}^{\bar{M}}\frac{R_{m,j}}{R}m=\sum_{j=1}^{J}\sum_{m=2}^{\bar{M}}\omega_{m,j}^{\ast}m.\label{eq:nstar-1}
\end{align}
Clearly $2\leqslant m^{\ast}\leqslant\bar{M}$, since $2\leqslant m_{r}\leqslant\bar{M}$.

Observe that in light of Assumptions \ref{assume:epsilon}, \ref{assume:alpha},
and \ref{assume:lambda} the parameter space $\Theta$ for $\theta=\left(\lambda,\sigma_{\alpha}^{2}\sigma_{\epsilon,1}^{2},...,\sigma_{\epsilon,J}^{2}\right)^{\prime}$
is a compact subset of the Euclidean space $\mathbb{R}^{2+J}$. Observe
further that 
\begin{align}
I_{m_{r}}-\lambda W_{m_{r}} & =(1+\frac{\lambda}{m_{r}-1})I_{m_{r}}^{\ast}+(1-\lambda)J_{m_{r}}^{\ast},
\end{align}
where $I_{m_{r}}^{\ast}=I_{m_{r}}-\iota_{m_{r}}\iota_{m_{r}}^{\prime}/m_{r}$
and $J_{m_{r}}^{\ast}=\iota_{m_{r}}\iota_{m_{r}}^{\prime}/m_{r}$
are symmetric, idempotent, orthogonal, and sum to the identity matrix.
Furthermore from the results in Appendix \ref{subsec:matrix properties}
we have $|I_{m_{r}}-\lambda W_{m_{r}}|=[1+\lambda/(m_{r}-1)]^{m_{r}-1}(1-\lambda)$.\footnote{In Appendix \ref{subsec:matrix properties} we review additional properties
of matrices of the form $pI_{m}^{\ast}+sJ_{m}^{\ast}$ , which will
be used repeatedly in this paper. In particular, their multiplication
is commutative. The products of such matrices are also of the form
of $pI_{m}^{\ast}+sJ_{m}^{\ast}$, and $|pI_{m}^{\ast}+sJ_{m}^{\ast}|=p^{m-1}s$,
$(pI_{m}^{\ast}+sJ_{m}^{\ast})^{-1}=\frac{1}{p}I_{m}^{\ast}+\frac{1}{s}J_{m}^{\ast}$.} Thus the matrix $I_{m_{r}}-\lambda W_{m_{r}}$ is nonsingular if
$1+\lambda/(m_{r}-1)\neq0$ and $1-\lambda\neq0$. Assumption \ref{assume:lambda}
ensures the non-singularity of $I_{m_{r}}-\lambda W_{m_{r}}$, and
hence the non-singularity of $I-\lambda W=\diag_{r=1}^{R}\{I_{m_{r}}-\lambda W_{m_{r}}\}$,
since for $m_{r}\geqslant2$ and $\lambda<1$ we have $1+\lambda/(m_{r}-1)>0$
and $1-\lambda>0$.

Let $\bar{z}_{r}=\frac{1}{m_{r}}\iota_{m_{r}}^{\prime}Z_{r}$ be the
row vector of column means of $Z_{r}$, and let $\ddot{Z}_{r}=Z_{r}-\iota_{m_{r}}\bar{z}_{r}$
be the deviations from the column means. Then $Z_{r}^{\prime}I_{m_{r}}^{\ast}Z_{r}=\ddot{Z}_{r}^{\prime}\ddot{Z}_{r}$,
$Z_{r}^{\prime}J_{m_{r}}^{\ast}Z_{r}=m_{r}\bar{z}_{r}^{\prime}\bar{z}_{r}$.
\begin{assumption}
\label{assume:z}(a) The $N\times k_{Z}$ matrix $Z$ is non-stochastic,
with $rank(Z)=k_{Z}>0$ for $N$ sufficiently large. The elements
of $Z$ are uniformly bounded in absolute value.

(b)For $2\leqslant m\leqslant\bar{M}$, and $1\leqslant j\leqslant J$
the following limits exist:
\[
\lim_{N\rightarrow\infty}N^{-1}\sum_{r\in\mathcal{I}_{m,j}}\ddot{Z}_{r}^{\prime}\ddot{Z}_{r}=\ddot{\varkappa}_{m,j},
\]
\[
\lim_{N\rightarrow\infty}N^{-1}\sum_{r\in\mathcal{I}_{m,j}}m\bar{z}_{r}^{\prime}\bar{z}_{r}=\bar{\varkappa}_{m,j},
\]
\[
\lim_{N\rightarrow\infty}N^{-1}\sum_{r\in\mathcal{I}_{m,j}}\bar{z}_{r}=\bar{z}_{m,j}.
\]

(c) For at least one pair of $(m,j)$ such that $\omega_{m,j}^{\ast}>0$,
and N sufficiently large, the smallest eigenvalues of $N^{-1}\sum_{r\in\mathcal{I}_{m,j}}Z_{r}^{\prime}Z_{r}=N^{-1}\sum_{r\in\mathcal{I}_{m,j}}\ddot{Z}_{r}^{\prime}\ddot{Z}_{r}+N^{-1}\sum_{r\in\mathcal{I}_{m,j}}m\bar{z}_{r}^{\prime}\bar{z}_{r}$
are bounded away from zero, uniformly in N, by some finite constant
$\underline{\xi}_{Z}>0$.
\end{assumption}
Suppose we have some $N\times N$ matrix $A_{N}(\theta)=diag_{r=1}^{R}\{p(m_{r},D_{r},\theta)I_{m_{r}}^{\ast}+s(m_{r},D_{r},\theta)J_{m_{r}}^{\ast}\}$,
where $p(m_{r},D_{r},\theta)$ and $s(m_{r},D_{r},\theta)$ are positive,
uniformly continuous and bounded on $\Theta$. An example of an expression
of this form is $\Omega(\theta)^{-1}$ which is obtained in closed
form in Equation \eqref{AppCon1} in Appendix \ref{subsec:matrix properties}.
Then under Assumption \ref{assume:z}(b), the limiting matrix of $N^{-1}Z^{\prime}A_{N}(\theta)Z$
always exists, is continuous in $\theta$ and takes the form 
\[
\lim_{N\rightarrow\infty}\frac{1}{N}Z^{\prime}A_{N}(\theta)Z=\sum_{j=1}^{J}\sum_{m=2}^{\bar{M}}[p(m,j,\theta)\ddot{\varkappa}_{m,j}+s(m,j,\theta)\bar{\varkappa}_{m,j}].
\]
Furthermore, $N^{-1}Z^{\prime}A_{N}(\theta)Z$ converges to its limiting
matrix uniformly on $\Theta$. With $p(m_{r},D_{r},\theta)>0$ and
$s(m_{r},D_{r},\theta)>0$, Assumption \ref{assume:z}(a) ensures
that $N^{-1}Z^{\prime}A_{N}(\theta)Z$ and its limiting matrix are
invertible, with the elements of the inverse matrix uniformly bounded
in absolute value. In the special case when $A_{N}(\theta)$ is the
identity matrix, $lim_{N\rightarrow\infty}\frac{1}{N}Z^{\prime}Z=\sum_{j=1}^{J}\sum_{m=2}^{\bar{M}}[\ddot{\varkappa}_{m,j}+\bar{\varkappa}_{m,j}]$,
which has the smallest eigenvalue bounded above zero by some finite
constant $\text{\ensuremath{\underbar{\ensuremath{\xi}}}}_{Z}>0$.
See Lemma \ref{lem:Aux1} for details and a proof.

As shown by Lemma \ref{lem:ID_1} in Section \ref{sec:graham}, identification
of $\lambda$ and $\sigma_{\alpha}^{2}$ requires variation in the
group size or variance of the error terms. The following assumption
ensures this so that in the limit we have non-negligible samples for
at least two different group sizes or two different categories with
different variances of the idiosyncratic errors $\epsilon_{ir}$.
\begin{assumption}
\label{assume:id}For some sizes $m$ and $m^{\prime}$, and some
categories $j$ and $j^{\prime}$ we have $\omega_{m,j}^{\ast}>0$
and $\omega_{m^{\prime},j^{\prime}}^{\ast}>0$, and either of the
following two scenarios hold,

(a) $m\neq m^{\prime}$, and $\sigma_{\epsilon0,j}^{2}=\sigma_{\epsilon0,j^{\prime}}^{2}$
for some $j,j^{\prime}\in\left\{ 1,...,J\right\} $.

(b) $m=m^{\prime}$, and $\sigma_{\epsilon0,j}^{2}\neq\sigma_{\epsilon0,j^{\prime}}^{2}$
for some $j,j^{\prime}\in\left\{ 1,...,J\right\} $ with $j\neq j^{\prime}.$
\end{assumption}
The conditions in Assumption \ref{assume:id} are the asymptotic analogs
of identification conditions imposed in Lemma \ref{lem:ID_1}. Assumption
\ref{assume:n} by itself is not sufficient for identification because
it only implies that no single pair $\left(m,j\right)$ asymptotically
dominates the sample. Assumption \ref{assume:n} alone does not guarantee
that there is enough variation in the underlying group sizes $m$
or the variances $\sigma_{\epsilon0,j}^{2}$. For example, it is possible
under Assumption \ref{assume:n} that all groups are of the same size
and that all variances $\sigma_{\epsilon0,j}^{2}$ are the same. Assumption
\ref{assume:id} rules out such cases. Assumption \ref{assume:id}(a)
is related to Assumption 6.1 and Footnote 9 of \citet{lee_identification_2007}
which requires group size variation to achieve identification for
the case where group sizes are bounded, the only case we consider.
Assumption \ref{assume:id}(b) has no analog in Lee (2007) because
of his Assumption 1 which imposes homoscedasticity on the errors $\epsilon_{ir}.$
We show that identification is possible purely based on group level
heteroscedasticity even if all group sizes are the same. This insight
also extends the analysis of \citet{graham_identifying_2008} where
types and class sizes are linked.

Below we give results on the consistency and asymptotic normality
of the QMLE $\hat{\delta}_{N}=(\hat{\theta}_{N}^{\prime},\hat{\beta}_{N}^{\prime})^{\prime}$
defined in (\ref{eq:Def_QMLE}).
\begin{thm}
\label{theorem:Consistency}Suppose Assumptions 1-6 hold, then 

(a) The parameter $\delta_{0}$ is asymptotically identified in the
sense that it is the unique maximizer of the criterion $\bar{R}(\theta,\beta)=\lim_{N\rightarrow\infty}E\left[\frac{1}{N}\textrm{ln}L(\theta,\beta)\right]$.

(b) The QMLE $\widehat{\delta}_{N}$ is consistent, i.e., $\widehat{\delta}_{N}\overset{p}{\rightarrow}\delta_{0}$
as $N\rightarrow\infty$.
\end{thm}
A detailed proof of the theorem is given in Appendices \ref{sec:Proof-of-cons}
and \ref{sec:Proof-of-cons-2}. As can be seen from the proof, the
argumentation that ensures part (a) of the theorem is analogous to
the argumentation used in establishing Lemma 2.1. Here is a sketch
of the proof to provide some intuition. The limiting expected value
of the concentrated log likelihood function $Q_{N}(\theta)$ is 
\[
\bar{Q}^{\ast}(\theta)=C^{\ast}+\frac{1}{2m^{\ast}}\sum_{j=1}^{J}\sum_{m=2}^{\bar{M}}\omega_{m,j}^{\ast}g(m,j,\theta)+Q^{(2)\ast}(\theta),
\]
where $C^{\ast}$ is a constant term, $g(m,j,\theta)=\ln|G(m,j,\theta)|-\tr G(m,j,\theta)$
with 
\[
G(m,j,\theta)=\frac{\sigma_{\epsilon0,j}^{2}}{\sigma_{\epsilon,j}^{2}}\left(\frac{m-1+\lambda}{m-1+\lambda_{0}}\right)^{2}I_{m}^{\ast}+\frac{(\sigma_{\epsilon0,j}^{2}+m\sigma_{\alpha0}^{2})}{(\sigma_{\epsilon,j}^{2}+m\sigma_{\alpha}^{2})}\left(\frac{1-\lambda}{1-\lambda_{0}}\right)^{2}J_{m}^{\ast},
\]
 and $Q^{(2)\ast}(\theta)=\lim_{N\rightarrow\infty}\bar{Q}_{N}^{(2)}(\theta)$
where $\bar{Q}_{N}^{(2)}(\theta)=-\frac{1}{2N}\tilde{\eta}_{Z}(\theta)^{\prime}\tilde{M}_{Z}(\theta)\tilde{\eta}_{Z}(\theta)$
with $\tilde{M}_{Z}(\theta)=I-\Omega(\theta)^{-1/2}Z(Z^{\prime}\Omega(\theta)^{-1}Z)^{-1}Z^{\prime}\Omega(\theta)^{-1/2}$
and 
\[
\tilde{\eta}_{Z}(\theta)=\Omega(\theta)^{-1/2}(I-\lambda W)(I-\lambda_{0}W)^{-1}Z\beta_{0}.
\]
 It is easy to see that $\theta_{0}$ is a global maximizer of $Q^{(2)\ast}(\theta)$,
given that $-\bar{Q}_{N}^{(2)}\left(\theta\right)$ is the quadratic
form of an idempotent and thus positive semi-definite matrix, and
$Q^{(2)\ast}(\theta_{0})=0$. However, this does not ensure that $\theta_{0}$
is a unique global maximizer. Identification thus comes from $\sum_{j=1}^{J}\sum_{m=2}^{\bar{M}}\omega_{m,j}^{\ast}g(m,j,\theta)$.
Note that for any symmetric positive definite $m\times m$ matrix
$A$, $\ln|A|-\tr(A)\leqslant-m$ with equality if and only if $A$
is an identity matrix.\footnote{\label{footnote:max}To see this, note that under the maintained assumptions
the eigenvalues of $A$, say, $\lambda_{i}$, are positive and $\ln\left\vert A\right\vert -\tr\left[A\right]=\sum_{i=1}^{m}\left[\ln(\lambda_{i})-\lambda_{i}\right]$.
The claim is seen to hold by observing that the function $f(x)=\ln(x)-x\leq-1$
for $x\in(0,\infty)$ with a unique maximum at $x=1$, and observing
that $A=I_{m}$ if and only if $\lambda_{i}=1$ for $i=1,\ldots,m$.} For any $m$ and $j$, $g(m,j,\theta)$ is maximized if and only
if $G(m,j,\theta)=I_{m}$, which is equivalent to $E\left[\chi_{r}(\theta)|m_{r}=m,D_{r}=j\right]=0$
with $\chi_{r}(\theta)=(\chi_{r}^{w}(\theta),\chi_{r}^{b}(\theta))$
defined in \eqref{eq:mwb}. It now follows from an asymptotic analogue
of Lemma \ref{lem:ID_1} that in either case (i) or (ii) of Assumption
\ref{assume:id}, $\theta_{0}$ is the only solution to $E\left[\chi_{r}(\theta)|m_{r}=m,D_{r}=j\right]=0$
and $E\left[\chi_{r}(\theta)|m_{r}=m^{\prime},D_{r}=j^{\prime}\right]=0$.
Thus for any $\theta\neq\theta_{0}$,
\[
\min\left(g(m,j,\theta_{0})-g(m,j,\theta),g(m^{\prime},j^{\prime},\theta_{0})-g(m^{\prime},j^{\prime},\theta)\right)>0.
\]
 As a result, $\theta_{0}$ is the unique global maximizer of $\bar{Q}^{\ast}(\theta)$
when one of the two scenarios holds true for some $\omega_{m,j}^{\ast}>0$
and $\omega_{m^{\prime},j^{\prime}}^{\ast}>0$.

To study the asymptotic distribution of the estimator, first note
that under Assumptions \ref{assume:epsilon} and \ref{assume:alpha},
the third and fourth moments of $\epsilon_{ir}$ and $\alpha_{r}$
exist. Let $E\left[\epsilon_{ir}^{3}|D_{r}=j\right]=\mu_{\epsilon0,j}^{(3)}$,
$E\left[\epsilon_{ir}^{4}|D_{r}=j\right]=\mu_{\epsilon0,j}^{(4)}$,
$E\left[\alpha_{r}^{3}\right]=\mu_{\alpha0}^{(3)}$ and $E\left[\alpha_{r}^{4}\right]=\mu_{\alpha0}^{(4)}$.
Also, define $\Gamma_{0}$ and $\Upsilon_{0}$ as
\begin{eqnarray*}
\Gamma_{0} & = & \lim_{N\rightarrow\infty}N^{-1}E\left[-\frac{\partial^{2}lnL_{N}(\delta_{0})}{\partial\delta\partial\delta^{\prime}}\right],\\
\Upsilon_{0} & = & \lim_{N\rightarrow\infty}N^{-1}E\left[\frac{\partial lnL_{N}(\delta_{0})}{\partial\delta}\frac{\partial lnL_{N}(\delta_{0})}{\partial\delta^{\prime}}\right].
\end{eqnarray*}
As shown in Appendix \ref{sec:Proof-of-asydis}, the two limiting
matrices exist. Specific expressions are given in Appendix \ref{sec:Variance-Covariance-Matrix}.
When $\epsilon_{ir}$ and $\alpha_{r}$ both follow normal distributions,
$\Upsilon_{0}=\Gamma_{0}$.

The next lemma shows that $\Gamma_{0}$ is p.d. under the maintained
assumptions. The lemma also provides a sufficient condition on the
moments of $\varepsilon$ under which $\Upsilon_{0}$ is p.d..
\begin{lem}
\label{lem:nd} Suppose Assumptions \ref{assume:epsilon}-\ref{assume:id}
hold, then $\Gamma_{0}$ is positive definite. Under the additional
assumption that $\mu_{\varepsilon0,j}^{(4)}-\sigma_{\varepsilon0,j}^{4}>(\mu_{\epsilon0,j}^{(3)})^{2}/\sigma_{\varepsilon0,j}^{2}$
for all $j\in\{1,...,J\}$, $\Upsilon_{0}$ is also positive definite.
\end{lem}
The proof of the lemma is in Appendix \ref{sec:Variance-Covariance-Matrix}.
Note that from Holder's inequality we have $\mu_{\varepsilon0,j}^{(4)}-\sigma_{\varepsilon0,j}^{4}\geqslant(\mu_{\epsilon0,j}^{(3)})^{2}/\sigma_{\varepsilon0,j}^{2}$.
The sufficient condition is mild in that it only postulates that the
inequality holds strongly. Of course, the condition holds, e.g., for
the Gaussian distribution.

With both $\Upsilon_{0}$ and $\Gamma_{0}$ ensured to be positive
definite, we have the following theorem.
\begin{thm}
\label{theorem:AsymptoticNormlity}Under Assumptions \ref{assume:epsilon}-\ref{assume:id},
and assuming that \textup{$\delta_{0}$ is in the interior of the
parameter space $\Theta$ defined in Assumption \ref{assume:lambda}
and that} $\mu_{\varepsilon0,j}^{(4)}-\sigma_{\varepsilon0,j}^{4}>(\mu_{\epsilon0,j}^{(3)})^{2}/\sigma_{\varepsilon0,j}^{2}$
for $j\in\{1,...,J\}$, we have $\sqrt{N}(\hat{\delta}_{N}-\delta_{0})\xrightarrow{d}N(0,\Gamma_{0}^{-1}\Upsilon_{0}\Gamma_{0}^{-1})$
as $N\rightarrow\infty$.
\end{thm}
The proof of the theorem is given in Appendix \ref{sec:Proof-of-asydis}.
We next discuss consistent estimators for the matrices $\Gamma_{0}$
and $\Upsilon_{0}$ composing the asymptotic variance covariance matrix.
An inspection shows that $\Gamma_{0}=\Gamma(\delta_{0},s_{0})$ and
$\Upsilon_{0}=\Upsilon(\delta_{0},\mu_{\alpha0}^{(3)},\mu_{\alpha0}^{(4)},\mu_{\epsilon0,1}^{(3)},...,\mu_{\epsilon0,J}^{(3)},\mu_{\epsilon0}^{(4)}...,\mu_{\epsilon0,J}^{(4)},s_{0})$,
with $s_{0}=(s_{0,1},...,s_{0,J},m^{\ast})$ and 
\[
s_{0,j}=[\ddot{\varkappa}_{2,j},...\ddot{\varkappa}_{\bar{M},j},\bar{\varkappa}_{2,j},...,\bar{\varkappa}_{\bar{M},j},\bar{z}_{2,j},...,\bar{z}_{\bar{M},j},\omega_{2,j}^{\ast},...,\omega_{\bar{M},j}^{\ast}],
\]
and where the functions $\Gamma(.)$ and $\Upsilon(.)$ are continuous.
Since the functions $\Gamma(.)$ and $\Upsilon(.)$ are continuous,
consistent estimators for $\Gamma_{0}$ and $\Upsilon_{0}$ can be
readily obtained by replacing the arguments of those functions by
consistent estimators thereof. Let $\hat{s}_{N}$ be the sample analogue
of $s_{0}$, then clearly $\hat{s}_{N}\overset{p}{\rightarrow}s_{0}$
in light of Assumptions \ref{assume:n} and \ref{assume:z}. Recall
further that by Theorem \ref{theorem:Consistency} the QMLE estimator
$\hat{\delta}_{N}$ is consistent for $\delta_{0}$, and suppose we
have consistent estimators for $\mu_{\alpha0}^{(3)}$, $\mu_{\alpha0}^{(4)}$,
$\mu_{\epsilon0,1}^{(3)},...,\mu_{\epsilon0,J}^{(3)}$, and $\mu_{\epsilon0,1}^{(4)}...,\mu_{\epsilon0,J}^{(4)}$,
denoted as $\hat{\mu}_{\alpha}^{(3)},\hat{\mu}_{\alpha}^{(4)},\hat{\mu}_{\epsilon,1}^{(3)},...\hat{\mu}_{\epsilon,J}^{(3)},\hat{\mu}_{\epsilon,1}^{(4)},...\hat{\mu}_{\epsilon,J}^{(4)}$.
Now define $\hat{\Gamma}_{N}$ and $\hat{\Upsilon}_{N}$ as 
\begin{align}
\hat{\Gamma}_{N} & =\Gamma(\hat{\delta}_{N},\hat{s}_{N}),\label{eq:Gamma_hat}\\
\hat{\Upsilon}_{N} & =\Upsilon(\hat{\delta}_{N},\hat{\mu}_{\alpha}^{(3)},\hat{\mu}_{\alpha}^{(4)},\hat{\mu}_{\epsilon,1}^{(3)},...\hat{\mu}_{\epsilon,J}^{(3)},\hat{\mu}_{\epsilon,1}^{(4)},...\hat{\mu}_{\epsilon,J}^{(4)},\hat{s}_{N}),\label{eq:Upsilon_hat}
\end{align}
then it follows from Slutsky's theorem that $\hat{\Gamma}_{N}$ and
$\hat{\Upsilon}_{N}$ are consistent estimators for $\Gamma_{0}$
and $\Upsilon_{0}$. A consistent estimator for the variance covariance
matrix of the limiting distribution is given by $\hat{\Gamma}_{N}^{-1}\hat{\Upsilon}_{N}\hat{\Gamma}_{N}^{-1}$.

The above discussion assumed the availability of consistent estimators
for the third and fourth moment of the error components. In the following
we now define consistent estimators for $\mu_{\alpha0}^{(3)}$, $\mu_{\alpha0}^{(4)}$
and $\mu_{\epsilon0,j}^{(3)}$, $\mu_{\epsilon0,j}^{(4)}$, $j=1,...,J$.
To motivate the estimators consider the composite error term for individual
$i$ in group $r$, $u_{ir}=\alpha_{r}+\epsilon_{ir}$, and let $\bar{u}_{r}=\frac{1}{m_{r}}\sum_{i=1}^{m_{r}}u_{ir}$,
and $\ddot{u}_{ir}=u_{ir}-\bar{u}_{r}$ . Then $\bar{u}_{r}=\alpha_{r}+\bar{\epsilon}_{r}$
and $\ddot{u}_{ir}=\epsilon_{ir}-\bar{\epsilon}_{r}$, where $\bar{\epsilon}_{r}$
is the group mean of $\epsilon_{ir}$. It is readily verified that
under Assumptions \ref{assume:epsilon} and \ref{assume:alpha}, we
have 
\[
E\left[\ddot{u}_{ir}^{3}\right]=(1-\frac{3}{m_{r}}+\frac{2}{m_{r}^{2}})\mu_{\epsilon0,D_{r}}^{(3)},
\]

\[
E\left[\ddot{u}_{ir}^{2}\bar{u}_{r}\right]=\frac{(m_{r}-1)}{m_{r}^{2}}\mu_{\epsilon0,D_{r}}^{(3)},
\]

\[
E\left[\bar{u}_{r}^{3}\right]=\mu_{\alpha0}^{(3)}+\frac{\mu_{\epsilon0,D_{r}}^{(3)}}{m_{r}^{2}},
\]

\[
E\left[\ddot{u}_{ir}^{4}\right]=\frac{m_{r}^{3}-4m_{r}^{2}+6m_{r}-3}{m_{r}^{3}}\mu_{\epsilon0,D_{r}}^{(4)}+\frac{3(m_{r}-1)(2m_{r}-3)}{m_{r}^{3}}\sigma_{\epsilon0,D_{r}}^{4},
\]

\[
E\left[\bar{u}_{r}^{4}\right]=\mu_{\alpha0}^{(4)}+\frac{1}{m_{r}^{3}}\mu_{\epsilon0,D_{r}}^{(4)}+3\frac{m_{r}-1}{m_{r}^{3}}\sigma_{\epsilon0,D_{r}}^{4}+\frac{6}{m_{r}}\sigma_{\alpha0}^{2}\sigma_{\epsilon0,D_{r}}^{2}.
\]
Next define for group $r$ ,

\[
f_{\epsilon,r}^{(3)}=\begin{cases}
\frac{1}{m_{r}}\sum_{i=1}^{m_{r}}\ddot{u}_{ir}^{3}/(1-\frac{3}{m_{r}}+\frac{2}{m_{r}^{2}}) & m_{r}\geqslant3\\
\frac{1}{m_{r}}\sum_{i=1}^{m_{r}}\ddot{u}_{ir}^{2}\bar{u}_{r}/(\frac{1}{m_{r}}-\frac{1}{m_{r}^{2}}) & m_{r}=2
\end{cases},
\]

\[
f_{\alpha,r}^{(3)}=\bar{u}_{r}^{3}-f_{\epsilon,r}^{(3)}/m_{r}^{2}\text{,}
\]
\[
f_{\epsilon,r}^{(4)}=\frac{m_{r}^{3}}{m_{r}^{3}-4m_{r}^{2}+6m_{r}-3}\text{\text{[(\ensuremath{\frac{1}{m_{r}}\sum_{i=1}^{m_{r}}\ddot{u}_{ir}^{4}})-\ensuremath{\frac{3(m_{r}-1)(2m_{r}-3)}{m_{r}^{3}}\sigma_{\epsilon0,D_{r}}^{4}}]}, }
\]

\[
f_{\alpha,r}^{(4)}=\bar{u}_{r}^{4}-f_{\epsilon,r}^{(4)}/m_{r}^{3}-\frac{3(m_{r}-1)}{m_{r}^{3}}\sigma_{\epsilon0,D_{r}}^{4}-\frac{6}{m_{r}}\sigma_{\alpha0}^{2}\sigma_{\epsilon0,D_{r}}^{2}.
\]
Then $E\left[f_{\epsilon,r}^{(3)}\right]=\mu_{\epsilon0,D_{r}}^{(3)}$,
$E\left[f_{\alpha,r}^{(3)}\right]=\mu_{\alpha0}^{(3)}$, $E\left[f_{\epsilon,r}^{(4)}\right]=\mu_{\epsilon0,D_{r}}^{(4)}$,
$E\left[f_{\alpha,r}^{(4)}\right]=\mu_{\alpha0}^{(4)}$. By Lemma
\ref{lem:moment34}(a), $\frac{1}{R_{j}}\sum_{r=1}^{R}1(D_{r}=j)f_{\epsilon,r}^{(l)}\overset{p}{\rightarrow}\mu_{\epsilon0,j}^{(l)}$
and $\frac{1}{R}\sum_{r=1}^{R}f_{\alpha,r}^{(l)}\overset{p}{\rightarrow}\mu_{\alpha0}^{(l)}$
for $l=3,4$ and $j=1,...,J$ as $R$ goes to infinity.

To construct feasible counterparts of these estimates, consider the
estimated disturbances $\hat{u}_{ir}=y_{ir}-\hat{\lambda}\bar{y}_{(-i)r}-z_{ir}\hat{\beta}$,
where $\hat{\lambda}$ and $\hat{\beta}$ denote the QML estimators,
and let $\hat{\bar{u}}_{r}=\frac{1}{m_{r}}\sum_{i=1}^{m_{r}}\hat{u}_{ir}$
and $\hat{\ddot{u}}_{ir}=\hat{u}_{ir}-\hat{\bar{u}}_{ir}$. Feasible
counterparts, say, $\hat{f}_{\epsilon,r}^{(3)}$, $\hat{f}_{\alpha,r}^{(3)}$,
$\hat{f}_{\epsilon,r}^{(4)}$, $\hat{f}_{\alpha,r}^{(4)}$ of $f_{\epsilon,r}^{(3)}$,
$f_{\alpha,r}^{(3)}$, $f_{\epsilon,r}^{(4)}$, $f_{\alpha,r}^{(4)}$
can now be defined by replacing $\bar{u}_{r}$ and $\ddot{u}_{ir}$
with $\hat{\bar{u}}_{r}$ and $\hat{\ddot{u}}_{ir}$, and $\sigma_{\alpha0}^{2}$
and $\sigma_{\epsilon0,j}^{2}$ with their QML estimators. Now consider
the following estimators for the third and fourth moments of the error
components: $\hat{\mu}_{\alpha}^{(3)}=\sum_{r=1}^{R}\hat{f}_{\alpha,r}^{(3)}/R$,
$\hat{\mu}_{\alpha}^{(4)}=\sum_{r=1}^{R}\hat{f}_{\alpha,r}^{(4)}/R$,
$\hat{\mu}_{\epsilon,j}^{(3)}=\sum_{r=1}^{R}1(D_{r}=j)\hat{f}_{\epsilon,r}^{(3)}/R_{j}$,
$\hat{\mu}_{\epsilon,j}^{(4)}=\sum_{r=1}^{R}1(D_{r}=j)\hat{f}_{\epsilon,r}^{(4)}/R_{j}$,
$j=1,...,J$.

The next theorem establishes that valid inference based on standardized
statistics is possible. At the core of this result is the fact that
$\hat{\Gamma}_{N}\overset{p}{\rightarrow}\Gamma_{0}$ and $\hat{\Upsilon}_{N}\overset{p}{\rightarrow}\Upsilon_{0}$
as shown in Appendix \ref{app:Proofthm}.
\begin{thm}
\label{thm:ValidInference}Under Assumptions \ref{assume:epsilon}-\ref{assume:id},
and assuming that $\mu_{\varepsilon0,j}^{(4)}-\sigma_{\varepsilon0,j}^{4}>(\mu_{\epsilon0,j}^{(3)})^{2}/\sigma_{\varepsilon0,j}^{2}$
for $j\in\{1,...,J\}$, and $\hat{\Gamma}_{N}$, $\hat{\Upsilon}_{N}$
defined in (\ref{eq:Gamma_hat}) and (\ref{eq:Upsilon_hat}) we have
$\sqrt{N}\left(\hat{\Gamma}_{N}^{-1}\hat{\Upsilon}_{N}\hat{\Gamma}_{N}^{-1}\right)^{-1/2}(\hat{\delta}_{N}-\delta_{0})\xrightarrow{d}N(0,I)$
as $N\rightarrow\infty$.
\end{thm}
The proof of the theorem is in Appendix \ref{app:Proofthm}.

\section{Monte Carlo Results\label{sec:Monte_Carlo}}

We conduct Monte-Carlo (MC) experiments to assess the finite sample
properties of the quasi-maximum likelihood (QML) estimator $\hat{\delta}_{N}$.
The data generating mechanism is determined by the main model in \eqref{eq:main_orig}.
For simplicity, $x_{1,ir}$, $x_{2,ir}$ and $x_{3,ir}$ each only
includes a scalar variable. We set the true value of the parameters
to $\lambda_{0}=0.5$, $\sigma_{\alpha0}^{2}=0.25$, $\beta_{10}=1$,
$\beta_{20}=1$,$\beta_{30}=1$, and $\beta_{40}=1$, while $\sigma_{\epsilon0}^{2}=1$
in the case of homoscedasticity. The model for the data generating
process (DGP) is thus
\begin{equation}
y_{ir}=0.5\bar{y}_{(-i)r}+1+x_{1,ir}+\bar{x}_{2,(-i)r}+x_{3,r}+\alpha_{r}+\epsilon_{ir}.\label{eq:monte_model}
\end{equation}

The inputs $x_{j,ir}$ $\alpha_{r}$ and $\epsilon_{ir}$ are generated
as follows. In the case when $x_{1}=x_{2}$, $x_{1,ir}=x_{2,ir}\sim\textrm{i.i.d.}\,N(0,1)$.
In the case when $x_{1}\neq x_{2}$, $x_{1,ir}$ and $x_{2,ir}$ are
generated mutually independently, each drawn from an $\textrm{i.i.d.}\,N(0,1)$.
We then calculate the leave-out-mean $\bar{x}_{2,(-i)r}=\frac{1}{m_{r}}\sum_{j\neq i}x_{2j,r}$.
Group characteristics are drawn as $x_{3,r}\sim\textrm{i.i.d.}\,N(0,1)$.
In the case of homoscedastic normal errors in Tables \ref{table:simtab_1_0}
to \ref{table:simtab_3_1}, the idiosyncratic error terms $\epsilon_{ir}$
are i.i.d $N(0,1)$ and group effects $\alpha_{r}$ are i.i.d $N(0,0.25)$.
Both $\epsilon_{ir}$ and $\alpha_{r}$ are drawn independently of
$x_{1,ir}$, $x_{2,ir}$, $x_{3,r}$, and of each other. The dependent
variable $y_{ir}$ is calculated using Equation \eqref{eq:y_sol_2}.
In Table \ref{table:simtab_4_1}, we use homoscedastic but nonnormal
errors. In the case of the Skew normal distribution, we set the location
parameter to 0, scale to 1 and shape to $0.9/\sqrt{1-0.9^{2}}$. Therefore,
Skewness is 0.472 and Kurtosis is 3.321. In the case of the student
distribution, degrees of freedom are set to 6. Therefore, Skewness
is 0 and Kurtosis is 6. In both cases, $\alpha_{r}$ and $\epsilon_{ir}$
are independently drawn from identical distributions and then standardized
to have mean 0 and variance 0.25 and 1 respectively. In Table \ref{table:simtab_5_1},
group effects $\alpha_{r}$ are still i.i.d $N(0,0.25)$, $\epsilon_{ir}$
follow normal distributions but are allowed to be heteroscedastic.
In the first case (Columns 1-2), we randomly select half of the groups
into category 1, with $\epsilon_{ir}$ i.i.d $N(0,0.5)$. The other
half of the groups have $\epsilon_{ir}$ i.i.d $N(0,1.5)$. In the
second case (Columns 3-4), $\epsilon_{ir}$ are i.i.d $N(0,1)$. But
we randomly divide the groups into two categories and allow for heteroscedasticity
of $\epsilon_{ir}$ between categories in estimation. In the third
case (Columns 5-6), groups are randomly divided into two categories,
with $\sigma_{\epsilon r}^{2}\in\{0.5,1.5\}$ and $\epsilon_{ir}$
i.i.d $N(0,\sigma_{\epsilon r}^{2})$. In the fourth case, groups
are randomly divided into four categories with $\sigma_{\epsilon r}^{2}\in\{0.4,0.8,1.2,1.6\}$
and $\epsilon_{ir}$ i.i.d $N(0,\sigma_{\epsilon r}^{2})$.

The number of groups $R$ is selected from the set $\{50,100,200,400,800,1600\}$.
In Tables \ref{table:simtab_1_0}, \ref{table:simtab_1_1} and \ref{table:simtab_4_1},
group size $m_{r}$ is drawn from a discrete uniform distribution
$\mathcal{U}\{2,6\}$ so that the average group size is 4. Small group
sizes are motivated by applications to college room mates, friendship
networks in the Add Health data set or golf tournaments, see \citet{sacerdote_peer_2001},
\citet{goldsmith-pinkham_social_2013} and \citet{guryan_peer_2009}.
In Tables \ref{table:simtab_2_0} and \ref{table:simtab_2_1}, group
size is drawn from $\mathcal{U}\{13,25\}$. The distribution is motivated
by Project STAR where class size ranges from 13 to 25. We also consider
the case when $m_{r}$ is drawn from $\mathcal{U}\{3,5\}$, $\mathcal{U}\{4,8\}$,
$\mathcal{U}\{8,30\}$ and $\mathcal{U}\{10,22\}$ in Table \ref{table:simtab_3_1}
to examine how the distribution of group size affects the performance
of the estimator. Note that $\mathcal{U}\{3,5\}$ has the same mean
as $\mathcal{U}\{2,6\}$ but smaller variance, $\mathcal{U}\{4,8\}$
has the same variance as $\mathcal{U}\{2,6\}$ but larger mean. Meanwhile
$\mathcal{U}\{8,30\}$ has the same mean as $\mathcal{U}\{13,25\}$
but larger variance, $\mathcal{U}\{10,22\}$ has the same variance
as $\mathcal{U}\{13,25\}$ but smaller mean.

In Tables 1-\ref{table:simtab_4_1}, we compare our QML estimator
with the conditional maximum likelihood (CML) estimator of \citet{lee_identification_2007}.
Table \ref{table:simtab_5_1} does not present CMLE estimates as it
does not allow for heteroscedasticity. \citet{lee_identification_2007}
assumes normality of the error terms. Our discussion suggests that
the CMLE is in fact consistent under nonnormal errors, as it can be
viewed as a GMM estimator based on the moment conditions from the
within equation. When group effects are in fact independent of the
observed characteristics, the CML estimator is still consistent but
less efficient than our QML estimator. The comparison thus helps to
evaluate the efficiency gain of our estimator over the CML estimator
in finite samples. The CML estimator is based on the within-group
variation hence $\sigma_{\alpha}^{2}$, $\beta_{1}$ and $\beta_{4}$
are not identified.

We generate 5000 repetitions for each of the experiments. Tables \ref{table:simtab_1_0}-\ref{table:simtab_5_1}
summarize the results of the Monte Carlo (MC) experiments. Each panel
displays the MC median, MC robust standard errors (Rob.Std.Dev), MC
sample standard deviation (Std.Dev.), MC median of the estimated standard
deviation (est.Std.Dev), and the mean rejection rate of the Wald test
with significance level 0.05 of our QMLE and Lee's CMLE across 5000
repetitions. The robust standard errors are defined as IQ/1.35, where
IQ denotes the inter-quantile range, that is $IQ=C_{0.75}-C_{0.25}$
with $C_{0.75}$ and $C_{0.25}$ being the 75th and 25th percentile
respectively. If the distribution of the estimate is normal, IQ/1.35
is (apart from rounding errors) equal to the standard deviation. The
null hypothesis for the Wald test is that the estimate equals its
true value. Critical values for the test are obtained at 5\% significance
level and are based on the asymptotic approximation in Theorem \ref{thm:ValidInference}.

Identification of our models is more challenging, the larger group
sizes are, all else equal. This follows from work of \citet{kelejian_2sls_2002}.
Identification is also more difficult when there is less variation
in group sizes, or less variation in type specific variances or both.
Finally, identification is more difficult in designs where $x_{1,ir}=x_{2,ir}$
because the implied correlation between $x_{1,ir}$ and $\bar{x}_{2,(-i)r}$
reduces the overall variation in the covariates. Standard finite sample
theory for the Gaussian regression model shows that maximum likelihood
estimators for the variance parameters are biased in finite samples.
In fixed effects panel regressions this finite sample bias can lead
to inconsistent estimates of the variance parameter due to incidental
parameter bias, as demonstrated by \citet{neyman_consistent_1948}.
In the current context, we expect the CML estimator to suffer from
such incidental parameter bias because the moment conditions that
identify $\lambda$ depend on the estimated variances. We also expect
Wald type statistics, such as the t-ratio, to perform poorly in designs
where identification is problematic, in line with insights from \citet{dufour_impossibility_1997}. 

Tables \ref{table:simtab_1_0} and \ref{table:simtab_1_1} contain
results for small groups and homoscedastic Gaussian errors. In Table
\ref{table:simtab_1_0} where $x_{1}\neq x_{2},$ both the QMLE and
CMLE perform well, with the CMLE being more biased for the parameter
$\lambda$ in sample sizes where $R$ is below 200. The QMLE is generally
less biased and significantly more precise than CMLE, demonstrating
the expected efficiency gains of QMLE. Size is better controlled for
CMLE but the size distortions for the parameters $\lambda$ and $\beta$
do not exceed 7\% in the smallest sample sizes even for the QMLE.
Size distortions for the t-ratios of the two estimated variance parameters
are somewhat larger, reaching 11.6\% for the t-ratio for $\sigma_{\alpha}^{2}$
when $R=50.$ The size distortion seems to be due both to some estimator
bias as well as standard errors that are a bit too small. Size distortions
for all parameters disappear in the larger samples. In Table \ref{table:simtab_1_1}
where $x_{1}=x_{2}$ the CMLE for $\lambda$ is even more biased in
small samples, and considerably more volatile than in the design in
Table \ref{table:simtab_1_0}. The performance of the QMLE is not
very different from the case with $x_{1}\neq x_{2}.$ The standard
deviation measured by IQ/1.35 is somewhat larger than when $x_{1}\neq x_{2}$,
as are size distortions, confirming the intuition that this design
is more difficult to identify.

Tables \ref{table:simtab_2_0} and \ref{table:simtab_2_1} differ
from Tables \ref{table:simtab_1_0} and \ref{table:simtab_1_1} in
that they consider the same designs but with larger group sizes, now
drawn from the uniform distribution on the interval $\left[13,25\right].$
In Table \ref{table:simtab_2_0} we consider the case with $x_{1}\neq x_{2}.$
The QMLE remains roughly unbiased across all sample sizes. The robust
standard deviation roughly doubles relative to the small group size
case and the size properties for t-ratios of the parameters $\lambda$
and $\beta$ deteriorate in samples where $R\leq100$ with size reaching
around 10\% in some cases. Size remains well controlled in larger
samples with $R\geq200$. The size distortions for the variance parameters
are not much affected by the larger class sizes. The CMLE is even
more biased when $R=50$ but less biased for larger sample sizes compared
to Tables \ref{table:simtab_1_0} and \ref{table:simtab_1_1}. This
is consistent with incidental parameter bias which is expected to
decrease with increasing group size. In addition the CMLE now is significantly
less precise. This is in line with results by Lee (2007). Table \ref{table:simtab_2_1}
contains results for the case $x_{1}=x_{2}$ and large group sizes.
The QMLE remains largely unbiased across all sample sizes but there
is notable loss in estimator precision as measured by IQ/1.35, indicating
the more challenging estimation environment. In line with theoretical
predictions, estimator precision increases monotonically with sample
size. Size distortions are now pronounced with empirical size reaching
more than 20\% in the smaller samples. The CMLE controls size well
across all four designs. This comes at the cost of much less precisely
and sometimes more biased estimated parameters. 

Table \ref{table:simtab_3_1} explores the effects that variation
in group size has on both estimators. The case with $\mathcal{U}\{3,5\}$
maintains the same mean group size as in Table \ref{table:simtab_1_1}
but reduces the group size variance. We only report results for $\lambda.$
The bias of the QMLE is not affected while the CMLE is somewhat less
biased. The variance of both estimators increases. For the QMLE size
distortions are somewhat larger than in Table \ref{table:simtab_1_1}.
The design with $\mathcal{U}\{4,8\}$ increases the mean while leaving
the variance of class sizes unchanged relative to Table \ref{table:simtab_1_1}.
Overall, the results for this case are quite similar to the scenario
with $\mathcal{U}\{3,5\}$. The designs with $\mathcal{U}\{8,30\}$
and $\mathcal{U}\{10,22\}$ both improve identification relative to
the design in Table \ref{table:simtab_2_1}. For the QMLE this results
in unchanged good bias properties except when $R=50$ where we now
see a small amount of bias, somewhat lower variance and slightly improved
size properties. For the CMLE bias increases while variance somewhat
improves relative to the results in Table \ref{table:simtab_2_1}
and the size properties remain similar.The larger bias for the CMLE
may be related to a larger fraction of smaller classes in both designs.
Smaller group sizes tend to amplify incidental parameter bias. 

Table \ref{table:simtab_4_1} explores the effects that non-Gaussian
error distributions have on the estimators. For the Skew Normal distribution
we see little difference to the results in Table \ref{table:simtab_1_1}
both for the QMLE and the CMLE estimator. The QMLE is also robust
to the second design which uses a t-distribution with 6 degrees of
freedom. The CMLE is more sensitive to this fat-tailed distribution.
It is somewhat more biased and has higher variance compared to the
Gaussian case. In addition, we now observe size distortions for the
t-ratio related to the parameter $\lambda.$ These size distortions
don't disappear in larger samples and seem to be due to the fact that
the standard errors show a significant downward bias. This is most
likely due to the fact that Lee (2007) bases standard errors on Gaussian
error distributions. The final set of results we discuss are in Table
\ref{table:simtab_5_1} where we examine the effects of heteroscedasticity
on the QMLE. We do not report results for the CMLE since this estimator
was designed for the homoscedastic case only. The first set of results
are based on a design where class size varies according to a $\mathcal{U}\{2,6\}$
distribution and where we maintain $x_{1}=x_{2}$. Compared to a homoscedastic
design the QMLE is somewhat less variable with no change in bias.
The size properties of the t-ratio are overall comparable between
the two cases, with slightly smaller size distortions in the heteroscedastic
case when $R=50.$ We also consider a scenario where group size is
fixed at $m=4$ while the type specific variances vary. While the
QMLE continues to be nearly unbiased it has a higher variance. The
size properties of t-ratios are slightly worse than in the homoscedastic
case. For larger sample sizes both standard errors and t-ratios are
well behaved.

\section{Conclusion}

In this paper, we show that moment conditions underlying the conditional
variance method of Graham (2008) can be related to and motivated from
a general class of linear peer effects models with random group effects.
When augmented with group specific covariates our specification of
the peer effects model is appropriate for settings where people are
randomly assigned to groups or where group level heterogeneity is
credibly controlled for with observed group level characteristics.
We show that the quasi maximum likelihood estimator (QMLE) related
to a linear Gaussian specification, as well as Graham's estimator
and the fixed effects estimator of Lee (2007) are contained in the
class of GMM estimators we consider. Under Gaussian error assumptions
the QMLE is the most efficient estimator in this class. We study conditions
of identification, extending results in Graham (2008) and Lee (2007)
for a simple model without covariates and a general model with covariates
estimated by QML. We also establish that our QMLE is asymptotically
normal and we construct consistent standard error formulas. Monte
Carlo results show that our QML estimator has good small sample properties.

\pagebreak{}

\bibliographystyle{elsarticle-harv}

\pagebreak{}

\appendix

\part*{\renewcommand{\thesection}{\Alph{section}} \numberwithin{equation}{section}}

\part*{Appendix}

\section{Monte Carlo Simulation Results}

\vspace{-5mm} \begin{table}[H] \begin{center}  \caption {Simulation Results:  $ m_r \sim \mathcal{U}\{2,6\} $, Homoscedastic Normal Errors, $ x_1\neq x_2 $} \label{table:simtab_1_0} \hspace*{-15mm} \renewcommand{\arraystretch}{0.84} \begin{threeparttable} \small \begin{tabular}{ccccccccccccc} \hline \hline  & \multicolumn{7}{c}{QMLE} & & \multicolumn{4}{c}{CMLE} \\ \cline{2-8} \cline{10-13}   & $ \lambda $ & $ \sigma^2_{\alpha} $ & $ \sigma^2_{\epsilon} $ & $\beta_1$ & $\beta_2$ & $\beta_3$ & $\beta_4$ & & $ \lambda $ & $ \sigma^2_{\epsilon} $ & $\beta_2$ & $\beta_3$  \\ \cline{1-13} \multicolumn{13}{c}{\textit{True value}} \\ &0.500&0.250&1.000&1.000&1.000&1.000&1.000&&0.500&1.000&1.000&1.000\\  \multicolumn{13}{c}{\textit{50 groups, 200 observations}}  \\  Median&0.500&0.214&0.979&0.999&0.998&0.998&1.001&&0.541&1.004&1.018&1.025\\ Rob.Std.Dev.&(0.070)&(0.177)&(0.118)&(0.173)&(0.076)&(0.162)&(0.178)&&(0.379)&(0.239)&(0.132)&(0.274)\\ Std.Dev.&[0.074]&[0.198]&[0.119]&[0.184]&[0.076]&[0.162]&[0.184]&&[0.428]&[0.275]&[0.139]&[0.294]\\ Est.Std.Dev.&0.067&0.163&0.110&0.166&0.075&0.153&0.168&&0.370&&0.131&0.272\\ Rej.&0.070&0.116&0.095&0.074&0.050&0.063&0.064&&0.036&&0.039&0.041\\  \multicolumn{13}{c}{\textit{100 groups, 400 observations}}  \\  Median&0.498&0.233&0.990&1.002&0.998&0.997&0.999&&0.524&1.007&1.008&1.012\\ Rob.Std.Dev.&(0.049)&(0.128)&(0.084)&(0.122)&(0.053)&(0.110)&(0.121)&&(0.257)&(0.162)&(0.091)&(0.190)\\ Std.Dev.&[0.050]&[0.131]&[0.085]&[0.126]&[0.053]&[0.109]&[0.125]&&[0.280]&[0.176]&[0.095]&[0.195]\\ Est.Std.Dev.&0.048&0.120&0.081&0.119&0.053&0.109&0.120&&0.256&&0.091&0.191\\ Rej.&0.057&0.099&0.074&0.059&0.052&0.051&0.059&&0.046&&0.043&0.044\\  \multicolumn{13}{c}{\textit{200 groups, 800 observations}}  \\  Median&0.500&0.241&0.995&1.000&1.000&1.001&0.998&&0.512&1.003&1.004&1.004\\ Rob.Std.Dev.&(0.035)&(0.089)&(0.060)&(0.085)&(0.037)&(0.080)&(0.087)&&(0.182)&(0.116)&(0.066)&(0.134)\\ Std.Dev.&[0.035]&[0.092]&[0.060]&[0.087]&[0.038]&[0.078]&[0.087]&&[0.188]&[0.119]&[0.066]&[0.138]\\ Est.Std.Dev.&0.034&0.087&0.059&0.085&0.038&0.077&0.085&&0.178&&0.064&0.134\\ Rej.&0.051&0.078&0.062&0.054&0.048&0.054&0.057&&0.052&&0.051&0.054\\  \multicolumn{13}{c}{\textit{400 groups, 1600 observations}}  \\  Median&0.500&0.245&0.998&1.000&1.000&0.999&1.001&&0.504&1.001&1.002&1.001\\ Rob.Std.Dev.&(0.024)&(0.061)&(0.041)&(0.059)&(0.027)&(0.056)&(0.060)&&(0.129)&(0.081)&(0.045)&(0.093)\\ Std.Dev.&[0.024]&[0.063]&[0.042]&[0.060]&[0.027]&[0.054]&[0.060]&&[0.129]&[0.082]&[0.045]&[0.094]\\ Est.Std.Dev.&0.024&0.062&0.042&0.060&0.027&0.055&0.060&&0.125&&0.045&0.094\\ Rej.&0.049&0.061&0.054&0.048&0.050&0.046&0.051&&0.046&&0.047&0.049\\  \multicolumn{13}{c}{\textit{800 groups, 3200 observations}}  \\  Median&0.500&0.247&0.999&0.999&0.999&0.999&0.999&&0.504&1.001&1.001&1.000\\ Rob.Std.Dev.&(0.017)&(0.045)&(0.031)&(0.043)&(0.019)&(0.038)&(0.043)&&(0.090)&(0.056)&(0.032)&(0.066)\\ Std.Dev.&[0.017]&[0.046]&[0.030]&[0.043]&[0.019]&[0.039]&[0.043]&&[0.088]&[0.056]&[0.031]&[0.067]\\ Est.Std.Dev.&0.017&0.044&0.030&0.043&0.019&0.039&0.043&&0.089&&0.032&0.066\\ Rej.&0.056&0.064&0.053&0.053&0.048&0.054&0.054&&0.044&&0.043&0.051\\  \multicolumn{13}{c}{\textit{1600 groups, 6400 observations}}  \\  Median&0.500&0.249&0.999&1.001&1.000&0.999&1.000&&0.503&1.002&1.001&1.002\\ Rob.Std.Dev.&(0.012)&(0.032)&(0.021)&(0.031)&(0.013)&(0.027)&(0.030)&&(0.063)&(0.040)&(0.022)&(0.046)\\ Std.Dev.&[0.012]&[0.032]&[0.021]&[0.031]&[0.013]&[0.027]&[0.031]&&[0.063]&[0.041]&[0.023]&[0.047]\\ Est.Std.Dev.&0.012&0.032&0.021&0.030&0.013&0.027&0.030&&0.063&&0.022&0.047\\ Rej.&0.055&0.057&0.052&0.054&0.047&0.046&0.057&&0.050&&0.050&0.048\\ \hline \hline \end{tabular} \begin{tablenotes} \linespread{1} \footnotesize \item 1. Median value, robust standard deviation (IQ/1.35), standard deviation, median of estimated standard deviation and mean rejection rate of the Wald test of our QMLE and Lee's CMLE across 5000 repetitions. The CMLE is based on the within-group variation hence $ \sigma_\alpha^2, \beta_1, \beta_4 $ are not estimated. Also, Lee(2007) does not offer estimate of the variance for $\sigma_{\epsilon}^2$.  \item 2. Data generating process is based on model (\ref{eq:main_orig}): $ y_{ir}=\beta_{1}+\lambda\bar{y}_{(-i)r}+x_{1,ir}\beta_{2}+\bar{x}_{2,(-i)r}\beta_{3}+x_{3,r}\beta_{4}+\alpha_{r}+\epsilon_{ir} $ , with the true parameter values given in the top panel of the table. Group size $ m_r $ is drawn from  a discrete uniform distribution $ \mathcal{U}\{2,6\} $ .  Sample is generated by: $ x_{1,ir}\sim N(0,1)$, $ x_{2,ir} \sim N(0,1) $, and $ \bar{x}_{2,(-i)r}$ is the leave out mean of $ x_{2,ir} $, $ x_{3,r} \sim N(0,1) $, $ \alpha_r \sim N(0,0.25) $, and $ \epsilon_{ir} \sim N(0,1) $. All variables are independent of each other across i and r .  \end{tablenotes} \end{threeparttable} \end{center} \end{table} 

{\tiny{}\begin{table}[htb!] \begin{center}  \caption {Simulation Results:  $ m_r \sim \mathcal{U}\{2,6\} $, Homoscedastic Normal Errors, $ x_1=x_2 $} \label{table:simtab_1_1} \hspace*{-15mm} \begin{threeparttable} \small \begin{tabular}{ccccccccccccc} \hline \hline  & \multicolumn{7}{c}{QMLE} & & \multicolumn{4}{c}{CMLE} \\ \cline{2-8} \cline{10-13}   & $ \lambda $ & $ \sigma^2_{\alpha} $ & $ \sigma^2_{\epsilon} $ & $\beta_1$ & $\beta_2$ & $\beta_3$ & $\beta_4$ & & $ \lambda $ & $ \sigma^2_{\epsilon} $ & $\beta_2$ & $\beta_3$  \\ \cline{1-13} \multicolumn{13}{c}{\textit{True value}} \\ &0.500&0.250&1.000&1.000&1.000&1.000&1.000&&0.500&1.000&1.000&1.000\\  \multicolumn{13}{c}{\textit{50 groups, 200 observations}}  \\  Median&0.501&0.202&0.976&0.992&0.999&0.994&0.989&&0.592&1.036&1.013&0.995\\ Rob.Std.Dev.&(0.114)&(0.241)&(0.130)&(0.251)&(0.122)&(0.338)&(0.246)&&(0.514)&(0.313)&(0.198)&(0.455)\\ Std.Dev.&[0.114]&[0.275]&[0.129]&[0.250]&[0.122]&[0.336]&[0.253]&&[0.632]&[0.417]&[0.222]&[0.533]\\ Est.Std.Dev.&0.106&0.230&0.116&0.236&0.115&0.313&0.235&&0.488&&0.193&0.444\\ Rej.&0.090&0.120&0.100&0.084&0.081&0.091&0.082&&0.042&&0.040&0.036\\  \multicolumn{13}{c}{\textit{100 groups, 400 observations}}  \\  Median&0.499&0.230&0.988&0.999&1.001&0.998&1.003&&0.539&1.019&1.011&1.008\\ Rob.Std.Dev.&(0.081)&(0.179)&(0.090)&(0.179)&(0.084)&(0.243)&(0.180)&&(0.335)&(0.197)&(0.136)&(0.304)\\ Std.Dev.&[0.080]&[0.190]&[0.092]&[0.177]&[0.085]&[0.239]&[0.178]&&[0.375]&[0.237]&[0.141]&[0.326]\\ Est.Std.Dev.&0.077&0.172&0.085&0.170&0.083&0.227&0.170&&0.327&&0.133&0.303\\ Rej.&0.077&0.104&0.078&0.079&0.070&0.077&0.074&&0.040&&0.044&0.044\\  \multicolumn{13}{c}{\textit{200 groups, 800 observations}}  \\  Median&0.501&0.239&0.995&0.999&0.999&0.996&0.995&&0.528&1.013&1.005&1.000\\ Rob.Std.Dev.&(0.055)&(0.128)&(0.065)&(0.119)&(0.060)&(0.164)&(0.119)&&(0.231)&(0.142)&(0.092)&(0.206)\\ Std.Dev.&[0.055]&[0.129]&[0.064]&[0.122]&[0.059]&[0.163]&[0.120]&&[0.241]&[0.147]&[0.094]&[0.213]\\ Est.Std.Dev.&0.055&0.126&0.062&0.121&0.059&0.163&0.121&&0.229&&0.093&0.208\\ Rej.&0.061&0.082&0.062&0.060&0.054&0.058&0.057&&0.042&&0.043&0.044\\  \multicolumn{13}{c}{\textit{400 groups, 1600 observations}}  \\  Median&0.500&0.243&0.999&0.998&1.000&0.998&0.998&&0.509&1.005&1.002&1.000\\ Rob.Std.Dev.&(0.040)&(0.091)&(0.045)&(0.087)&(0.042)&(0.117)&(0.085)&&(0.159)&(0.097)&(0.064)&(0.143)\\ Std.Dev.&[0.039]&[0.090]&[0.045]&[0.087]&[0.042]&[0.117]&[0.086]&&[0.163]&[0.098]&[0.065]&[0.147]\\ Est.Std.Dev.&0.039&0.090&0.044&0.086&0.042&0.117&0.086&&0.159&&0.065&0.146\\ Rej.&0.055&0.068&0.058&0.058&0.048&0.055&0.055&&0.046&&0.046&0.049\\  \multicolumn{13}{c}{\textit{800 groups, 3200 observations}}  \\  Median&0.501&0.248&0.999&0.999&0.999&0.999&0.999&&0.509&1.004&1.001&0.999\\ Rob.Std.Dev.&(0.027)&(0.062)&(0.031)&(0.060)&(0.030)&(0.080)&(0.060)&&(0.111)&(0.069)&(0.047)&(0.103)\\ Std.Dev.&[0.028]&[0.064]&[0.031]&[0.061]&[0.030]&[0.083]&[0.061]&&[0.115]&[0.070]&[0.046]&[0.104]\\ Est.Std.Dev.&0.028&0.064&0.032&0.061&0.030&0.083&0.061&&0.113&&0.046&0.103\\ Rej.&0.049&0.059&0.043&0.050&0.050&0.053&0.054&&0.049&&0.048&0.050\\  \multicolumn{13}{c}{\textit{1600 groups, 6400 observations}}  \\  Median&0.500&0.249&1.000&1.000&1.000&0.999&1.000&&0.503&1.001&1.001&1.000\\ Rob.Std.Dev.&(0.020)&(0.047)&(0.023)&(0.043)&(0.021)&(0.060)&(0.043)&&(0.078)&(0.047)&(0.032)&(0.073)\\ Std.Dev.&[0.020]&[0.046]&[0.023]&[0.043]&[0.021]&[0.059]&[0.043]&&[0.078]&[0.048]&[0.032]&[0.072]\\ Est.Std.Dev.&0.020&0.046&0.022&0.043&0.021&0.059&0.043&&0.079&&0.032&0.072\\ Rej.&0.052&0.056&0.054&0.048&0.049&0.052&0.053&&0.042&&0.048&0.048\\ \hline \hline \end{tabular} \begin{tablenotes} \footnotesize \item 1. Median value, robust standard deviation (IQ/1.35), standard deviation, median of estimated standard deviation and mean rejection rate of the Wald test of our QMLE and Lee's CMLE across 5000 repetitions. The CMLE is based on the within-group variation hence $ \sigma_\alpha^2, \beta_1, \beta_4 $ are not estimated. Also, Lee(2007) does not offer estimate of the variance for $\sigma_{\epsilon}^2$.  \item 2. Data generating process is based on model (\ref{eq:main_orig}): $ y_{ir}=\beta_{1}+\lambda\bar{y}_{(-i)r}+x_{1,ir}\beta_{2}+\bar{x}_{2,(-i)r}\beta_{3}+x_{3,r}\beta_{4}+\alpha_{r}+\epsilon_{ir} $ , with the true parameter values given in the top panel of the table. Group size $ m_r $ is drawn from  a discrete uniform distribution $ \mathcal{U}\{2,6\} $ .  Sample is generated by: $ x_{1,ir}\sim N(0,1)$, $ x_{2,ir}=x_{1,ir} $ , and $ \bar{x}_{2,(-i)r}$ is the leave out mean of $ x_{2,ir} $, $ x_{3,r} \sim N(0,1) $, $ \alpha_r \sim N(0,0.25) $, and $ \epsilon_{ir} \sim N(0,1) $. All variables are independent of each other across i and r .  \end{tablenotes} \end{threeparttable} \end{center} \end{table} }{\tiny\par}

\begin{table}[htb!] \begin{center}  \caption {Simulation Results:  $ m_r \sim \mathcal{U}\{13,25\} $, Homoscedastic Normal Errors, $ x_1\neq x_2 $} \label{table:simtab_2_0} \hspace*{-15mm} \begin{threeparttable} \small \begin{tabular}{ccccccccccccc} \hline \hline  & \multicolumn{7}{c}{QMLE} & & \multicolumn{4}{c}{CMLE} \\ \cline{2-8} \cline{10-13}   & $ \lambda $ & $ \sigma^2_{\alpha} $ & $ \sigma^2_{\epsilon} $ & $\beta_1$ & $\beta_2$ & $\beta_3$ & $\beta_4$ & & $ \lambda $ & $ \sigma^2_{\epsilon} $ & $\beta_2$ & $\beta_3$  \\ \cline{1-13} \multicolumn{13}{c}{\textit{True value}} \\ &0.500&0.250&1.000&1.000&1.000&1.000&1.000&&0.500&1.000&1.000&1.000\\  \multicolumn{13}{c}{\textit{50 groups, 950 observations}}  \\  Median&0.500&0.224&0.994&0.999&0.998&0.992&0.999&&0.599&1.007&1.005&0.991\\ Rob.Std.Dev.&(0.151)&(0.184)&(0.050)&(0.311)&(0.035)&(0.370)&(0.323)&&(1.781)&(0.198)&(0.101)&(0.597)\\ Std.Dev.&[0.255]&[0.970]&[0.054]&[0.519]&[0.035]&[0.384]&[0.518]&&[1.840]&[0.214]&[0.104]&[0.611]\\ Est.Std.Dev.&0.140&0.167&0.050&0.291&0.034&0.362&0.292&&1.732&&0.099&0.601\\ Rej.&0.097&0.166&0.054&0.096&0.054&0.080&0.094&&0.044&&0.046&0.046\\  \multicolumn{13}{c}{\textit{100 groups, 1900 observations}}  \\  Median&0.501&0.235&0.997&0.999&0.999&0.996&1.000&&0.510&0.999&1.001&0.988\\ Rob.Std.Dev.&(0.105)&(0.132)&(0.036)&(0.219)&(0.024)&(0.258)&(0.217)&&(1.216)&(0.136)&(0.071)&(0.422)\\ Std.Dev.&[0.121]&[0.185]&[0.036]&[0.250]&[0.024]&[0.265]&[0.249]&&[1.264]&[0.143]&[0.073]&[0.427]\\ Est.Std.Dev.&0.101&0.125&0.035&0.210&0.024&0.256&0.210&&1.206&&0.069&0.422\\ Rej.&0.072&0.119&0.053&0.071&0.045&0.067&0.070&&0.058&&0.057&0.049\\  \multicolumn{13}{c}{\textit{200 groups, 3800 observations}}  \\  Median&0.500&0.243&0.999&0.999&1.000&0.994&1.000&&0.517&1.001&1.002&0.997\\ Rob.Std.Dev.&(0.073)&(0.092)&(0.026)&(0.151)&(0.017)&(0.182)&(0.151)&&(0.849)&(0.094)&(0.050)&(0.302)\\ Std.Dev.&[0.077]&[0.107]&[0.025]&[0.159]&[0.017]&[0.184]&[0.160]&&[0.847]&[0.095]&[0.049]&[0.301]\\ Est.Std.Dev.&0.072&0.091&0.025&0.149&0.017&0.181&0.149&&0.853&&0.049&0.298\\ Rej.&0.060&0.084&0.049&0.056&0.051&0.060&0.057&&0.046&&0.048&0.050\\  \multicolumn{13}{c}{\textit{400 groups, 7600 observations}}  \\  Median&0.500&0.246&0.999&1.000&1.000&0.999&1.001&&0.490&0.999&0.999&1.000\\ Rob.Std.Dev.&(0.052)&(0.068)&(0.018)&(0.109)&(0.012)&(0.127)&(0.112)&&(0.606)&(0.068)&(0.034)&(0.211)\\ Std.Dev.&[0.054]&[0.071]&[0.018]&[0.111]&[0.012]&[0.130]&[0.111]&&[0.605]&[0.068]&[0.035]&[0.212]\\ Est.Std.Dev.&0.051&0.065&0.018&0.106&0.012&0.128&0.106&&0.600&&0.034&0.210\\ Rej.&0.051&0.065&0.051&0.055&0.054&0.056&0.050&&0.051&&0.050&0.052\\  \multicolumn{13}{c}{\textit{800 groups, 15200 observations}}  \\  Median&0.499&0.249&1.000&1.001&1.000&1.000&1.001&&0.499&0.999&1.000&1.003\\ Rob.Std.Dev.&(0.036)&(0.047)&(0.013)&(0.075)&(0.008)&(0.093)&(0.074)&&(0.426)&(0.047)&(0.024)&(0.152)\\ Std.Dev.&[0.037]&[0.049]&[0.012]&[0.077]&[0.008]&[0.092]&[0.077]&&[0.434]&[0.048]&[0.025]&[0.151]\\ Est.Std.Dev.&0.036&0.047&0.012&0.075&0.008&0.091&0.075&&0.424&&0.024&0.149\\ Rej.&0.056&0.062&0.050&0.054&0.052&0.052&0.055&&0.054&&0.054&0.055\\  \multicolumn{13}{c}{\textit{1600 groups, 30400 observations}}  \\  Median&0.500&0.249&1.000&1.001&1.000&1.000&1.000&&0.500&1.001&1.000&1.002\\ Rob.Std.Dev.&(0.025)&(0.033)&(0.009)&(0.052)&(0.006)&(0.063)&(0.053)&&(0.306)&(0.033)&(0.017)&(0.108)\\ Std.Dev.&[0.026]&[0.033]&[0.009]&[0.053]&[0.006]&[0.064]&[0.053]&&[0.300]&[0.033]&[0.017]&[0.106]\\ Est.Std.Dev.&0.026&0.033&0.009&0.053&0.006&0.064&0.053&&0.300&&0.017&0.105\\ Rej.&0.051&0.054&0.054&0.055&0.051&0.049&0.049&&0.047&&0.043&0.050\\ \hline \hline \end{tabular} \begin{tablenotes} \footnotesize \item 1. Median value, robust standard deviation (IQ/1.35), standard deviation, median of estimated standard deviation and mean rejection rate of the Wald test of our QMLE and Lee's CMLE across 5000 repetitions. The CMLE is based on the within-group variation hence $ \sigma_\alpha^2, \beta_1, \beta_4 $ are not estimated. Also, Lee(2007) does not offer estiamte of the variance for $\sigma_{\epsilon}^2$.  \item 2. Data generating process is based on model (\ref{eq:main_orig}): $ y_{ir}=\beta_{1}+\lambda\bar{y}_{(-i)r}+x_{1,ir}\beta_{2}+\bar{x}_{2,(-i)r}\beta_{3}+x_{3,r}\beta_{4}+\alpha_{r}+\epsilon_{ir} $ , with the true parameter values given in the top panel of the table. Group size $ m_r $ is drawn from  a discrete uniform distribution $ \mathcal{U}\{13,25\} $ .  Sample is generated by: $ x_{1,ir}\sim N(0,1)$, $ x_{2,ir} \sim N(0,1) $, and $ \bar{x}_{2,(-i)r}$ is the leave out mean of $ x_{2,ir} $, $ x_{3,r} \sim N(0,1) $, $ \alpha_r \sim N(0,0.25) $, and $ \epsilon_{ir} \sim N(0,1) $. All variables are independent of each other across i and r .  \end{tablenotes} \end{threeparttable} \end{center} \end{table} 

\begin{table}[htb!] \begin{center}  \caption {Simulation Results:  $ m_r \sim \mathcal{U}\{13,25\} $, Homoscedastic Normal Errors, $ x_1=x_2 $} \label{table:simtab_2_1} \hspace*{-15mm} \begin{threeparttable} \small \begin{tabular}{ccccccccccccc} \hline \hline  & \multicolumn{7}{c}{QMLE} & & \multicolumn{4}{c}{CMLE} \\ \cline{2-8} \cline{10-13}   & $ \lambda $ & $ \sigma^2_{\alpha} $ & $ \sigma^2_{\epsilon} $ & $\beta_1$ & $\beta_2$ & $\beta_3$ & $\beta_4$ & & $ \lambda $ & $ \sigma^2_{\epsilon} $ & $\beta_2$ & $\beta_3$  \\ \cline{1-13} \multicolumn{13}{c}{\textit{True value}} \\ &0.500&0.250&1.000&1.000&1.000&1.000&1.000&&0.500&1.000&1.000&1.000\\  \multicolumn{13}{c}{\textit{50 groups, 950 observations}}  \\  Median&0.500&0.225&0.987&1.001&1.005&0.995&0.995&&0.545&1.003&0.999&0.898\\ Rob.Std.Dev.&(0.553)&(0.788)&(0.073)&(1.121)&(0.091)&(2.123)&(1.109)&&(2.141)&(0.241)&(0.163)&(3.344)\\ Std.Dev.&[0.592]&[1.868]&[0.078]&[1.198]&[0.093]&[2.076]&[1.194]&&[2.301]&[0.272]&[0.167]&[3.482]\\ Est.Std.Dev.&0.634&0.836&0.079&1.251&0.109&2.517&1.265&&2.120&&0.162&3.336\\ Rej.&0.230&0.238&0.066&0.231&0.114&0.242&0.229&&0.045&&0.045&0.043\\  \multicolumn{13}{c}{\textit{100 groups, 1900 observations}}  \\  Median&0.486&0.253&0.991&1.032&1.004&1.055&1.025&&0.490&0.996&1.001&1.036\\ Rob.Std.Dev.&(0.439)&(0.625)&(0.055)&(0.888)&(0.070)&(1.669)&(0.870)&&(1.463)&(0.160)&(0.112)&(2.286)\\ Std.Dev.&[0.422]&[0.931]&[0.056]&[0.849]&[0.070]&[1.566]&[0.851]&&[1.512]&[0.171]&[0.116]&[2.366]\\ Est.Std.Dev.&0.473&0.599&0.059&0.940&0.082&1.883&0.931&&1.475&&0.113&2.321\\ Rej.&0.221&0.231&0.067&0.218&0.129&0.224&0.220&&0.049&&0.052&0.049\\  \multicolumn{13}{c}{\textit{200 groups, 3800 observations}}  \\  Median&0.500&0.241&0.996&0.994&1.001&0.994&0.999&&0.515&1.000&0.999&0.947\\ Rob.Std.Dev.&(0.320)&(0.414)&(0.040)&(0.636)&(0.052)&(1.218)&(0.639)&&(1.062)&(0.119)&(0.079)&(1.659)\\ Std.Dev.&[0.308]&[0.572]&[0.040]&[0.619]&[0.051]&[1.153]&[0.620]&&[1.055]&[0.118]&[0.080]&[1.643]\\ Est.Std.Dev.&0.340&0.416&0.043&0.677&0.059&1.341&0.678&&1.044&&0.079&1.638\\ Rej.&0.195&0.208&0.066&0.190&0.132&0.189&0.189&&0.046&&0.052&0.047\\  \multicolumn{13}{c}{\textit{400 groups, 7600 observations}}  \\  Median&0.495&0.252&0.997&1.011&1.002&1.019&1.009&&0.480&0.998&1.000&1.024\\ Rob.Std.Dev.&(0.248)&(0.320)&(0.031)&(0.495)&(0.041)&(0.939)&(0.498)&&(0.712)&(0.078)&(0.057)&(1.154)\\ Std.Dev.&[0.236]&[0.381]&[0.030]&[0.474]&[0.039]&[0.889]&[0.473]&&[0.735]&[0.081]&[0.057]&[1.178]\\ Est.Std.Dev.&0.249&0.307&0.031&0.494&0.043&0.968&0.497&&0.735&&0.056&1.153\\ Rej.&0.171&0.181&0.073&0.169&0.134&0.172&0.170&&0.048&&0.056&0.052\\  \multicolumn{13}{c}{\textit{800 groups, 15200 observations}}  \\  Median&0.498&0.250&0.998&1.004&1.000&1.010&1.004&&0.505&1.000&1.000&0.992\\ Rob.Std.Dev.&(0.175)&(0.219)&(0.022)&(0.351)&(0.030)&(0.672)&(0.351)&&(0.517)&(0.058)&(0.040)&(0.821)\\ Std.Dev.&[0.168]&[0.242]&[0.022]&[0.336]&[0.029]&[0.641]&[0.337]&&[0.523]&[0.058]&[0.040]&[0.826]\\ Est.Std.Dev.&0.178&0.218&0.022&0.356&0.030&0.684&0.356&&0.520&&0.040&0.817\\ Rej.&0.129&0.144&0.073&0.130&0.113&0.132&0.127&&0.052&&0.051&0.052\\  \multicolumn{13}{c}{\textit{1600 groups, 30400 observations}}  \\  Median&0.497&0.253&0.999&1.005&1.001&1.007&1.007&&0.498&1.000&1.001&1.004\\ Rob.Std.Dev.&(0.127)&(0.158)&(0.016)&(0.254)&(0.021)&(0.484)&(0.253)&&(0.369)&(0.041)&(0.028)&(0.586)\\ Std.Dev.&[0.123]&[0.165]&[0.016]&[0.246]&[0.021]&[0.468]&[0.246]&&[0.368]&[0.041]&[0.028]&[0.577]\\ Est.Std.Dev.&0.126&0.155&0.016&0.253&0.021&0.484&0.253&&0.368&&0.028&0.577\\ Rej.&0.111&0.120&0.076&0.110&0.093&0.108&0.111&&0.047&&0.050&0.050\\ \hline \hline \end{tabular} \begin{tablenotes} \footnotesize \item 1. Median value, robust standard deviation (IQ/1.35), standard deviation, median of estimated standard deviation and mean rejection rate of the Wald test of our QMLE and Lee's CMLE across 5000 repetitions. The CMLE is based on the within-group variation hence $ \sigma_\alpha^2, \beta_1, \beta_4 $ are not estimated. Also, Lee(2007) does not offer estiamte of the variance for $\sigma_{\epsilon}^2$.  \item 2. Data generating process is based on model (\ref{eq:main_orig}): $ y_{ir}=\beta_{1}+\lambda\bar{y}_{(-i)r}+x_{1,ir}\beta_{2}+\bar{x}_{2,(-i)r}\beta_{3}+x_{3,r}\beta_{4}+\alpha_{r}+\epsilon_{ir} $ , with the true parameter values given in the top panel of the table. Group size $ m_r $ is drawn from  a discrete uniform distribution $ \mathcal{U}\{13,25\} $ .  Sample is generated by: $ x_{1,ir}\sim N(0,1)$, $ x_{2,ir}=x_{1,ir} $ , and $ \bar{x}_{2,(-i)r}$ is the leave out mean of $ x_{2,ir} $, $ x_{3,r} \sim N(0,1) $, $ \alpha_r \sim N(0,0.25) $, and $ \epsilon_{ir} \sim N(0,1) $. All variables are independent of each other across i and r .  \end{tablenotes} \end{threeparttable} \end{center} \end{table} 

\begin{table}[!h] \begin{center}  \caption {Simulation Results for $ \lambda $ : Alternative Group Size Distributions, Homoscedastic Normal Errors, $ x_1=x_2 $}   \label{table:simtab_3_1}   \begin{threeparttable} \small \begin{tabular}{cccccccccccc} \hline \hline  & \multicolumn{2}{c}{$ m_r\sim \mathcal{U}\{3,5\} $} & & \multicolumn{2}{c}{$ m_r\sim \mathcal{U}\{4,8\} $} & & \multicolumn{2}{c}{$ m_r\sim \mathcal{U}\{8,30\} $ } & & \multicolumn{2}{c}{$ m_r\sim \mathcal{U}\{10,22\} $ } \\ \cline{2-3} \cline{5-6} \cline{8-9} \cline{11-12}   & QMLE & CMLE & & QMLE & CMLE & & QMLE & CMLE & & QMLE & CMLE \\ \cline{2-12}   & $ \lambda $ & $ \lambda $ & & $ \lambda $ & $ \lambda $ & & $ \lambda $ & $ \lambda $ & & $ \lambda $ & $ \lambda $ \\ \hline  \multicolumn{12}{c}{\textit{True value}} \\ &0.500&0.500&&0.500&0.500&&0.500&0.500&&0.500&0.500\\  \multicolumn{12}{c}{\textit{50 Groups}}  \\  Median&0.502&0.561&&0.494&0.575&&0.487&0.545&&0.481&0.567\\ Rob.Std.Dev.&(0.222)&(0.916)&&(0.257)&(0.993)&&(0.331)&(1.090)&&(0.457)&(1.682)\\ Std.Dev.&[0.201]&[1.580]&&[0.241]&[1.193]&&[0.340]&[1.136]&&[0.451]&[1.777]\\ Est.Std.Dev.&0.205&0.887&&0.243&0.953&&0.330&1.088&&0.483&1.603\\ Rej.&0.163&0.064&&0.175&0.058&&0.175&0.040&&0.209&0.051\\  \multicolumn{12}{c}{\textit{100 Groups}}  \\  Median&0.500&0.542&&0.497&0.549&&0.493&0.531&&0.496&0.550\\ Rob.Std.Dev.&(0.164)&(0.629)&&(0.184)&(0.667)&&(0.243)&(0.758)&&(0.329)&(1.140)\\ Std.Dev.&[0.151]&[0.734]&&[0.175]&[0.736]&&[0.237]&[0.785]&&[0.325]&[1.184]\\ Est.Std.Dev.&0.152&0.614&&0.179&0.659&&0.242&0.761&&0.348&1.129\\ Rej.&0.140&0.054&&0.145&0.049&&0.151&0.050&&0.188&0.046\\  \multicolumn{12}{c}{\textit{200 Groups}}  \\  Median&0.502&0.534&&0.503&0.518&&0.498&0.534&&0.500&0.537\\ Rob.Std.Dev.&(0.110)&(0.437)&&(0.127)&(0.459)&&(0.171)&(0.529)&&(0.249)&(0.812)\\ Std.Dev.&[0.108]&[0.464]&&[0.127]&[0.481]&&[0.169]&[0.534]&&[0.240]&[0.811]\\ Est.Std.Dev.&0.110&0.430&&0.129&0.460&&0.173&0.538&&0.254&0.793\\ Rej.&0.113&0.046&&0.119&0.049&&0.121&0.047&&0.168&0.047\\  \multicolumn{12}{c}{\textit{400 Groups}}  \\  Median&0.497&0.507&&0.500&0.517&&0.499&0.511&&0.497&0.501\\ Rob.Std.Dev.&(0.081)&(0.298)&&(0.095)&(0.333)&&(0.124)&(0.383)&&(0.183)&(0.552)\\ Std.Dev.&[0.079]&[0.312]&&[0.093]&[0.332]&&[0.124]&[0.381]&&[0.173]&[0.561]\\ Est.Std.Dev.&0.080&0.299&&0.093&0.325&&0.125&0.379&&0.185&0.557\\ Rej.&0.085&0.046&&0.093&0.050&&0.105&0.046&&0.135&0.049\\  \multicolumn{12}{c}{\textit{800 Groups}}  \\  Median&0.501&0.508&&0.500&0.505&&0.498&0.501&&0.501&0.499\\ Rob.Std.Dev.&(0.058)&(0.211)&&(0.067)&(0.229)&&(0.088)&(0.265)&&(0.129)&(0.388)\\ Std.Dev.&[0.057]&[0.216]&&[0.065]&[0.232]&&[0.087]&[0.266]&&[0.127]&[0.388]\\ Est.Std.Dev.&0.057&0.211&&0.066&0.228&&0.089&0.267&&0.130&0.394\\ Rej.&0.067&0.051&&0.072&0.049&&0.080&0.047&&0.107&0.046\\  \multicolumn{12}{c}{\textit{1600 Groups}}  \\  Median&0.499&0.506&&0.500&0.499&&0.501&0.504&&0.499&0.497\\ Rob.Std.Dev.&(0.040)&(0.150)&&(0.046)&(0.161)&&(0.062)&(0.188)&&(0.093)&(0.274)\\ Std.Dev.&[0.040]&[0.152]&&[0.047]&[0.163]&&[0.062]&[0.187]&&[0.091]&[0.278]\\ Est.Std.Dev.&0.040&0.149&&0.047&0.161&&0.063&0.189&&0.093&0.278\\ Rej.&0.059&0.048&&0.061&0.052&&0.066&0.044&&0.077&0.050\\ \hline \hline \end{tabular} \begin{tablenotes} \footnotesize \item 1. Median value, robust standard deviation (IQ/1.35), standard deviation, median of estimated standard deviation and mean rejection rate of the Wald test of our QMLE and Lee's CMLE across 5000 repetitions.  For simplicity, we only present estimates of the endogeneous peer effects ($\lambda$). \item 2. Data generating process is based on model (\ref{eq:main_orig}): $ y_{ir}=\beta_{1}+\lambda\bar{y}_{(-i)r}+x_{1,ir}\beta_{2}+\bar{x}_{2,(-i)r}\beta_{3}+x_{3,r}\beta_{4}+\alpha_{r}+\epsilon_{ir} $ , with the $\lambda=0.5$ and all $ \beta $ s being 1. Sample is generated by: $ x_{1,ir}\sim N(0,1)$, $ x_{2,ir}=x_{1,ir} $  and $ \bar{x}_{2,(-i)r}$ is the leave out mean of $ x_{2,ir} $, $ x_{3,r} \sim N(0,1) $, $ \alpha_r \sim N(0,0.25) $, and $ \epsilon_{ir} \sim N(0,1) $. All variables are independent of each other across i and r. \item 3. Group size $ m_r $ is drawn from $ \mathcal{U}\{3,5\} $ (Case 1), $ \mathcal{U}\{4,8\} $ (Case 2), $ \mathcal{U}\{8,30\} $ (Case 3), $ \mathcal{U}\{10,22\} $ (Case 4). \end{tablenotes} \end{threeparttable} \end{center} \end{table} 

\begin{table}[!h] \begin{center}  \caption {Simulation Results: Homoscedastic but Nonnormal Errors, $ m_r \sim \mathcal{U}\{2,6\} $, $ x_1=x_2 $} \label{table:simtab_4_1}  \begin{threeparttable} \small \begin{tabular}{cccccccccc} \hline \hline  & \multicolumn{4}{c}{Skew Normal} & & \multicolumn{4}{c}{Student Distribution} \\ \cline{2-5} \cline{7-10}   & \multicolumn{2}{c}{QMLE} & \multicolumn{2}{c}{CMLE} & & \multicolumn{2}{c}{QMLE} & \multicolumn{2}{c}{CMLE} \\ \cline{2-10}  & $ \lambda $ &  $\beta_3$ & $ \lambda $ &  $\beta_3$ &  & $ \lambda $ &  $\beta_3$ & $ \lambda $ &  $\beta_3$ \\ \hline \multicolumn{10}{c}{\textit{True value}} \\ &0.500&1.000&0.500&1.000&&0.500&1.000&0.500&1.000\\  \multicolumn{10}{c}{\textit{50 Groups}}  \\  Median&0.500&0.993&0.580&0.991&&0.499&1.004&0.610&1.010\\ Rob.Std.Dev.&(0.119)&(0.345)&(0.525)&(0.443)&&(0.106)&(0.303)&(0.633)&(0.386)\\ Std.Dev.&[0.114]&[0.339]&[0.634]&[0.516]&&[0.109]&[0.315]&[0.862]&[0.470]\\ Est.Std.Dev.&0.107&0.314&0.483&0.446&&0.099&0.288&0.487&0.390\\ Rej.&0.091&0.089&0.047&0.033&&0.084&0.084&0.078&0.034\\  \multicolumn{10}{c}{\textit{100 Groups}}  \\  Median&0.501&1.002&0.539&0.993&&0.500&1.001&0.553&1.007\\ Rob.Std.Dev.&(0.080)&(0.241)&(0.346)&(0.309)&&(0.074)&(0.213)&(0.422)&(0.268)\\ Std.Dev.&[0.080]&[0.236]&[0.380]&[0.325]&&[0.075]&[0.219]&[0.468]&[0.284]\\ Est.Std.Dev.&0.077&0.228&0.326&0.301&&0.070&0.207&0.328&0.262\\ Rej.&0.082&0.070&0.046&0.043&&0.070&0.071&0.085&0.042\\  \multicolumn{10}{c}{\textit{200 Groups}}  \\  Median&0.501&1.000&0.520&0.998&&0.499&1.001&0.534&1.002\\ Rob.Std.Dev.&(0.057)&(0.169)&(0.235)&(0.209)&&(0.051)&(0.152)&(0.295)&(0.186)\\ Std.Dev.&[0.058]&[0.171]&[0.254]&[0.217]&&[0.052]&[0.154]&[0.313]&[0.189]\\ Est.Std.Dev.&0.055&0.164&0.228&0.208&&0.051&0.149&0.229&0.182\\ Rej.&0.074&0.068&0.058&0.046&&0.065&0.062&0.121&0.047\\  \multicolumn{10}{c}{\textit{400 Groups}}  \\  Median&0.500&1.000&0.515&0.999&&0.501&0.998&0.516&0.998\\ Rob.Std.Dev.&(0.039)&(0.115)&(0.169)&(0.143)&&(0.036)&(0.107)&(0.204)&(0.128)\\ Std.Dev.&[0.039]&[0.116]&[0.172]&[0.148]&&[0.036]&[0.109]&[0.219]&[0.131]\\ Est.Std.Dev.&0.039&0.117&0.160&0.146&&0.036&0.106&0.159&0.127\\ Rej.&0.054&0.053&0.059&0.047&&0.056&0.056&0.131&0.052\\  \multicolumn{10}{c}{\textit{800 Groups}}  \\  Median&0.500&1.000&0.504&0.998&&0.501&0.998&0.510&0.998\\ Rob.Std.Dev.&(0.028)&(0.086)&(0.116)&(0.104)&&(0.026)&(0.076)&(0.152)&(0.091)\\ Std.Dev.&[0.028]&[0.084]&[0.119]&[0.104]&&[0.026]&[0.077]&[0.155]&[0.093]\\ Est.Std.Dev.&0.028&0.083&0.112&0.102&&0.026&0.075&0.112&0.089\\ Rej.&0.052&0.055&0.062&0.052&&0.056&0.052&0.137&0.058\\  \multicolumn{10}{c}{\textit{1600 Groups}}  \\  Median&0.500&1.000&0.504&0.999&&0.501&1.000&0.504&1.001\\ Rob.Std.Dev.&(0.019)&(0.057)&(0.084)&(0.071)&&(0.018)&(0.052)&(0.103)&(0.064)\\ Std.Dev.&[0.019]&[0.058]&[0.085]&[0.071]&&[0.018]&[0.054]&[0.106]&[0.064]\\ Est.Std.Dev.&0.020&0.059&0.079&0.072&&0.018&0.054&0.079&0.063\\ Rej.&0.045&0.045&0.068&0.043&&0.054&0.053&0.136&0.051\\ \hline \hline \end{tabular} \begin{tablenotes} \footnotesize \item 1. Median value, robust standard deviation (IQ/1.35), standard deviation, median of estimated standard deviation and mean rejection rate of the Wald test of our QMLE and Lee's CMLE across 5000 repetitions.  For simplicity, we only present estimates of the endogeneous peer effects ($\lambda$) and exogenous peer effects ($\beta_3$). \item 2. Data generating process is based on model (\ref{eq:main_orig}): $ y_{ir}=\beta_{1}+\lambda\bar{y}_{(-i)r}+x_{1,ir}\beta_{2}+\bar{x}_{2,(-i)r}\beta_{3}+x_{3,r}\beta_{4}+\alpha_{r}+\epsilon_{ir} $ , with $\lambda=0.5$ and all $ \beta $ s being 1. Group size $ m_r $ is drawn from  a discrete uniform distribution $ \mathcal{U}\{2,6\} $. Sample is generated by: $ x_{1,ir}\sim N(0,1)$, $ x_{2,ir}=x_{1,ir} $ , and $ \bar{x}_{2,(-i)r}$ is the leave out mean of $ x_{2,ir} $, $ x_{3,r} \sim N(0,1) $. All variables are independent of each other across i and r .  \item 3. In the case of Skew normal distribution, location is 0, scale is 1 and shape is $ 0.9/\sqrt{1-0.9^2} $. In the case of student distribution, degree of freedom is 6. In all cases, $\alpha_r$ and $\epsilon_{ir}$ are indepdently drawn from identical distribtion and standardized to have mean 0 and variance 0.25, 1 respectively.  \end{tablenotes} \end{threeparttable} \end{center} \end{table} 

\begin{table}[!h] \begin{center}  \caption {Simulation Results: Heteroscedastic Normal Errors, $ x_1=x_2 $} \label{table:simtab_5_1}   \begin{threeparttable} \small \begin{tabular}{cccccccccccc} \hline \hline & \multicolumn{5}{c}{$ m_r \sim \mathcal{U}\{2,6\} $} & & \multicolumn{5}{c}{$ m_r =4$ for all}  \\ \cline{2-6} \cline{8-12}  & \multicolumn{2}{c}{$ \sigma_{\epsilon }^2 \in \{0.5,1.5\} $} & & \multicolumn{2}{c}{$ \sigma_{\epsilon} =1 $ for all}  & & \multicolumn{2}{c}{$ \sigma_{\epsilon}^2 \in \{0.5,1.5\} $} & & \multicolumn{2}{c}{$ \sigma_{\epsilon}^2 \in \{0.4,0.8,1.2,1.6\} $} \\ \cline{2-3} \cline{5-6} \cline{8-9} \cline{11-12} & $ \lambda $ &  $ \beta_{3} $ & & $ \lambda $ &  $ \beta_{3}$ & & $ \lambda $ &  $ \beta_{3} $ & & $ \lambda $ &  $ \beta_{3}$ \\ \hline \multicolumn{12}{c}{\textit{True value}} \\ &0.500&1.000&&0.500&1.000&&0.500&1.000&&0.500&1.000\\  \multicolumn{12}{c}{\textit{50 Groups}}  \\  Median&0.502&0.994&&0.504&0.990&&0.485&1.047&&0.499&1.003\\ Rob.Std.Dev.&(0.097)&(0.298)&&(0.116)&(0.352)&&(0.246)&(0.794)&&(0.305)&(0.989)\\ Std.Dev.&[0.100]&[0.299]&&[0.118]&[0.352]&&[0.783]&[2.493]&&[0.969]&[3.083]\\ Est.Std.Dev.&0.090&0.265&&0.106&0.311&&0.173&0.566&&0.172&0.545\\ Rej.&0.079&0.085&&0.093&0.098&&0.104&0.103&&0.144&0.141\\  \multicolumn{12}{c}{\textit{100 Groups}}  \\  Median&0.502&0.996&&0.501&0.999&&0.500&1.003&&0.500&1.002\\ Rob.Std.Dev.&(0.067)&(0.196)&&(0.080)&(0.239)&&(0.147)&(0.479)&&(0.182)&(0.591)\\ Std.Dev.&[0.067]&[0.200]&&[0.079]&[0.235]&&[0.479]&[1.509]&&[0.639]&[2.012]\\ Est.Std.Dev.&0.064&0.191&&0.077&0.227&&0.125&0.410&&0.130&0.422\\ Rej.&0.077&0.076&&0.074&0.074&&0.085&0.081&&0.122&0.114\\  \multicolumn{12}{c}{\textit{200 Groups}}  \\  Median&0.500&1.004&&0.500&0.998&&0.496&1.011&&0.498&1.007\\ Rob.Std.Dev.&(0.046)&(0.140)&&(0.055)&(0.168)&&(0.099)&(0.321)&&(0.114)&(0.365)\\ Std.Dev.&[0.047]&[0.140]&&[0.056]&[0.166]&&[0.187]&[0.593]&&[0.327]&[1.018]\\ Est.Std.Dev.&0.045&0.136&&0.055&0.163&&0.092&0.300&&0.096&0.312\\ Rej.&0.068&0.065&&0.068&0.067&&0.070&0.067&&0.093&0.086\\  \multicolumn{12}{c}{\textit{400 Groups}}  \\  Median&0.500&1.002&&0.500&1.000&&0.498&1.002&&0.500&1.000\\ Rob.Std.Dev.&(0.033)&(0.098)&&(0.040)&(0.122)&&(0.068)&(0.221)&&(0.072)&(0.236)\\ Std.Dev.&[0.032]&[0.097]&&[0.040]&[0.119]&&[0.088]&[0.283]&&[0.128]&[0.405]\\ Est.Std.Dev.&0.032&0.097&&0.039&0.117&&0.066&0.214&&0.069&0.222\\ Rej.&0.056&0.057&&0.059&0.059&&0.057&0.055&&0.067&0.061\\  \multicolumn{12}{c}{\textit{800 Groups}}  \\  Median&0.500&0.999&&0.500&1.000&&0.500&1.002&&0.500&1.001\\ Rob.Std.Dev.&(0.023)&(0.071)&&(0.028)&(0.082)&&(0.049)&(0.161)&&(0.050)&(0.160)\\ Std.Dev.&[0.023]&[0.069]&&[0.028]&[0.083]&&[0.053]&[0.171]&&[0.056]&[0.180]\\ Est.Std.Dev.&0.023&0.069&&0.028&0.083&&0.047&0.152&&0.049&0.159\\ Rej.&0.055&0.050&&0.059&0.054&&0.050&0.050&&0.056&0.055\\  \multicolumn{12}{c}{\textit{1600 Groups}}  \\  Median&0.500&1.001&&0.500&1.001&&0.499&1.003&&0.500&0.998\\ Rob.Std.Dev.&(0.016)&(0.048)&&(0.020)&(0.060)&&(0.033)&(0.109)&&(0.035)&(0.116)\\ Std.Dev.&[0.016]&[0.049]&&[0.020]&[0.059]&&[0.034]&[0.111]&&[0.037]&[0.120]\\ Est.Std.Dev.&0.016&0.049&&0.020&0.059&&0.033&0.108&&0.035&0.113\\ Rej.&0.053&0.053&&0.055&0.059&&0.046&0.044&&0.050&0.045\\ \hline \hline \end{tabular} \begin{tablenotes} \linespread{1} \footnotesize \item 1. Median value, robust standard deviation (IQ/1.35), standard deviation, median of estimated standard deviation and mean rejection rate of the Wald test of our QMLE across 5000 repetitions.  For simplicity, we only present estimates of the endogeneous peer effects ($\lambda$) and exogenous peer effects ($\beta_3$). \item 2. Data generating process is based on model (\ref{eq:main_orig}): $ y_{ir}=\beta_{1}+\lambda\bar{y}_{(-i)r}+x_{1,ir}\beta_{2}+\bar{x}_{2,(-i)r}\beta_{3}+x_{3,r}\beta_{4}+\alpha_{r}+\epsilon_{ir} $ , with $\lambda=0.5$ and all $ \beta $ s being 1. Sample is generated by: $ x_{1,ir}\sim N(0,1)$, $ x_{2,ir}=x_{1,ir} $  and $ \bar{x}_{2,(-i)r}$ is the leave out mean of $ x_{2,ir} $, $ x_{3,r} \sim N(0,1) $, $ \alpha_r \sim N(0,0.25) $. When there are more than one category of $ \sigma_\epsilon $, groups are equally distributed into different categories. All variables are independent of each other across i and r.  \item 3. In the first case (Columns 1-2), the model has both heteroscedasticity and group size variation. In the second case (Columns 3-4),the DGP has homoscedastic $\sigma^2_{\epsilon}$ and group size variation. But the estimation process assumes two categories of $\sigma^2_{\epsilon}$. In both case 1 and case 2, group size $ m_r $ is drawn from  $ \mathcal{U}\{2,6\} $. In Cases 3 and 4, group size is 4 for all.  \end{tablenotes} \end{threeparttable} \end{center} \end{table} 

\clearpage{}

\section{Preliminaries\label{app:qudratic}}

In proving consistency and asymptotic normality of the QMLE estimator
we encounter linear quadratic forms of the form 
\begin{equation}
S_{N}(\theta)=U^{\prime}A_{N}(\theta)U+U^{\prime}a_{N}(\theta)\label{eq:Sn}
\end{equation}
where $A_{N}(\theta)$ is an $N\times N$ non-stochastic matrix, $a_{N}(\theta)$
is an $N-$dimensional non-stochastic column vector, and where $A_{N}(\theta)$
and $a_{N}(\theta)$ exhibit some special structures. In the following
we describe that structure in more detail, and collect some basic
lemmata used in proving the consistency and asymptotic normality of
the QMLE.

We adopt the following notation: Partition an $N\times N$ matrix
$A_{N}$ into $R\times R$ submatrices, with the $(r,r^{\prime})$-th
submatrix being an $m_{r}\times m_{r^{\prime}}$ matrix, $r,r^{\prime}=1,...,R$.
We then denote the $(r,r^{\prime})$-th submatrix of $A_{N}$ as $A_{(r,r^{\prime}),N}$,
and the $(i,j)$-th element of $A_{N}$ as $a_{ij,N}$, $1\leqslant i\leqslant N$,
$1\leqslant j\leqslant N$. Partition an $N\times1$ vector $a_{N}$
into $R$ subvectors, with the $r$-th subvector being an $m_{r}\times1$
vector. We then denote the $r$-th subvector of $a_{N}$ as $a_{(r),N}$
and the $i$-th element of $a_{N}$ as $a_{i,N}$. In line with \citet{kelejian_asymptotic_2001},
we call the column and row sums of an $N\times N$ matrix $A_{N}(\theta)$
uniformly bounded in absolute value if there exists some finite constant
$C$ (which does not depend on $N$ or $\theta$) such that 
\[
\begin{array}{c}
\sup_{\theta\in\Theta}\sum_{i=1}^{N}|a(\theta)_{ij,N}|\leq C,\quad\sup_{\theta\in\Theta}\sum_{j=1}^{N}|a(\theta)_{ij,N}|\leq C.\end{array}
\]
A corresponding definition applies to rectangular matrices. Of course,
if the row sums of $A_{N}(\theta)$ are uniformly bounded in absolute
value, and the elements of $a_{N}(\theta)$ are uniformly bounded
in absolute value, then the elements of $A_{N}(\theta)a_{N}(\theta)$
are uniformly bounded in absolute value. Note that if the row and
column sums of $A_{N}(\theta)$ and $B_{N}(\theta)$ are uniformly
bounded in absolute value, then $A_{N}(\theta)+B_{N}(\theta)$ and
$A_{N}(\theta)B_{N}(\theta)$ (if dimension permits addition or multiplication)
also have row and column sums uniformly bounded in absolute value.\footnote{This is readily seen by argumentation in line with \citet{kelejian_generalized_1999}.}

\subsection{Basic Properties of Matrices Forming the Log-Likelihood Function
\label{subsec:matrix properties}}

Recall that $\theta=(\theta_{1},...,\theta_{J+2})$, with $\theta_{1}=\lambda,\text{\ensuremath{\theta_{2}=\sigma_{\alpha}^{2}}, and \ensuremath{\theta_{j+2}=\sigma_{\epsilon,j}^{2}}}$
for $j=1,...,J$, and that in light of Assumptions \ref{assume:epsilon},
\ref{assume:alpha}, and \ref{assume:lambda} the parameter space
$\Theta$ is compact. An inspection of the expression of the log-likelihood
function shows that it depends on the following set of matrices: $I-\lambda W$,
$(I-\lambda W)^{-1}$, $\Omega(\theta)$, $\Omega(\theta)^{-1}$,
$W$. For generic functions $p(m_{r},D_{r},\theta)$ and $s(m_{r},D_{r},\theta)$
all these matrices are symmetric block diagonal matrices of the form
\begin{equation}
A_{N}(\theta)=\diag_{r=1}^{R}\left\{ p(m_{r},D_{r},\theta)I_{m_{r}}^{\ast}+s(m_{r},D_{r},\theta)J_{m_{r}}^{\ast}\right\} .\label{AppCon0}
\end{equation}
In particular, by replacing $p\left(.\right)$ and $s\left(.\right)$
with specific functions $\phi_{S},\phi_{\Omega},\phi_{W},\psi_{S},\psi_{\Omega}$
and $\psi_{W}$ defined below, one obtains
\begin{eqnarray}
I-\lambda W & = & \diag_{r=1}^{R}\left\{ \phi_{S}(m_{r},\theta)I_{m_{r}}^{\ast}+\psi_{S}(m_{r},\theta)J_{m_{r}}^{\ast}\right\} \label{AppCon1}\\
(I-\lambda_{0}W)^{-1} & = & \diag_{r=1}^{R}\left\{ \phi_{S}^{-1}(m_{r},\theta_{0})I_{m_{r}}^{\ast}+\psi_{S}^{-1}(m_{r},\theta_{0})J_{m_{r}}^{\ast}\right\} ,\nonumber \\
\Omega_{0} & = & \diag_{r=1}^{R}\{\phi_{\Omega}(m_{r},D_{r},\theta_{0})I_{m_{r}}^{\ast}+\psi_{\Omega}(m_{r},D_{r},\theta_{0})J_{m_{r}}^{\ast}\},\nonumber \\
\Omega(\theta)^{-1} & = & \diag_{r=1}^{R}\{\phi_{\Omega}^{-1}(m_{r},D_{r},\theta)I_{m_{r}}^{\ast}+\psi_{\Omega}^{-1}(m_{r},D_{r},\theta)J_{m_{r}}^{\ast}\},\nonumber \\
W & = & \diag_{r=1}^{R}\{\phi_{W}(m_{r},\theta)I_{m_{r}}^{\ast}+\psi_{W}(m_{r},\theta)J_{m_{r}}^{\ast}\},\nonumber 
\end{eqnarray}
where 
\begin{equation}
\begin{array}{cc}
\phi_{S}(m_{r},\theta)=\frac{m_{r}-1+\lambda}{m_{r}-1}, & \psi_{S}(m_{r},\theta)=1-\lambda,\\
\phi_{\Omega}(m_{r},D_{r},\theta)=\sigma_{\epsilon,D_{r}}^{2}, & \psi_{\Omega}(m_{r},D_{r},\theta)=\sigma_{\epsilon,D_{r}}^{2}+m_{r}\sigma_{\alpha}^{2}.\\
\phi_{W}(m_{r},\theta)=-\frac{1}{m_{r}-1}, & \psi_{W}(m_{r},\theta)=1
\end{array}\label{AppCon2}
\end{equation}

It is readily seen that there exists an open bounded set $\Theta_{o}$
such that $\Theta\subset\Theta_{o}\subset(-1,1)\times R^{J+1}$ such
that the placeholder functions $p(m_{r},D_{r},\theta)$ and $s(m_{r},D_{r},\theta)$,
explicitly defined in \eqref{AppCon2}, are continuously differentiable
on $\Theta_{o}$. Thus, by Bolzano-Weierstrass' extreme value theorem
there exists a positive constant $C$, which does not depend on $\theta$,
such that 
\begin{equation}
0\leq\left\vert p(m_{r},D_{r},\theta)\right\vert ,\left\vert s(m_{r},D_{r},\theta)\right\vert ,\left\vert \partial p(m_{r},D_{r},\theta)/\partial\theta_{i}\right\vert ,\left\vert \partial s(m_{r},D_{r},\theta)/\partial\theta_{i}\right\vert \leq C<\infty,\label{AppCon2a}
\end{equation}
for all $\theta\in\Theta.$ \footnote{Of course, since $m_{r}$ only takes on finitely many values, the
constants can also be taken such that they do not depend on $m_{r}$.
We note, although not stated explicitly, all subsequent uniformity
results also hold uniformly for $m_{r}\in\{m:2\leq m\leq\bar{M}\}$.} This implies that $p(m_{r},D_{r},\theta)$ and $s(m_{r},D_{r},\theta)$
are both uniformly continuous on $\Theta$. Observing that $\phi_{\Omega}(m_{r},D_{r},\theta)$
and $\psi_{\Omega}(m_{r},D_{r},\theta)$ are positive on $\Theta$
it follows further that there exists a positive constant $c$, which
does not depend on $\theta$, such that
\begin{equation}
0<c\leq\phi_{\Omega}(m_{r},D_{r},\theta),\psi_{\Omega}(m_{r},D_{r},\theta)\leq C<\infty.\label{AppCon2b}
\end{equation}
Since $I_{m_{r}}^{\ast}$ and $J_{m_{r}}^{\ast}$ are orthogonal and
idempotent, the multiplication of block diagonal matrices, where the
blocks are of the form $p(m_{r},D_{r},\theta)I_{m_{r}}^{\ast}+s(m_{r},D_{r},\theta)J_{m_{r}}^{\ast}$,
yields a matrix with the same structure. Furthermore the multiplication
of those matrices is commutative. More specifically, let

\[
A_{N}(\theta)=\diag_{r=1}^{R}\left\{ p(m_{r},D_{r},\theta)I_{m_{r}}^{\ast}+s(m_{r},D_{r},\theta)J_{m_{r}}^{\ast}\right\} 
\]
 and $\mathring{A}_{N}(\theta)=\diag_{r=1}^{R}\left\{ \text{\ensuremath{\mathring{p}}}(m_{r},D_{r},\theta)I_{m_{r}}^{\ast}+\mathring{s}(m_{r},D_{r},\theta)J_{m_{r}}^{\ast}\right\} $,
then
\[
A_{N}(\theta)\mathring{A}_{N}(\theta)=\diag_{r=1}^{R}\left\{ p(m_{r},D_{r},\theta)\text{\ensuremath{\mathring{p}}}(m_{r},D_{r},\theta)I_{m_{r}}^{\ast}+s(m_{r},D_{r},\theta)\mathring{s}(m_{r},D_{r},\theta)J_{m_{r}}^{\ast}\right\} .
\]
 In addition, $A_{N}(\theta)\text{ and }\mathring{A}_{N}(\theta)$
commute, $A_{N}(\theta)\mathring{A}_{N}(\theta)=\mathring{A}_{N}(\theta)A_{N}(\theta)$.
Also,
\begin{equation}
\left\vert A_{N}(\theta)\right\vert =\prod_{j=1}^{J}\prod\limits _{m=2}^{\bar{M}}\left\vert p(m,j,\theta)I_{m}^{\ast}+s(m,j,\theta)J_{m}^{\ast}\right\vert ^{R_{m,j}}=\prod_{j=1}^{J}\prod\limits _{m=2}^{\bar{M}}\left[p(m,j,\theta)^{m-1}s(m,j,\theta)\right]^{R_{m,j}},\label{AppCon2c}
\end{equation}
as is readily checked observing that $pI_{m}^{\ast}+sJ_{m}^{\ast}=p\{I_{m}+[(s-p)/(pm)]\iota_{m}\iota_{m}^{\prime}\}$
and applying Proposition 31 in \citet{dhrymes_mathematics_1978},
Section 2.7 on p. 38, to compute the determinant of the matrix in
curly brackets. Furthermore,

\[
\tr\left(A_{N}(\theta)\right)=\sum_{j=1}^{J}\sum_{m=2}^{\bar{M}}R_{m,j}\left((m-1)p(m,j,\theta)+s(m,j,\theta)\right).
\]
and 
\begin{eqnarray}
{\color{blue}}\frac{1}{N}Z^{\prime}A_{N}(\theta)Z & = & \frac{1}{N}\sum_{r=1}^{R}\left(p(m_{r},D_{r},\theta)Z_{r}^{\prime}I_{m_{r}}^{\ast}Z_{r}+s(m_{r},D_{r},\theta)Z_{r}^{\prime}J_{m_{r}}^{\ast}Z_{r}\right)\label{AppCon3}\\
 & = & \sum_{j=1}^{J}\sum_{m=2}^{\bar{M}}\left(p(m,j,\theta)\frac{1}{N}\sum_{r\in\mathcal{I}_{m,j}}\ddot{Z}_{r}^{\prime}\ddot{Z}_{r}+s(m,j,\theta)\frac{1}{N}\sum_{r\in\mathcal{I}_{m,j}}m\bar{z}_{r}^{\prime}\bar{z}_{r}\right).\nonumber 
\end{eqnarray}

We note that the row and column sums of any matrix $A_{N}(\theta)$
of the form (\ref{AppCon0}) are uniformly bounded in absolute value,
if $p(m_{r},D_{r},\theta)$ and $s(m_{r},D_{r},\theta)$ are uniformly
bounded in absolute value (observing that $m_{r}$ is bounded by Assumption
\ref{assume:n}). We note further that in light of Assumption \ref{assume:z}
the elements of $N^{-1}Z^{\prime}A_{N}(\theta)Z$ are uniformly bounded
in absolute value. If additionally $p(m_{r},D_{r},\theta)$ and $s(m_{r},D_{r},\theta)$
are positive and bounded away from zero, then also the elements of
$\left(N^{-1}Z^{\prime}A_{N}(\theta)Z\right)^{-1}$ are uniformly
bounded; see Lemma \ref{lem:Aux1}. Consequently the elements of $\left(N^{-1}Z^{\prime}\Omega(\theta)Z\right)^{-1}$
are uniformly bounded, and the row and column sums of $Z(Z^{\prime}\Omega(\theta)Z)^{-1}Z^{\prime}=N^{-1}Z\left(N^{-1}Z^{\prime}\Omega(\theta)Z\right)^{-1}Z^{\prime}$,
$\Omega(\theta)^{-1}Z\left(N^{-1}Z^{\prime}\Omega(\theta)Z\right)^{-1}Z^{\prime}\Omega(\theta)^{-1}$
and $M_{Z}(\theta)$ are uniformly bounded in absolute value. As a
result, $M_{Z}(\theta)$ and $\partial M_{Z}(\theta)/\partial\theta_{i}=-M_{Z}(\theta)\left(\partial\Omega(\theta)/\partial\theta_{i}\right)M_{Z}(\theta)$
have row and column sums uniformly bounded in absolute value.

In all, if a matrix $A_{N}(\theta)$ is the product of $I-\lambda W$,
$(I-\lambda W)^{-1}$, $\Omega(\theta)$, $\Omega(\theta)^{-1}$,
$W$, $\partial\Omega(\theta)/\partial\theta_{i}$, and $M_{Z}(\theta)$,
then both $A_{N}(\theta)$ and $\partial A_{N}(\theta)/\partial\theta_{i}$
have row and column sums uniformly bounded in absolute value, and
the elements of $A_{N}(\theta)Z\beta_{0}$ are uniformly bounded in
absolute value over $\theta\in\Theta$ and $N$.

\subsection{Limit Theorems for Linear Quadratic Forms in $U$}

The following result follows trivially from Lemma A.1 in \citet{kelejian_specification_2010},
and is only given for the convenience of the reader.
\begin{lem}
{[}Mean and Covariance{]} \label{lemma:moment}Let $A$ and $B$ be
$N\times N$ nonstochastic symmetric matrices, which are partitioned
into $R^{2}$ submatrices and let $a$ and $b$ be $N\times1$ vectors,
which are conformably partitioned into $R$ subvectors. Let $a_{ij}$
and $b_{ij}$ denote the $(i,j)$-th element of $A$ and $B$, let
$A_{(r,r^{\prime})}$ and $B_{(r.r^{\prime})}$ denote the $(r,r^{\prime})$-th
block of dimension $m_{r}\times m_{r^{\prime}}$, let $\vecd(A_{(r,r)})$
and $\vecd(B_{(r,r)})$ denote the column vectors of the diagonal
elements of $A_{(r.r)}$ and $B_{(r.r)}$, let $a_{i}$ and $b_{i}$
denote the $i$-th element of $a$ and $b$, and $a_{(r)}$ and $b_{(r)}$
denote the $r$-th subvectors of dimension $m_{r}\times1$. Let $J_{m}$
be the $m\times m$ matrix of ones. Then, under Assumptions \ref{assume:epsilon}
and \ref{assume:alpha}, 
\begin{eqnarray*}
 &  & E(U^{\prime}AU+U^{\prime}a)=tr(\Omega_{0}A),\\
 &  & \cov(U^{\prime}AU+U^{\prime}a,U^{\prime}BU+U^{\prime}b)=2\tr\left(A\Omega_{0}B\Omega_{0}\right)+a^{\prime}\Omega_{0}b\\
 &  & +\sum_{r=1}^{R}\vecd(A_{(r,r)})'\vecd(B_{(r,r)})(\mu_{\epsilon0,D_{r}}^{(4)}-3\sigma_{\epsilon0,D_{r}}^{4})+\sum_{r=1}^{R}\left[\tr(A_{(rr)}J_{m_{r}})\right]\left[\tr(B_{(rr)}J_{m_{r}})\right](\mu_{\alpha0}^{(4)}-3\sigma_{\alpha0}^{4})\\
 &  & +\sum_{r=1}^{R}\left(\vecd(A_{(r,r)})'b_{(r)}+\vecd(B_{(r,r)})'a_{(r)}\right)\mu_{\epsilon0,D_{r}}^{(3)}+\sum_{r=1}^{R}\left[\iota_{m_{r}}^{\prime}A_{(rr)}J_{m_{r}}b_{(r)}+\iota_{m_{r}}^{\prime}B_{(rr)}J_{m_{r}}a_{(r)}\right]\mu_{\alpha0}^{(3)}.
\end{eqnarray*}
\end{lem}
\begin{proof}
Let $H=[\sigma_{\alpha0}\diag_{r=1}^{R}\left\{ \iota_{m_{r}}\right\} ,\diag_{r=1}^{R}\{\sigma_{\epsilon0,D_{r}}I_{m_{r}}\}]$.
Consider the $(N+R)\times1$ dimensional vector 
\[
\xi=(\alpha_{1}/\sigma_{\alpha0},....,\alpha_{R}/\sigma_{\alpha0},\epsilon_{1}^{\prime}/\sigma_{\epsilon0,D_{1}},...,\epsilon_{R}^{\prime}/\sigma_{\epsilon0,D_{R}})^{\prime},
\]
Then $U=H\xi$, and 
\begin{equation}
U^{\prime}AU+U^{\prime}a=\xi^{\prime}(H^{\prime}AH)\xi+\xi^{\prime}(H^{\prime}a).\label{eq:Sn2}
\end{equation}
Note that by Assumptions \ref{assume:epsilon} and \ref{assume:alpha},
the elements of $\xi$ are independently distributed with $E\left[\xi\right]=0_{(N+R)\times1}$,
$\var\left(\xi\right)=I_{N+R}$. Denote the $i$-th entry of $\xi$
as $\xi_{i}$, $1\leqslant i\leqslant N+R$. Denote the third and
fourth moments of $\xi_{i}$ as $\mu_{\xi_{i}}^{(3)}$ and $\mu_{\xi_{i}}^{(4)}$
respectively. Under Assumptions \ref{assume:epsilon} and \ref{assume:alpha},
when $1\leqslant i\leqslant R$, $\mu_{\xi_{i}}^{(3)}=\mu_{\alpha0}^{(3)}/\sigma_{\alpha0}^{3}$
and $\mu_{\xi_{i}}^{(4)}=\mu_{\alpha0}^{(4)}/\sigma_{\alpha0}^{4}$.
When $R+m_{1}+...+m_{r-1}+1\leqslant i\leqslant R+m_{1}+...+m_{r}$,
$\mu_{\xi_{i}}^{(3)}=\mu_{\epsilon0,D_{r}}^{(3)}/\sigma_{\epsilon0,D_{r}}^{3}$
and $\mu_{\xi_{i}}^{(4)}=\mu_{\epsilon0,D_{r}}^{(4)}/\sigma_{\epsilon0,D_{r}}^{4}$.
Furthermore, there exists some $\eta_{\xi}>0$ such that $E[|\xi_{i}|^{4+\eta_{\xi}}]<\infty$.

Using the transformation of linear quadratic forms in (\ref{eq:Sn2})
and applying Lemma A.1 in \citet{kelejian_specification_2010} yields,
\[
E\left[U^{\prime}AU+U^{\prime}a\right]=E[\xi^{\prime}(H^{\prime}AH)\xi+\xi^{\prime}(H^{\prime}a)]=\tr(H^{\prime}AH)=\tr\left(A\Omega_{0}\right),
\]
observing that 
\[
HH^{\prime}=\sigma_{\alpha0}^{2}\diag_{r=1}^{R}\left\{ \iota_{m_{r}}\iota_{m_{r}}^{\prime}\right\} +\diag_{r=1}^{R}\left\{ \sigma_{\epsilon0,D_{r}}^{2}I_{m_{r}}\right\} =\Omega_{0}.
\]
 Furthermore the variance of the linear quadratic forms in $U$ is
given by 
\begin{align*}
 & \cov(U^{\prime}AU+U^{\prime}a,U^{\prime}BU+U^{\prime}b)\\
= & \cov\left(\xi^{\prime}(H^{\prime}AH)\xi+\xi^{\prime}(H^{\prime}a),\xi^{\prime}(H^{\prime}BH)\xi+\xi^{\prime}(H^{\prime}b)\right)\\
= & 2\tr\left(H^{\prime}AHH^{\prime}BH\right)+a^{\prime}HH^{\prime}b\\
+ & \sum_{i=1}^{N+R}(H^{\prime}AH)_{ii}(H^{\prime}BH)_{ii}(\mu_{\xi_{i}}^{(4)}-3)+\sum_{i=1}^{N+R}[(H^{\prime}AH)_{ii}(H^{\prime}b)_{i}+(H^{\prime}BH)_{ii}(H^{\prime}a)_{i}]\mu_{\xi_{i}}^{(3)}\\
= & 2\tr\left(A\Omega_{0}B\Omega_{0}\right)+a^{\prime}\Omega_{0}b\\
+ & \sum_{r=1}^{R}\vecd(A_{(r,r)})'\vecd(B_{(r,r)})(\mu_{\epsilon0,D_{r}}^{(4)}-3\sigma_{\epsilon0,D_{r}}^{4})+\sum_{r=1}^{R}\left[\tr(A_{(rr)}J_{m_{r}})\right]\left[\tr(B_{(rr)}J_{m_{r}})\right](\mu_{\alpha0}^{(4)}-3\sigma_{\alpha0}^{4})\\
+ & \sum_{r=1}^{R}(\vecd(A_{(r,r)})'b_{(r)}+\vecd(B_{(r,r)})'a_{(r)})\mu_{\epsilon0,D_{r}}^{(3)}+\sum_{r=1}^{R}\left(\iota_{m_{r}}^{\prime}A_{(rr)}J_{m_{r}}b_{(r)}+\iota_{m_{r}}^{\prime}B_{(rr)}J_{m_{r}}a_{(r)}\right)\mu_{\alpha0}^{(3)}.
\end{align*}
\end{proof}
\begin{lem}
{[}Central Limit Theorem{]} \label{lem:clt} Suppose Assumptions 1-6
hold. For $l=1,\ldots,L$ let $A_{N}^{(l)}$ be non-stochastic $N\times N$
matrices where the row and column sums of the absolute elements are
uniformly bounded in $N$, and let $a_{N}^{(l)}$ be $N\times1$ non-stochastic
vectors where the absolute elements are uniformly bounded in $N$.
Let $S_{N}=[S_{N}^{(1)},S_{N}^{(2)},..,S_{N}^{(L)}]^{\prime}$ be
an $L\times1$ vector of linear quadratic forms of $U$, with 
\[
S_{N}^{(l)}=U^{\prime}A_{N}^{(l)}U+U^{\prime}a_{N}^{(l)},\quad l=1,...,L.
\]
Let $\Sigma_{S,N}$ denote variance covariance matrix of $S_{N}$,
where explicit expressions for the elements of $\Sigma_{S,N}$ are
readily obtained from Lemma \ref{lemma:moment}, and assume that $\rho_{min}(\Sigma_{S,N})\geq c$
for some constant $c>0$. Let $\Sigma_{S,N}=\Sigma_{S,N}^{1/2}\Sigma_{S,N}^{1/2}$,
then 
\[
\Sigma_{S,N}^{-1/2}(S_{N}-E\left[S_{N}\right])\xrightarrow{d}N(0,I_{L})
\]
as $N\rightarrow\infty$. \\
(Note that under Assumption \ref{assume:z} the conditions postulated
for $a_{N}^{(l)}$ hold if $a_{N}^{(l)}=B_{N}^{(l)}Z\beta_{0}$, and
the $B_{N}^{(l)}$ are non-stochastic $N\times N$ matrices where
the row and column sums of the absolute elements are uniformly bounded
in $N$.)
\end{lem}
\begin{proof}
Let $H$ and $\xi$ be defined as in the proof of Lemma \ref{lemma:moment},
so that $U=H\xi$. Upon substitution of this expression for $U$ we
have 
\[
S_{N}^{(l)}=\xi^{\prime}\tilde{A}_{N+R}^{(l)}\xi+\xi^{\prime}\tilde{a}_{N+R}^{(l)}
\]
where $\tilde{A}_{N+R}^{(l)}=(1/2)H^{\prime}(A_{N}^{(l)}+A_{N}^{(l)\prime})H$,
$\tilde{a}_{N+R}^{(l)}=H^{\prime}a_{N}^{(l)}$. Clearly, in light
of Assumptions \ref{assume:epsilon} and \ref{assume:alpha}, $\xi$
satisfies Assumptions A.1 and A.3 in \citet{kelejian_specification_2010}.
Furthermore, given the maintained assumptions on $A_{N}^{(l)}$ and
$a_{N}^{(l)}$, and since the row an column sums of $H$ are uniformly
bounded in absolute value, it follows that the row and column sums
of $\tilde{A}_{N+R}^{(l)}$ and the elements of $\tilde{a}_{N+R}^{(l)}$
are uniformly bounded in absolute value. This verifies that those
matrices and vectors satisfy the conditions of Assumption A.2 in \citet{kelejian_specification_2010}.
The lemma now follows from Theorem A.1 in \citet{kelejian_specification_2010}.
\end{proof}
As above, let $\Theta_{o}$ be an open bounded set with $\Theta\subset\Theta_{o}\subset(-1,1)\times R^{J+1}$.
\begin{lem}
{[}Uniform Convergence{]} \label{lem:Uniform-convergence} Let $\Theta_{0}$
be an open set containing $\Theta$. Let $A_{N}(\theta)$ and $B_{N}(\theta)$
be $N\times N$ matrices and let $S_{N}(\theta)$ be a linear-quadratic
form of $U$: 
\[
S_{N}(\theta)=U^{\prime}A_{N}(\theta)U+U^{\prime}B_{N}(\theta)Z\beta_{0}
\]
where $A_{N}(\theta)$ and $B_{N}(\theta)$ are differentiable $N\times N$
matrices defined for $\theta\in\Theta_{0}$. Suppose Assumptions 1-5
hold, and suppose the row and column sums of $A_{N}(\theta),B_{N}(\theta)$,
$\partial A_{N}(\theta)/\partial\theta_{i}$ and $\partial B_{N}(\theta)/\partial\theta_{i}$,
$i=1,...,J+2$, are bounded in absolute value uniformly in $N$ and
$\theta$. Then $N^{-1}S_{N}(\theta)-$ $N^{-1}E\left[S_{N}(\theta)\right]$
converges uniformly to zero i.p., i.e., 
\[
\underset{N\rightarrow\infty}{\plim}\sup_{\theta\in\Theta}|N^{-1}S_{N}(\theta)-N^{-1}E\left[S_{N}(\theta)\right]|=0.
\]
\end{lem}
\begin{rem}
Given the uniform convergence in probability of $N^{-1}S_{N}(\theta)$
to its mean and the equicontinuity of $N^{-1}S_{N}(\theta)$, we have
$\underset{N\rightarrow\infty}{\plim}|N^{-1}S_{N}(\hat{\theta}_{N})-N^{-1}E\left[S_{N}(\theta_{0})\right]|$
as $N$ goes to infinity if $\hat{\theta}_{N}\rightarrow\theta_{0}$.
\end{rem}
\begin{proof}
To prove the lemma we verify that $N^{-1}S_{N}(\theta)$ and $N^{-1}\bar{S}_{N}(\theta)=N^{-1}E\left[S_{N}(\theta)\right]$
satisfy the conditions postulated by Corollary 2.2 of \citet{newey_uniform_1991};
cp., also Theorem 3.1(a) and the discussion after eq. (2.7) in \citet{potscher_generic_1994}.

The parameter space $\Theta$ is compact by assumption. We next verify
that $N^{-1}\bar{S}_{N}(\theta)$ is uniformly equicontinuous. By
Lemma \ref{lemma:moment}, $N^{-1}\bar{S}_{N}(\theta)=N^{-1}\tr(\Omega_{0}A(\theta))$.
Let $\theta,\theta^{\prime}\in\Theta$, then by the mean value theorem
\[
\tr(\Omega_{0}A_{N}(\theta))=\tr(\Omega_{0}A_{N}(\theta^{\prime}))+[\tr(\Omega_{0}\frac{\partial A_{N}(\theta^{\ast})}{\partial\theta_{1}}),...,\tr(\Omega_{0}\frac{\partial A_{N}(\theta^{\ast})}{\partial\theta_{J+2}})](\theta-\theta^{\prime}).
\]
where $\theta^{\ast}$ is a ``vector of between values\textquotedblright .
Note that the row and column sums of $A_{N}(\theta)$, $\nabla_{\theta_{i}}A_{N}(\theta)=\partial A_{N}(\theta)/\partial\theta_{i}$,
$\Omega_{0}$, and consequently the row and column sums of $\Omega_{0}A_{N}(\theta)$
and $\left(\Omega_{0}\nabla_{\theta_{i}}A_{N}(\theta)\right)$, are
uniformly (in $\theta$ and $N$) bounded in absolute value. Consequently
there exists a constant $C_{A}$ which does not depend on $\theta$,
$\theta^{\prime}$, or $N$ such that 
\begin{equation}
\left\vert N^{-1}\tr\left(\Omega_{0}A_{N}(\theta)\right)-N^{-1}\tr\left(\Omega_{0}A_{N}(\theta^{\prime})\right)\right\vert \leqslant C_{A}||\theta-\theta^{\prime}||,\label{eq:uniequi}
\end{equation}
which establishes that $N^{-1}\bar{S}_{N}(\theta)=N^{-1}\tr(\Omega_{0}A_{N}(\theta))$
is uniformly equicontinuous on $\Theta$.

We next prove point-wise convergence i.p. of $N^{-1}S_{N}(\theta)-N^{-1}E\left[S_{N}(\theta)\right]$
to zero. In light of Chebychev's inequality it suffices to show that
the variance of $N^{-1}S_{N}(\theta)$ converges to $0$ for any $\theta\in\Theta$.
Let $\overline{A}_{N}=\left(A_{N}(\theta)+A_{N}^{\prime}(\theta)\right)/2$
and $\overline{a}_{N}=B_{N}(\theta)Z\beta_{0}$, then by Lemma \ref{lemma:moment},
the variance of $S_{N}(\theta)$ is
\begin{align*}
\var\left(N^{-1/2}S_{N}(\theta)\right)= & N^{-1}\tr\left(\overline{A}_{N}\Omega_{0}\overline{A}_{N}\Omega_{0}\right)+N^{-1}\overline{a}_{N}^{\prime}\Omega_{0}\overline{a}_{N}\\
+ & N^{-1}\sum_{i=1}^{N}(\overline{a}_{ii,N})^{2}(\mu_{\epsilon0,D_{r(i)}}^{(4)}-3\sigma_{\epsilon0,D_{r(i)}}^{4})+N^{-1}\sum_{r=1}^{R}\left(\tr(\overline{A}_{(rr),N}J_{m_{r}})\right)^{2}(\mu_{\alpha}^{(4)}-3\sigma_{\alpha0}^{4})\\
+ & 2N^{-1}\sum_{i=1}^{N}\overline{a}_{ii,N}\overline{a}_{i,N}\mu_{\epsilon0,D_{r(i)}}^{(3)}+2N^{-1}\sum_{r=1}^{R}\left(\iota_{m_{r}}^{\prime}\overline{A}_{(rr),N}J_{m_{r}}\overline{a}_{(r),N}\right)\mu_{\alpha}^{(3)}.
\end{align*}
Under our assumptions the row and column sums of $\Omega_{0}$, $\overline{A}_{N}$,
and thus of $\overline{A}_{N}\Omega_{0}\overline{A}_{N}\Omega_{0}$
are uniformly bounded. Furthermore, it is readily seen that the elements
of $\overline{a}_{N}$ are uniformly bounded in absolute value. Consequently
$N^{-1}\tr\left(\overline{A}_{N}\Omega_{0}\overline{A}_{N}\Omega_{0}\right)$,
$N^{-1}\overline{a}_{N}^{\prime}\Omega_{0}\overline{a}_{N}$, and
all sums in the above expression are seen to be bounded by a finite
constant uniformly in $N$. In turn this implies that $\var(N^{-1}S_{N}(\theta))\rightarrow0$.

Finally we prove that $N^{-1}S_{N}(\theta)$ satisfies the following
Lipschitz condition:
\[
\left\vert N^{-1}S_{N}(\theta)-N^{-1}S_{N}(\theta^{\prime})\right\vert \leq C_{N}\left\Vert \theta-\theta^{\prime}\right\Vert 
\]
for all $\theta,\theta^{\prime}\in\Theta$ and some nonnegative random
variable $C_{N}$ that does not depend on $\theta,\theta^{\prime}$
and where $C_{N}=O_{p}(1)$. It prove again convenient to rewrite
as $S_{N}(\theta)$ as $S_{N}(\theta)=\xi^{\prime}\tilde{A}_{N}(\theta)\xi+\xi^{\prime}\tilde{a}_{N}(\theta)$
with $\tilde{A}_{N}(\theta)=H^{\prime}A_{N}(\theta)H$ and $\tilde{a}_{N}(\theta)=H^{\prime}B_{N}(\theta)Z\beta_{0}$,
where $H$ and $\xi$ are defined as in the proof of Lemma \ref{lemma:moment}.
Under the maintained assumptions $\tilde{A}_{N}(\theta)$ and $\tilde{a}_{N}(\theta)$
are differentiable for $\theta\in\Theta_{o}$, and the row and column
sums of $\tilde{A}_{N}(\theta)$, $\partial\tilde{A}_{N}(\theta)/\partial\theta_{i}$
and the elements of $\tilde{a}_{N}(\theta)$, $\partial\tilde{a}_{N}(\theta)/\partial\theta_{i}$
are uniformly bounded in absolute value in $\theta$ and $N$, with
$i=1,...,J+2$. Consequently, for some finite constant, say $K$,
we have $\left\vert \tilde{a}_{i,N}(\theta)\right\vert \leq K/2$
and, using the mean value theorem,
\[
\sum_{j=1}^{N+R}\left\vert \tilde{a}_{ij,N}(\theta)-\tilde{a}_{ij,N}(\theta^{\prime})\right\vert \leq\sum_{j=1}^{N+R}\left\Vert \partial\tilde{a}_{ij,N}(\theta_{\ast})/\partial\theta\right\Vert \left\Vert \theta-\theta^{\prime}\right\Vert \leq K\left\Vert \theta-\theta^{\prime}\right\Vert ,
\]
with $\theta_{\ast}$ a ``between value\textquotedblright . Observing
further that $|\xi_{i}\xi_{j}|\leqslant(\xi_{i}^{2}+\xi_{j}^{2})/2$
we have for any $\theta,\theta^{\prime}\in\Theta$ 
\begin{eqnarray*}
 &  & \left\vert N^{-1}S_{N}(\theta)-N^{-1}S_{N}(\theta^{\prime})\right\vert =\left\vert N^{-1}\xi^{\prime}\left[\tilde{A}_{N}(\theta)-\tilde{A}_{N}(\theta^{\prime})\right]\xi+N^{-1}\xi^{\prime}[\tilde{a}_{N}(\theta)-\tilde{a}_{N}(\theta^{\prime})]\right\vert \\
 & \leq & \frac{2}{N+R}\sum_{i=1}^{N+R}\sum_{j=1}^{N+R}\left\vert \tilde{a}_{ij,N}(\theta)-\tilde{a}_{ij,N}(\theta^{\prime})\right\vert (\xi_{i}^{2}+\xi_{j}^{2})/2+\frac{2}{N+R}\sum_{i=1}^{N+R}[\left\vert \tilde{a}_{i,N}(\theta)\right\vert +\left\vert \tilde{a}_{i,N}(\theta^{\prime})\right\vert ]\left\vert \xi_{i}\right\vert \\
 & \leq & \frac{1}{N+R}\sum_{i=1}^{N+R}\xi_{i}^{2}\sum_{j=1}^{N+R}\left\vert \tilde{a}_{ij,N}(\theta)-\tilde{a}_{ij,N}(\theta^{\prime})\right\vert +\frac{1}{N+R}\sum_{j=1}^{N+R}\xi_{j}^{2}\sum_{i=1}^{N+R}\left\vert \tilde{a}_{ij,N}(\theta)-\tilde{a}_{ij,N}(\theta^{\prime})\right\vert \\
 &  & +\frac{2K}{N+R}\sum_{i=1}^{N+R}\left\vert \xi_{i}\right\vert 
\end{eqnarray*}
and consequently $\left\vert N^{-1}S_{N}(\theta)-N^{-1}S_{N}(\theta^{\prime})\right\vert \leq C_{N}\left\Vert \theta-\theta^{\prime}\right\Vert $
where 
\[
C_{N}=\frac{2K}{N+R}\sum_{i=1}^{N+R}(\xi_{i}^{2}+\left\vert \xi_{i}\right\vert ).
\]
 Since $E\xi_{i}^{2}=1$, it follows that $EC_{N}\leq4KE(\xi_{i}^{2}+\left\vert \xi_{i}\right\vert )\leq8K$,
and thus $C_{N}$ is $O_{p}(1)$. Having verified all conditions of
Corollary 2.2 of \citet{newey_uniform_1991}, this concludes the proof
of the lemma.
\end{proof}
The following Lemma is helpful in proving Theorem \ref{Lemma:consistency_moment34}
later. Motivated by the proof of Theorem \ref{Lemma:consistency_moment34},
we define the following variables: $\bar{\phi}=\frac{(1-\hat{\lambda})\beta_{0}}{1-\lambda_{0}}-\hat{\beta}$,
$\bar{\varphi}=\frac{1-\hat{\lambda}}{1-\lambda_{0}}$, $\bar{z}_{r}=\frac{1}{m_{r}}\sum_{i=1}^{m_{r}}z_{ir}$,
$\ddot{\phi}_{r}=\frac{m_{r}-1+\hat{\lambda}}{m_{r}-1+\lambda_{0}}\beta_{0}-\hat{\beta}$,
$\ddot{\varphi}_{r}=1+\frac{\hat{\lambda}-\lambda_{0}}{m_{r}-1+\lambda_{0}}$,
$\ddot{z}_{ir}=z_{ir}-\bar{z}_{r}$.
\begin{lem}
\label{lem:moment34} Suppose Assumptions \ref{assume:epsilon}-\ref{assume:z}
hold, and let $\psi(.)$ be a finite positive scalar function of $m_{r}$,
$0<c_{\psi}\leqslant\psi(m_{r})\leqslant C_{\psi}<\infty$ for $r=1,...,R$.
Assume that $\hat{\lambda}\overset{p}{\rightarrow}\lambda_{0}$ and
$\hat{\beta}\overset{p}{\rightarrow}\beta_{0}$.Let $p_{1}$ and $p_{2}$
be integers such that $p_{1}\geqslant0$, $p_{2}\geqslant0$ and $p_{1}+p_{2}\leqslant4$,
then:

(a) The term $\frac{1}{R_{j}}\sum_{r=1}^{R}1(D_{r}=j)\psi(m_{r})\sum_{i=1}^{m_{r}}\ddot{u}_{ir}^{p_{1}}\bar{u}_{r}^{p_{2}}$
has a finite expected value, and its deviation from the expected value
converges in probability to zero as $R\rightarrow\infty$.

(b) For integers $0\leqslant s_{1}\leqslant p_{1}$, $0\leqslant s_{2}\leqslant p_{2}$,
and $s_{1}+s_{2}\geqslant1$, 
\[
\frac{1}{R_{j}}\sum_{r=1}^{R}1(D_{r}=j)\psi(m_{r})\sum_{i=1}^{m_{r}}(\ddot{z}_{ir}\ddot{\phi}_{r})^{s_{1}}(\ddot{\varphi}_{r}\ddot{u}_{ir})^{p_{1}-s_{1}}(\bar{z}_{r}\bar{\phi})^{s_{2}}(\bar{\varphi}\bar{u}_{r})^{p_{2}-s_{2}}\rightarrow_{p}0.
\]

(c) As $R$ goes to infinity, 
\[
\frac{1}{R_{j}}\sum_{r=1}^{R}1(D_{r}=j)\psi(m_{r})\sum_{i=1}^{m_{r}}[(\ddot{\varphi}_{r}\ddot{u}_{ir})^{p_{1}}(\bar{\varphi}\bar{u}_{r})^{p_{2}}-\ddot{u}_{ir}^{p_{1}}\bar{u}_{r}^{p_{2}}]\rightarrow_{p}0.
\]
\end{lem}
\begin{proof}
(a) With both $m_{r}$ and $1(D_{r}=j)\psi(m_{r})$ being finite and
$0\leqslant p_{1}+p_{2}\leqslant4$, Assumptions \ref{assume:epsilon}
and \ref{assume:alpha} imply that $E\left[|1(D_{r}=j)\psi(m_{r})\sum_{i=1}^{m_{r}}\ddot{u}_{ir}^{p_{1}}\bar{u}_{r}^{p_{2}}|^{1+\eta_{\mu}}\right]\leq C_{\mu}<\infty$
uniformly in $r$ for some constant $C_{\mu}$ and some $\eta_{\mu}>0$,
and that $1(D_{r}=j)\psi(m_{r})\sum_{i=1}^{m_{r}}\ddot{u}_{ir}^{p_{1}}\bar{u}_{r}^{p_{2}}$
are independently distributed across $r$. The claim thus follows
from Theorem 5.4.1 and Corollary(ii) to that theorem in \citet{chung_course_2001}. 

(b) Under Assumption \ref{assume:z}, the elements of $\bar{z}_{r}$
and $\ddot{z}_{ir}$ are uniformly bounded in absolute value by some
constants $0<C_{Z}<\infty$. Under Assumptions \ref{assume:lambda}
and \ref{assume:n}, $|\ddot{\varphi}|$ and $|\bar{\varphi}|$ are
uniformly bounded by some constants $0<C_{\varphi}<\infty$. Let $|\ddot{\phi}_{r}|_{1}$,
$|\beta_{0}|_{1}$, $|\beta_{0}-\hat{\beta}|_{1}$ be the ${\displaystyle \ell}_{1}$
norm of $\ddot{\phi}_{r}$, $\beta_{0}$ and $\beta_{0}-\hat{\beta}$
respectively. Observe that $m_{r}-1+\lambda_{0}\geqslant\epsilon_{\lambda}$
for some $\epsilon_{\lambda}>0$, and thus
\begin{align*}
|\ddot{\phi}_{r}|_{1} & =|\frac{m_{r}-1+\hat{\lambda}}{m_{r}-1+\lambda_{0}}\beta_{0}-\hat{\beta}|_{1}\\
 & =|\frac{(\hat{\lambda}-\lambda_{0})\beta_{0}}{m_{r}-1+\lambda_{0}}+(\beta_{0}-\hat{\beta})|_{1}\\
 & \leqslant|\hat{\lambda}-\lambda_{0}||\beta_{0}|_{1}\frac{1}{\epsilon_{\lambda}}+|\beta_{0}-\hat{\beta}|_{1}.
\end{align*}
Therefore, 
\begin{align*}
 & |\psi(m_{r})||\ddot{z}_{ir}\ddot{\phi}_{r}|^{s_{1}}|\ddot{\varphi}_{r}\ddot{u}_{ir}|^{p_{1}-s_{1}}|\bar{z}_{r}\bar{\phi}|^{s_{2}}|\bar{\varphi}\bar{u}_{r}|^{p_{2}-s_{2}}\\
\leqslant & C_{\psi}C_{Z}^{s_{1}+s_{2}}C_{\varphi}^{(p_{1}+p_{2}-s_{1}-s_{2})}(|\hat{\lambda}-\lambda_{0}||\beta_{0}|_{1}\frac{1}{\epsilon_{\lambda}}+|\beta_{0}-\hat{\beta}|_{1})^{s_{1}}|\bar{\phi}|^{s_{2}}|\ddot{u}_{ir}|^{p_{1}-s_{1}}|\bar{u}_{r}|^{p_{2}-s_{2}},
\end{align*}
and
\begin{align*}
 & \left|\frac{1}{R_{j}}\sum_{r=1}^{R}1(D_{r}=j)\psi(m_{r})\sum_{i=1}^{m_{r}}(\ddot{z}_{ir}\ddot{\phi}_{r})^{s_{1}}(\ddot{\varphi}_{r}\ddot{u}_{ir})^{p_{1}}(\bar{z}_{r}\bar{\phi})^{s_{2}}(\bar{\varphi}\bar{u}_{r})^{p_{2}}\right|\\
\leqslant & \frac{1}{R_{j}}\sum_{r=1}^{R}1(D_{r}=j)\sum_{i=1}^{m_{r}}|\psi(m_{r})||\ddot{z}_{ir}\ddot{\phi}_{r}|^{s_{1}}|\ddot{\varphi}_{r}\ddot{u}_{ir}|^{p_{1}-s_{1}}|\bar{z}_{r}\bar{\phi}|^{s_{2}}|\bar{\varphi}\bar{u}_{r}|^{p_{2}-s_{2}}\\
\leqslant & C_{\psi}C_{Z}^{s_{1}+s_{2}}C_{\varphi}^{(p_{1}+p_{2}-s_{1}-s_{2})}(|\hat{\lambda}-\lambda_{0}||\beta_{0}|_{1}\frac{1}{\epsilon_{\lambda}}+|\beta_{0}-\hat{\beta}|_{1})^{s_{1}}|\bar{\phi}|^{s_{2}}\frac{1}{R_{j}}\sum_{r=1}^{R}1(D_{r}=j)\sum_{i=1}^{m_{r}}|\ddot{u}_{ir}|^{p_{1}-s_{1}}|\bar{u}_{r}|^{p_{2}-s_{2}}.
\end{align*}
Note that $\bar{\phi}\rightarrow_{p}0$. With $s_{1}\geqslant0$,
$s_{2}\geqslant0$ and $s_{1}+s_{2}\geqslant1$, $(|\hat{\lambda}-\lambda_{0}||\beta_{0}|_{1}\frac{1}{\epsilon_{\lambda}}+|\beta_{0}-\hat{\beta}|_{1})^{s_{1}}|\bar{\phi}|^{s_{2}}\rightarrow_{p}0$.
By part (a) of the Lemma, $\frac{1}{R_{j}}\sum_{r=1}^{R}1(D_{r}=j)\sum_{i=1}^{m_{r}}|\ddot{u}_{ir}|^{p_{1}-s_{1}}|\bar{u}_{r}|^{p_{2}-s_{2}}$
is bounded in probability. Consequently the equation above converges
to 0 in probability.

(c) We can rewrite $\ddot{\varphi}_{r}=1+\ddot{\varsigma}_{r}$, and
$\bar{\varphi}=1+\bar{\varsigma}$, where $\ddot{\varsigma}_{r}=\frac{\hat{\lambda}-\lambda_{0}}{m_{r}-1+\lambda_{0}}$
, $\bar{\varsigma}=-\frac{\hat{\lambda}-\lambda_{0}}{1-\lambda_{0}}$.
Since $m_{r}-1+\lambda_{0}>\epsilon_{\lambda}$ for some $\epsilon_{\lambda}>0$
and $|\hat{\lambda}-\lambda_{0}|<2$, both $\ddot{\varsigma}_{r}$
and $\bar{\varsigma}$ are uniformly bounded in absolute value and
there exists some constant $0<C_{\varsigma}<\infty$ such that $|\ddot{\varsigma}_{r}|\leqslant C_{\varsigma}|\hat{\lambda}-\lambda_{0}|$
and $|\bar{\varsigma}_{r}|\leqslant C_{\varsigma}|\hat{\lambda}-\lambda_{0}|$.
Next observe that by the mean-value theorem, for $p\geqslant1$ we
have $(1+x)^{p}=1+px(1+\tilde{x})^{p-1}$ where $\tilde{x}$ lies
between $x$ and 0. The equation also holds trivially for $p=0$.
Consequently, 
\begin{align*}
|\ddot{\varphi}_{r}^{p_{1}}\bar{\varphi}^{p_{2}}-1| & =\left|\left(1+p_{1}\ddot{\varsigma}_{r}(1+\tilde{\ddot{\varsigma}})^{p_{1}-1}\right)\left(1+p_{2}\bar{\varsigma}_{r}(1+\tilde{\bar{\varsigma}})^{p_{2}-1}\right)-1\right|\\
 & =\left|p_{1}p_{2}\ddot{\varsigma}_{r}(1+\tilde{\ddot{\varsigma}})^{p_{1}-1}\bar{\varsigma}_{r}(1+\tilde{\bar{\varsigma}})^{p_{2}-1}+p_{1}\ddot{\varsigma}_{r}(1+\tilde{\ddot{\varsigma}})^{p_{1}-1}+p_{2}\bar{\varsigma}_{r}(1+\tilde{\bar{\varsigma}})^{p_{2}-1}\right|\\
 & \leqslant p_{1}p_{2}|\ddot{\varsigma}_{r}\bar{\varsigma}_{r}||(1+\tilde{\ddot{\varsigma}})^{p_{1}-1}(1+\tilde{\bar{\varsigma}})^{p_{2}-1}|+p_{1}|\ddot{\varsigma}_{r}||(1+\tilde{\ddot{\varsigma}})^{p_{1}-1}|+p_{2}|\bar{\varsigma}_{r}||(1+\tilde{\bar{\varsigma}})^{p_{2}-1}|
\end{align*}
where $\tilde{\ddot{\varsigma}}$ lies between $\ddot{\varsigma}$
and 0, $\tilde{\bar{\varsigma}}$ lies between $\bar{\varsigma}$
and 0, and thus $\tilde{\ddot{\varsigma}}$ and $\tilde{\bar{\varsigma}}$
are both uniformly bounded in absolute value. Therefore $|(1+\tilde{\ddot{\varsigma}})^{p_{1}-1}(1+\tilde{\bar{\varsigma}})^{p_{2}-1}|$,
$|(1+\tilde{\ddot{\varsigma}})^{p_{1}-1}|$ and $|(1+\tilde{\bar{\varsigma}})^{p_{2}-1}|$
are all uniformly bounded. Therefore there exists some constant $0<C_{p}<\infty$
such that
\begin{align*}
|\ddot{\varphi}_{r}^{p_{1}}\bar{\varphi}^{p_{2}}-1| & \leqslant C_{p}|\hat{\lambda}-\lambda_{0}|,
\end{align*}
and

\begin{align*}
 & \left|\frac{1}{R_{j}}\sum_{r=1}^{R}1(D_{r}=j)\psi(m_{r})\sum_{i=1}^{m_{r}}\left((\ddot{\varphi}_{r}\ddot{u}_{ir})^{p_{1}}(\bar{\varphi}\bar{u}_{r})^{p_{2}}-\ddot{u}_{ir}^{p_{1}}\bar{u}_{r}^{p_{2}}\right)\right|\\
\leqslant & \frac{1}{R_{j}}\sum_{r=1}^{R}1(D_{r}=j)\sum_{i=1}^{m_{r}}\left|\psi(m_{r})\right|\left|\ddot{\varphi}_{r}^{p_{1}}\bar{\varphi}^{p_{2}}-1\right|\left|\ddot{u}_{ir}^{p_{1}}\bar{u}_{r}^{p_{2}}\right|\\
\leqslant & C_{\varphi}C_{p}|\hat{\lambda}-\lambda_{0}|\frac{1}{R_{j}}\sum_{r=1}^{R}1(D_{r}=j)\sum_{i=1}^{m_{r}}\left|\ddot{u}_{ir}^{p_{1}}\bar{u}_{r}^{p_{2}}\right|.
\end{align*}
By part (a) of the lemma, $\frac{1}{R_{j}}\sum_{r=1}^{R}1(D_{r}=j)\sum_{i=1}^{m_{r}}\left|\ddot{u}_{ir}^{p_{1}}\bar{u}_{r}^{p_{2}}\right|$
is bounded in probability. The lemma then follows.
\end{proof}
\begin{lem}
\label{lem:Aux1}Suppose Assumption \ref{assume:z} holds. Let 
\[
A_{N}(\theta)=\diag_{r=1}^{R}\left\{ p(m_{r},\theta)I_{m_{r}}^{\ast}+s(m_{r},\theta)J_{m_{r}}^{\ast}\right\} ,
\]
 where, for $2\leq m_{r}\leq\bar{M}$, the scalar functions $p(m_{r},\theta)$
and $s(m_{r},\theta)$ are positive and continuous on the compact
parameter space $\Theta$. Let $S_{N}(\theta)=N^{-1}Z^{\prime}A_{N}(\theta)Z$,
then there exist positive constants $\underline{c}$ and $\underline{C}$
that do not depend on $\theta$ and $N$ such that
\begin{equation}
0<\underline{c}\leq\lambda_{\min}\left[S_{N}(\theta)\right]\leq\lambda_{\max}\left[S_{N}(\theta)\right]\leq\underline{C}<\infty.\label{eq:AR1}
\end{equation}
Furthermore
\begin{equation}
\sup_{\theta\in\Theta}\left\vert S_{N}(\theta)-S(\theta)\right\vert \rightarrow0\text{ as }N\rightarrow\infty,\label{eq:AR2}
\end{equation}
where $S(\theta)=\sum_{m=2}^{\bar{M}}\sum_{j=1}^{J}[p(m,\theta)\ddot{\varkappa}_{m,j}+s(m,\theta)\overline{\varkappa}_{m,j}]$.
The elements of $S(\theta)$ are continuous on $\Theta$, and 
\begin{equation}
0<\underline{c}\leq\lambda_{\min}\left[S(\theta)\right]\leq\lambda_{\max}\left[S(\theta)\right]\leq\underline{C}<\infty.\label{eq:AR3}
\end{equation}
\end{lem}
\begin{rem}
It follows from the uniform convergence of $S_{N}(\theta)$ and the
continuity of $S(\theta)$ that if $\hat{\theta}_{N}\rightarrow_{p}\theta_{0}$,
then $\left\vert S_{N}(\hat{\theta}_{N})-S(\theta_{0})\right\vert \rightarrow_{p}0\text{ as }N\rightarrow\infty$.
\end{rem}
\begin{proof}
Observe that by the Bolzano-Weierstrass' extreme value theorem there
exist positive constants $c$ and $C$, which do not depend on $\theta$,
such that 
\[
0<c\leq p(m_{r},\theta),s(m_{r},\theta)\leq C<\infty.
\]
Since $m_{r}$ only takes on finitely many values the constants $c$
and $C$ can be chosen such that the above inequalities hold for all
$m$. By Assumption \ref{assume:z} we have $0<\underline{\xi}_{Z}\leq\lambda_{\min}[N^{-1}\sum_{r\in\mathcal{I}_{m,j}}\ddot{Z}_{r}^{\prime}\ddot{Z}_{r}+N^{-1}\sum_{r\in\mathcal{I}_{m,j}}m\bar{z}_{r}^{\prime}\bar{z}_{r}]$
f\textcolor{black}{or some pair of $(m,j)$. Sin}ce the elements of
$Z$ are bounded in absolute value it follows further that there exists
a finite constant $\overline{\xi}_{Z}$ such that for all pairs of
$(m,j)$, $\lambda_{\max}[N^{-1}\sum_{r\in\mathcal{I}_{m,j}}\ddot{Z}_{r}^{\prime}\ddot{Z}_{r}]\leq\overline{\xi}_{Z}<\infty$
and $\lambda_{\max}[N^{-1}\sum_{r\in\mathcal{I}_{m,j}}m\bar{z}_{r}^{\prime}\bar{z}_{r}]\leq\overline{\xi}_{Z}<\infty$;
see, e.g., \citet{johnson_matrix_1985}, Lemma 5.6.10 and the equivalence
of matrix norms. Next observe that 
\begin{eqnarray*}
S_{N}(\theta) & = & N^{-1}\sum_{r=1}^{R}\left(p(m_{r},\theta)Z_{r}^{\prime}I_{m_{r}}^{\ast}Z_{r}+s(m_{r},\theta)Z_{r}^{\prime}J_{m_{r}}^{\ast}Z_{r}\right)\\
 & = & \sum_{m=2}^{\bar{M}}\sum_{j=1}^{J}\left(p(m,\theta)N^{-1}\sum_{r\in\mathcal{I}_{m,j}}\ddot{Z}_{r}^{\prime}\ddot{Z}_{r}+s(m,\theta)N^{-1}\sum_{r\in\mathcal{I}_{m,j}}m\bar{z}_{r}^{\prime}\bar{z}_{r}\right).
\end{eqnarray*}

Consequently
\begin{eqnarray*}
 &  & \lambda_{\min}\left[S_{N}(\theta)\right]=\inf_{\phi\in R^{k_{Z}}}\frac{\phi^{\prime}S_{N}(\theta)\phi}{\phi^{\prime}\phi}\\
 & \geq & \sum_{m=2}^{\bar{M}}\sum_{j=1}^{J}\inf_{\phi\in R^{k_{Z}}}\frac{\phi^{\prime}\left(p(m,\theta)N^{-1}\sum_{r\in\mathcal{I}_{m,j}}\ddot{Z}_{r}^{\prime}\ddot{Z}_{r}+s(m,\theta)N^{-1}\sum_{r\in\mathcal{I}_{m,j}}m\bar{z}_{r}^{\prime}\bar{z}_{r}\right)\phi}{\phi^{\prime}\phi}\\
 & \geq & c\sum_{m=2}^{\bar{M}}\sum_{j=1}^{J}\inf_{\phi\in R^{k_{Z}}}\frac{\phi^{\prime}\left(N^{-1}\sum_{r\in\mathcal{I}_{m,j}}\ddot{Z}_{r}^{\prime}\ddot{Z}_{r}+N^{-1}\sum_{r\in\mathcal{I}_{m,j}}m\bar{z}_{r}^{\prime}\bar{z}_{r}\right)\phi}{\phi^{\prime}\phi}\\
 & \geq & c\underline{\xi}_{Z}>0
\end{eqnarray*}
 and 
\begin{eqnarray*}
 &  & \lambda_{\max}\left[S_{N}(\theta)\right]\leq\sup_{\phi\in R^{k_{Z}}}\frac{\phi^{\prime}S_{N}(\theta)\phi}{\phi^{\prime}\phi}\\
 & \leq & \sum_{m=2}^{\bar{M}}\sum_{j=1}^{J}\left(p(m,\theta)\sup_{\phi\in R^{k_{Z}}}\frac{\phi^{\prime}[N^{-1}\sum_{r\in\mathcal{I}_{m,j}}\ddot{Z}_{r}^{\prime}\ddot{Z}_{r}]\phi}{\phi^{\prime}\phi}+s(m,\theta)\sup_{\phi\in R^{k_{Z}}}\frac{\phi^{\prime}[N^{-1}\sum_{r\in\mathcal{I}_{m,j}}m\bar{z}_{r}^{\prime}\bar{z}_{r}]\phi}{\phi^{\prime}\phi}\right)\\
 & \leq2 & \sum_{m=2}^{\bar{M}}\sum_{j=1}^{J}C\overline{\xi}_{Z}<\infty.
\end{eqnarray*}
 This proves the first part of the lemma. Next observe that
\begin{eqnarray*}
 &  & \sup_{\theta\in\Theta}\left\vert S_{N}(\theta)-S(\theta)\right\vert \\
 & \leq & \sup_{\theta\in\Theta}\sum_{m=2}^{\bar{M}}\sum_{j=1}^{J}\left(p(m,\theta)\left\vert N^{-1}\sum_{r\in\mathcal{I}_{m,j}}\ddot{Z}_{r}^{\prime}\ddot{Z}_{r}-\ddot{\varkappa}_{m,j}\right\vert +s(m,\theta)\left\vert N^{-1}\sum_{r\in\mathcal{I}_{m,j}}m\bar{z}_{r}^{\prime}\bar{z}_{r}-\bar{\varkappa}_{m,j}\right\vert \right)\\
 & \leq & C\sum_{m=2}^{\bar{M}}\sum_{j=1}^{J}\left(\left\vert N^{-1}\sum_{r\in\mathcal{I}_{m,j}}\ddot{Z}_{r}^{\prime}\ddot{Z}_{r}-\ddot{\varkappa}_{m,j}\right\vert +\left\vert N^{-1}\sum_{r\in\mathcal{I}_{m,j}}m\bar{z}_{r}^{\prime}\bar{z}_{r}-\bar{\varkappa}_{m,j}\right\vert \right)\rightarrow0
\end{eqnarray*}
by Assumption \ref{assume:z}. Clearly $S(\theta)$ is continuous
given the assumptions maintained w.r.t. $p(m_{r},\theta)$ and $s(m_{r},\theta)$.
Recall that by Assumption \ref{assume:z} we have $0<\underline{\xi}_{Z}\leq\lambda_{\min}[N^{-1}\sum_{r\in\mathcal{I}_{m,j}}\ddot{Z}_{r}^{\prime}\ddot{Z}_{r}+N^{-1}\sum_{r\in\mathcal{I}_{m,j}}m\bar{z}_{r}^{\prime}\bar{z}_{r}]$
\textcolor{black}{for some }pair of $(m,j)$. Therefore and since
the eigenvalues of a matrix are continuous functions of the elements
of the matrix \textcolor{black}{we have for some }pair of $(m,j)$
\begin{eqnarray*}
 & 0<\underline{\xi}_{Z} & \leq\lim_{N\rightarrow\infty}\lambda_{\min}\left[N^{-1}\sum_{r\in\mathcal{I}_{m,j}}\ddot{Z}_{r}^{\prime}\ddot{Z}_{r}+N^{-1}\sum_{r\in\mathcal{I}_{m,j}}m\bar{z}_{r}^{\prime}\bar{z}_{r}\right]\\
 &  & =\lambda_{\min}\left[\lim_{N\rightarrow\infty}N^{-1}\sum_{r\in\mathcal{I}_{m,j}}\ddot{Z}_{r}^{\prime}\ddot{Z}_{r}+\lim_{N\rightarrow\infty}N^{-1}\sum_{r\in\mathcal{I}_{m,j}}m\bar{z}_{r}^{\prime}\bar{z}_{r}\right]\\
 &  & =\lambda_{\min}\left[\ddot{\varkappa}_{m,j}+\overline{\varkappa}_{m,j}\right],
\end{eqnarray*}
\textcolor{black}{and we have for all pairs of $(m,j)$,}

\begin{eqnarray*}
 &  & \lambda_{\max}(\ddot{\varkappa}_{m,j})=\lambda_{\max}[\lim_{N\rightarrow\infty}N^{-1}\sum_{r\in\mathcal{I}_{m,j}}\ddot{Z}_{r}^{\prime}\ddot{Z}_{r}]=\lim_{N\rightarrow\infty}\lambda_{\max}[N^{-1}\sum_{r\in\mathcal{I}_{m,j}}\ddot{Z}_{r}^{\prime}\ddot{Z}_{r}]\leq\overline{\xi}_{Z}<\infty,\\
 &  & \lambda_{\max}(\overline{\varkappa}_{m,j})=\lambda_{\max}[\lim_{N\rightarrow\infty}N^{-1}\sum_{r\in\mathcal{I}_{m,j}}m\bar{z}_{r}^{\prime}\bar{z}_{r}]=\lim_{N\rightarrow\infty}\lambda_{\max}[N^{-1}\sum_{r\in\mathcal{I}_{m,j}}m\bar{z}_{r}^{\prime}\bar{z}_{r}]\leq\overline{\xi}_{Z}<\infty.
\end{eqnarray*}
The remainder of the proof of (\ref{eq:AR3}) is analogous to the
proof of (\ref{eq:AR1}).

\newpage{}
\end{proof}

\section{Proof of Lemma \ref{lem:ID_1} \label{sec:prooflemma2}}

We first consider Scenario (i) . By assumption there are two groups
$r$ and $s$ such that $m_{r}\neq m_{s}$ and $E\left[\epsilon_{ir}^{2}|m_{r},D_{r}\right]/E\left[\epsilon_{is}^{2}|m_{s},D_{s}\right]=\sigma_{\epsilon0,D_{r}}^{2}/\sigma_{\epsilon0,D_{s}}^{2}=1$
. Now consider 
\begin{align}
E\left[\chi_{r}^{w}(\theta)|m_{r},D_{r}\right] & =\frac{(m_{r}-1+\lambda)^{2}E\left[\ddot{Y}_{r}^{\prime}\ddot{Y}_{r}\right]}{(m_{r}-1)^{2}}-\sigma_{\epsilon,D_{r}}^{2}(m_{r}-1)=0.\label{eq:vare_yy_id}
\end{align}
Using $E\left[\ddot{Y}_{r}^{\prime}\ddot{Y}_{r}\right]=(m_{r}-1)^{2}/\left(m_{r}-1+\lambda_{0}\right)^{2}\sigma_{\epsilon0,D_{r}}^{2}$
gives 
\[
E\left[\chi_{r}^{w}(\theta)|m_{r},D_{r}\right]=(m_{r}-1)[\frac{(m_{r}-1+\lambda)^{2}\sigma_{\epsilon0,D_{r}}^{2}}{\left(m_{r}-1+\lambda_{0}\right)^{2}}-\sigma_{\epsilon,D_{r}}^{2}]=0
\]
which leads to the equation 
\begin{equation}
(m_{r}-1+\lambda)^{2}\sigma_{\epsilon0,D_{r}}^{2}=\sigma_{\epsilon,D_{r}}^{2}\left(m_{r}-1+\lambda_{0}\right)^{2}.\label{eq:vare_id}
\end{equation}
 Now use the moment condition for groups $r$ and $s$ and noting
that $\sigma_{\epsilon0,D_{r}}^{2}/\sigma_{\epsilon0,D_{s}}^{2}=\sigma_{\epsilon,D_{r}}^{2}/\sigma_{\epsilon,D_{s}}^{2}=1$
It follows that 
\begin{equation}
\frac{(m_{r}-1+\lambda)^{2}}{(m_{s}-1+\lambda)^{2}}=\frac{\left(m_{r}-1+\lambda_{0}\right)^{2}}{\left(m_{s}-1+\lambda_{0}\right)^{2}}.\label{eq:lambda_ID}
\end{equation}
Clearly the equation in \eqref{eq:lambda_ID} holds for $\lambda=\lambda_{0}.$
The RHS is constant in $\lambda.$ The LHS is a monotonic function
in $\lambda.$ To see this, compute the derivative $\partial h\left(\lambda\right)/\partial\lambda$
of 
\[
h\left(\lambda\right)=\frac{(m_{r}-1+\lambda)^{2}}{(m_{s}-1+\lambda)^{2}}
\]
given by 
\begin{align*}
\frac{\partial h\left(\lambda\right)}{\partial\lambda} & =\frac{2(m_{r}-1+\lambda)}{(m_{s}-1+\lambda)^{2}}-\frac{2(m_{r}-1+\lambda)^{2}(m_{s}-1+\lambda)}{(m_{s}-1+\lambda)^{4}}\\
 & =\frac{2(m_{r}-1+\lambda)(m_{s}-1+\lambda)\left(m_{s}-m_{r}\right)}{(m_{s}-1+\lambda)^{4}}.
\end{align*}
 Since $\lambda>-1$, $m_{r}>1$ and $m_{s}>1$ then $\textrm{sign}\left(\partial h\left(\lambda\right)/\partial\lambda\right)=\textrm{sign}\left(m_{s}-m_{r}\right).$
This implies that \eqref{eq:lambda_ID} can only have one solution
when $m_{r}\neq m_{s}$. Thus $E\left[\chi_{r}^{w}(\theta)|m_{q},D_{q}\right]=0$,
$q=r,s$ alone identifies $\lambda$ under Scenario (i). Plugging
$\lambda=\lambda_{0}$ into \eqref{eq:vare_id} and noting that $m_{r}-1+\lambda_{0}>0$
shows that $\sigma_{\epsilon0,D_{r}}^{2}$is identified. By assumption,
$\sigma_{\epsilon0,D_{r}}^{2}=\sigma_{\epsilon0,D_{s}}^{2}$ such
that $\sigma_{\epsilon0,D_{s}}^{2}$ is also identified. The remaining
moments $E\left[\chi_{r}^{w}(\theta)|m_{q},D_{q}\right]=0$, $q\neq r,s$
now determine the remaining parameters $\sigma_{\epsilon,j}^{2}.$

Now consider
\begin{eqnarray*}
E\left[\chi_{r}^{b}(\theta)|m_{r},D_{r}\right] & = & (1-\lambda)^{2}E\bar{y}_{r}^{2}-\sigma_{\alpha}^{2}-\frac{\sigma_{\epsilon,D_{r}}^{2}}{m_{r}}\\
 & = & \frac{\left(1-\lambda\right)^{2}}{\left(1-\lambda_{0}\right)^{2}}\left(\sigma_{\alpha,0}^{2}+\frac{\sigma_{\epsilon0,D_{r}}^{2}}{m_{r}}\right)-\left(\sigma_{\alpha}^{2}+\frac{\sigma_{\epsilon,D_{r}}^{2}}{m_{r}}\right)
\end{eqnarray*}

We already established that $E\left[\chi_{q}^{w}(\theta)|m_{q},D_{q}\right]=0$,
$q=r,s$ imply that $\lambda=\lambda_{0}$, $\sigma_{\epsilon,D_{q}}^{2}=\sigma_{\epsilon0,D_{q}}^{2}$
for $q=r,s$.Then, for $\lambda=\lambda_{0},\sigma_{\epsilon,D_{r}}^{2}=\sigma_{\epsilon0,D_{r}}^{2}$
\begin{equation}
E\left[\chi_{r}^{b}(\theta)|m_{r},D_{r}\right]=\frac{\left(1-\lambda\right)^{2}}{\left(1-\lambda_{0}\right)^{2}}\left(\sigma_{\alpha,0}^{2}+\frac{\sigma_{\epsilon0,D_{r}}^{2}}{m_{r}}\right)-\left(\sigma_{\alpha}^{2}+\frac{\sigma_{\epsilon,D_{r}}^{2}}{m_{r}}\right)=0\label{eq:sigmaalpha}
\end{equation}
reduces to $\sigma_{\alpha}^{2}=\sigma_{\alpha,0}^{2}$, and thus
also $\sigma_{\alpha,0}^{2}$ is identified. If there are additional
groups with sizes different from $m_{r}$ and $m_{s}$ then moment
conditions related to these groups constitute overidentifying restrictions.

Now consider Scenario (ii). By assumption, $m_{r}=m_{s}=m$ and $\sigma_{\epsilon0,D_{r}}^{2}\neq\sigma_{\epsilon0,D_{s}}^{2}$.
Thus 
\begin{align*}
 & E\left[\nu_{r}(\theta)|m_{r},D_{r}\right]-E\left[\nu_{s}(\theta)|m_{s},D_{s}\right]\\
= & (1-\lambda)^{2}\left(E\left[\bar{y}_{r}^{2}|m,D_{r}\right]-E\left[\bar{y}_{s}^{2}|m,D_{s}\right]\right)\\
- & (m-1+\lambda)^{2}\left(E\left[\frac{\ddot{Y}_{r}^{\prime}\ddot{Y}_{r}}{m(m-1)^{3}}|m,D_{r}\right]-E\left[\frac{\ddot{Y}_{s}^{\prime}\ddot{Y}_{s}}{m(m_{s}-1)^{3}}|m,D_{s}\right]\right).\\
= & \left[\frac{(1-\lambda)^{2}}{(1-\lambda_{0})^{2}}-\frac{(m-1+\lambda)^{2}}{\left(m-1-\lambda_{0}\right)^{2}}\right]\frac{\sigma_{\epsilon0,D_{r}}^{2}-\sigma_{\epsilon0,D_{s}}^{2}}{m}
\end{align*}
observing that, 
\begin{align*}
E\left[\bar{y}_{r}^{2}|m,D_{r}\right]-E\left[\bar{y}_{s}^{2}|m,D_{s}\right] & =\frac{\sigma_{\alpha0}^{2}+\frac{\sigma_{\epsilon0,D_{r}}^{2}}{m}}{(1-\lambda_{0})^{2}}-\frac{\sigma_{\alpha0}^{2}+\frac{\sigma_{\epsilon0,D_{s}}^{2}}{m}}{(1-\lambda_{0})^{2}}=\frac{\sigma_{\epsilon0,D_{r}}^{2}-\sigma_{\epsilon0,D_{s}}^{2}}{m(1-\lambda_{0})^{2}},
\end{align*}
and 
\[
E\left[\frac{\ddot{Y}_{r}^{\prime}\ddot{Y}_{r}}{m(m-1)^{3}}|m,D_{r}\right]-E\left[\frac{\ddot{Y}_{s}^{\prime}\ddot{Y}_{s}}{m(m-1)^{3}}|m,D_{s}\right]=\frac{1}{m\left(m-1-\lambda_{0}\right)^{2}}\left(\sigma_{\epsilon0,D_{r}}^{2}-\sigma_{\epsilon0,D_{s}}^{2}\right).
\]
Since $\sigma_{\epsilon0,D_{r}}^{2}-\sigma_{\epsilon0,D_{s}}^{2}\neq0$
it follows that $E\left[\nu_{r}(\theta)|m_{r},D_{r}\right]-E\left[\nu_{s}(\theta)|m_{s},D_{s}\right]=0$
implies
\[
\frac{(m-1+\lambda)^{2}}{(1-\lambda)^{2}}=\frac{(m-1+\lambda_{0})^{2}}{(1-\lambda_{0})^{2}}.
\]
Define $c=\frac{(m-1+\lambda_{0})^{2}}{(1-\lambda_{0})^{2}}$ , then
the equation is equivalent to
\[
(m-1+\lambda)^{2}=c\left(1-\lambda\right)^{2}
\]
which in turn is equivalent to the following polynomial in $\lambda$:,
\[
\left(m-1\right)^{2}-c+2\left(\left(m-1\right)+c\right)\lambda+(1-c)\lambda^{2}=0.
\]
Clearly, $\lambda=\lambda_{0}$ is a solution. Consider the derivative
\begin{align*}
\frac{\partial\left(\frac{(m-1+\lambda)^{2}}{(1-\lambda)^{2}}\right)}{\partial\lambda} & =\frac{2(m-1+\lambda)}{(1-\lambda)^{2}}+\frac{2(m-1+\lambda)^{2}}{(1-\lambda)^{3}}\\
 & =\frac{2\left((m-1+\lambda)\left(1-\lambda\right)+(m-1+\lambda)^{2}\right)}{(1-\lambda)^{3}}\\
 & =2(m-1+\lambda)\frac{\left(\left(1-\lambda\right)+(m-1+\lambda)\right)}{(1-\lambda)^{3}}\\
 & =\frac{2(m-1+\lambda)m}{(1-\lambda)^{3}}>0.
\end{align*}
Since $m\geqslant2$ and $\lambda\in\left(-1,1\right)$ it follows
that $(m-1+\lambda)m>0$ and $1-\lambda>0$ such that the derivative
is positive for all values of $\lambda$ on the parameter space. This
implies that $\lambda=\lambda_{0}$ is the only solution to the moment
condition. Once $\lambda$ is identified, $E\left[\nu_{r}(\theta)\right]=0$
identifies $\sigma_{\alpha}^{2}$ as
\begin{align*}
\sigma_{\alpha}^{2} & =E\left[(1-\lambda_{0})^{2}\bar{y}_{r}^{2}-\frac{\ddot{Y}_{r}^{\prime}\ddot{Y}_{r}}{m_{r}(m_{r}-1)^{3}}\right]=\sigma_{\alpha0}^{2}.
\end{align*}
Finally, note that $\nu_{r}(\theta)\text{ }$ is a function of $\chi_{r}(\theta)$,
and thus the moment conditions $E[\chi_{r}(\theta_{0})|m_{r},D_{r}]=0$
are sufficient to identify the parameter $\lambda$ and $\sigma_{\alpha}^{2}$.

Identification of the remaining parameters $\sigma_{\epsilon,j}^{2}$
follows trivially from an inspection of $\chi_{r}^{w}(\theta)$ as
once $\lambda$ is identified, $\sigma_{\epsilon,j}^{2}$ is identified
from $\chi_{r}^{w}(\theta)$ recalling that by Assumption \ref{assume:epsilon}
there exists some $r$ such that $D_{r}=j$.

\section{The CV estimator of \citet{graham_identifying_2008}\label{sec:cv}}

In this appendix we interpret the CV estimator of \citet{graham_identifying_2008}
as based on moment conditions developed in Section \ref{sec:graham}.
Specifically, we show that the identification results in \citet{graham_identifying_2008}
can be seen as an adapted version of Scenario (ii) of Lemma \ref{lem:ID_1}.

The peer effects model of \citet{graham_identifying_2008} can be
written as 
\begin{equation}
y_{ir}=v_{r}+\epsilon_{ir}+\left(\gamma-1\right)\bar{\epsilon}_{r},\label{eq:Graham_Model}
\end{equation}
where $\bar{\epsilon}_{r}=m_{r}^{-1}\sum_{i=1}^{m_{r}}\epsilon_{ir}$
is the group average of unobserved characteristics. The parameter
$\gamma$ captures the peer effect. Taking group averages on both
sides of \eqref{eq:Graham_Model}, we get $\bar{\epsilon}_{r}=\frac{1}{\gamma}(\bar{y}_{r}-v_{r})$.
Plugging back into \eqref{eq:Graham_Model}, and letting $\tilde{\lambda}=1-1/\gamma$
as well as $\alpha_{r}=v_{r}/\gamma$, we get the following structural
model 
\begin{equation}
y_{ir}=\tilde{\lambda}\bar{y}_{r}+\alpha_{r}+\epsilon_{ir}.\label{eq:full-mean}
\end{equation}
The specification differs from our main model in \eqref{eq:Graham_Mod_L-M-O}
in that it uses the full group mean $\bar{y}_{r}$ rather than the
leave-out-mean $\bar{y}_{(-i)r}$. The leave-out-mean specifications
are often preferred in the literature, see for example \citet{angrist_perils_2014}
for a discussion. Defining $\tilde{W}_{m_{r}}=\frac{1}{m_{r}}\iota_{m_{r}}\iota_{m_{r}}^{\prime}$,
we can rewrite \eqref{eq:full-mean} in matrix form as
\[
Y_{r}=\tilde{\lambda}\tilde{W}_{m_{r}}Y_{r}+\alpha\iota_{m_{r}}+\epsilon_{r}.
\]
Using the same notation as in Section \ref{sec:graham}, it can be
shown that $\bar{y}_{r}=\bar{u}_{r}/(1-\tilde{\lambda})$ with $\bar{u}_{r}=\alpha_{r}+\bar{\epsilon}_{r}$,
and $\ddot{Y}_{r}=\ddot{U}_{r}=\ddot{\epsilon}_{r}$. Note that in
the context of Model (\ref{eq:full-mean}) the results in \citet{manski_identification_1993},
\citet{kelejian_estimation_2006} or \citet{bramoulle_identification_2009}
show that $\tilde{\lambda}$ cannot be identified by instrumenting
$\tilde{W}_{m_{r}}Y_{r}$ with $\tilde{W}_{m_{r}}^{2}Z_{r}$ when
$\tilde{W}_{m_{r}}Z_{r}$ is included as a covariate, observing that
$\tilde{W}_{m_{r}}^{2}=\tilde{W}_{m_{r}}$.

To isolate or identify the social interaction effect \citet{graham_identifying_2008}
imposes restrictions on the unobservables and group size. Graham considers
the case when $J=2$ and $D_{r}$ is a categorical variable for small/regular
classes, which are coded as $D_{r}=1$ whenever $m_{r}\geqslant\bar{m}$
for some constant $\bar{m}$ and $D_{r}=2$ otherwise. Assumptions
1.1 and 1.2 in \citet{graham_identifying_2008} amount to Assumptions
1 and 2 in Section \ref{sec:graham}. Assumption 1.3 in \citet{graham_identifying_2008}
imposes that $E\left[\ddot{Y}_{r}^{\prime}\ddot{Y}_{r}|D_{r}=1\right]\neq E\left[\ddot{Y}_{r}^{\prime}\ddot{Y}_{r}|D_{r}=2\right].$
Observing that $E\left[\ddot{Y}_{r}^{\prime}\ddot{Y}_{r}|D_{r}=d\right]=\sigma_{\epsilon,d}^{2}(m_{r}-1)$
the latter condition is seen to be satisfied if Scenario (ii) in Lemma
\ref{lem:ID_1} holds true. In addition, the condition is also true
if $m_{r}$ varies and the idiosyncratic errors $\epsilon_{ir}$ are
homoscedastic as in Scenario (i).

The parameter $\gamma^{2}=1/(1-\tilde{\lambda})^{2}$ is identified
under Assumptions 1.1-1.3, see Proposition 1.1 in \citet{graham_identifying_2008}.
Below we prove the proposition by adapting the proof for Scenario
(ii) in Lemma \ref{lem:ID_1} to the full-mean specification, thus
verifying that the CV estimator is a special case of the moment based
estimators studied in Section \ref{sec:graham}.

Under the full-mean specification of \citet{graham_identifying_2008},
our moment conditions in \eqref{eq:mwb} change correspondingly to

\begin{equation}
\tilde{\chi}_{r}^{w}(\theta)=\ddot{Y}_{r}^{\prime}\ddot{Y}_{r}-(m_{r}-1)\sigma_{\epsilon,D_{r}}^{2},\label{eq:mw-2}
\end{equation}
\begin{equation}
\tilde{\chi}_{r}^{b}(\theta)=(1-\tilde{\lambda})^{2}\bar{y}_{r}^{2}-\sigma_{\alpha}^{2}-\frac{\sigma_{\epsilon,D_{r}}^{2}}{m_{r}}.\label{eq:mb-2}
\end{equation}
 The combined moment condition in \eqref{eq:nu} is now 
\begin{equation}
\tilde{\nu}_{r}(\tilde{\lambda},\sigma_{\alpha}^{2})=\tilde{\chi}_{r}^{b}(\theta)-\frac{\tilde{\chi}_{r}^{w}(\theta)}{m_{r}(m_{r}-1)}=(1-\tilde{\lambda})^{2}\bar{y}_{r}^{2}-\sigma_{\alpha}^{2}-\frac{\ddot{Y}_{r}^{\prime}\ddot{Y}_{r}}{m_{r}(m_{r}-1)}.\label{eq:momgraham}
\end{equation}
Note that by design $E(\tilde{\nu}_{r}(\tilde{\lambda},\sigma_{\alpha}^{2})|m_{r},D_{r})=0$
at the true parameter vector $\theta_{0}$.

The restriction that group effect variances are homoscedastic can
be exploited by taking differences $E\left[\nu_{r}(\tilde{\lambda},\sigma_{\alpha}^{2})|m_{r},D_{r}=1\right]-E\left[\nu_{r}(\tilde{\lambda},\sigma_{\alpha}^{2})|m_{r},D_{r}=2\right]$
to eliminate $\sigma_{\alpha}^{2}$. This implies the following population
equation for $\left(1-\tilde{\lambda}\right)^{2}$
\begin{equation}
\frac{1}{(1-\tilde{\lambda})^{2}}=\frac{E\left[\bar{y}_{r}^{2}|m_{r},D_{r}=1\right]-E\left[\bar{y}_{r}^{2}|m_{r},D_{r}=2\right]}{E\left[\frac{\ddot{Y}_{r}^{\prime}\ddot{Y}_{r}}{m_{r}(m_{r}-1)}|m_{r},D_{r}=1\right]-E\left[\frac{\ddot{Y}_{r}^{\prime}\ddot{Y}_{r}}{m_{r}(m_{r}-1)}|m_{r},D_{r}=2\right]},\label{eq:wald}
\end{equation}
which is a modified version of \eqref{eq:Wald_L-O-M}. The Wald estimate
for $\tilde{\lambda}$ can then be calculated from the sample analog
of the right-hand side above. Two points are worth noting. First,
identification is possible under Scenario (ii) of Lemma \ref{lem:ID_1},
i.e., when there exists some $m_{r}=m_{s}$ and $\sigma_{\epsilon,D_{r}}^{2}\neq\sigma_{\epsilon,D_{s}}^{2}$
. This confirms the applicability of Lemma \ref{lem:ID_1} to the
full-mean specification. Second, due to the full-mean specification,
identification is also possible even when $\sigma_{\epsilon,1}^{2}=\sigma_{\epsilon,2}^{2}$
as long as there is variation in $m_{r}.$ So, in the case of homoscedasticity,
group size variation alone is enough for the CV estimator to identify
$\tilde{\lambda}$. It can be shown that the score of a Gaussian likelihood
estimator is a function of $\tilde{\chi}_{r}(\theta)=\left(\tilde{\chi}_{r}^{w}(\theta),\tilde{\chi}_{r}^{b}(\theta)\right)^{\prime}$
and therefore the Gaussian maximum likelihood estimator shares the
same identification properties.

It is well known that identification for the case of peer effects
captured by full group means is difficult, see \citet{manski_identification_1993},
\citet{bramoulle_identification_2009} or \citet{angrist_perils_2014}.
In the case of the conditional variance restrictions or likelihood
approaches considered here, this manifests itself in the fact that
$\left(1-\tilde{\lambda}\right)^{2}$ but not $\tilde{\lambda}$ is
identified without additional constraints on the parameter space.
An inspection of \eqref{eq:wald} shows that while $\gamma^{2}=1/\left(1-\tilde{\lambda}\right)^{2}$
is identified, the sign of $\gamma=1/\left(1-\tilde{\lambda}\right)$
is not identified, unless $\tilde{\lambda}$ is constrained to take
values in $\left(-\infty,1\right)$. The reason is that the function
$1/\left(1-\tilde{\lambda}\right)^{2}$ is monotonically increasing
on the interval $\left(-\infty,1\right)$ and $\left(1-\tilde{\lambda}\right)>0$
for $\tilde{\lambda}\in\left(-\infty,1\right).$ The implied range
for $\gamma$ then is $\left(0,\infty\right)$ and the permissible
parametrizations of the term $\bar{\epsilon}_{r}$ in (\ref{eq:Graham_Model})
is $\left(-1,\infty\right).$ For the latter, the positive values
are most relevant in peer effects applications. 

\newpage{}

\section{\label{app:Proofthm}Proofs of Theorems \ref{theorem:Consistency},
\ref{theorem:AsymptoticNormlity}, and \ref{thm:ValidInference}}

In this section, we collect the proof of Theorem \ref{theorem:Consistency}
in Sections \ref{sec:Proof-of-cons} and \ref{sec:Proof-of-cons-2}.
This theorem establishes the identification and the consistency of
the quasi-maximum likelihood estimator. Then we provide a proof of
Theorem \ref{theorem:AsymptoticNormlity} and Theorem \ref{thm:ValidInference}
in Section \ref{sec:Proof-of-asydis} and Section \ref{sec:Proof-consistency34},
respectively. These theorems establish the asymptotic distribution
of the QMLE and the consistency of our estimators for the third and
fourth moments.

\subsection{Proof of Theorem \ref{theorem:Consistency}(a) \label{sec:Proof-of-cons}}

For the un-concentrated log likelihood function in \eqref{eq:logL},
let 
\[
\bar{R}(\theta,\beta)=\lim_{N\rightarrow\infty}E\left[\frac{1}{N}\textrm{ln}L(\theta,\beta)\right].
\]
Let $\bar{\beta}(\theta)$ be the maximizer of $\bar{R}(\theta,\beta)$
with respect to $\beta$,

\[
\bar{R}(\theta,\bar{\beta}(\theta))=\max_{\beta}\bar{R}(\theta,\beta),
\]
and let
\[
\bar{Q}^{\ast\ast}(\theta)=\bar{R}(\theta,\bar{\beta}(\theta)).
\]
For the concentrated log likelihood function in \eqref{eq:QN}, let
$\bar{Q}^{\ast}(\theta)=\lim_{N\rightarrow\infty}E\left[Q_{N}(\theta)\right]$.
To prove that $\theta_{0}$ is identifiably unique, it suffices to
show that Condition \ref{cond:lim}(a) and Condition \ref{cond:ind}
below hold; cp., e.g., Definition 3.1 of identifiable uniqueness and
the subsequent discussion in \citet{potscher_basic_1991}. In fact,
under Condition \ref{cond:lim}(a) and Condition \ref{cond:ind} the
identifiable uniqueness of the parameter vector $\theta_{0}$ is equivalent
with $\theta_{0}$ being asymptotically identified in the sense that
it is the unique maximizer of $\bar{Q}^{\ast}(\theta)$.
\begin{condition}
\label{cond:lim}(a) The non-stochastic functions $\bar{Q}^{\ast}(\theta)$
and $\bar{Q}^{\ast\ast}(\theta)$ exist, and $\bar{Q}^{\ast}(\theta)=\bar{Q}^{\ast\ast}(\theta)$
are continuous and finite ;

(b) As $N$ goes to infinity, $\sup_{\theta\in\Theta}\left|E\left[Q_{N}(\theta)\right]-\bar{Q}^{\ast}(\theta)\right|\rightarrow0$.
\end{condition}
\begin{condition}
\label{cond:ind}The parameter space $\Theta$ is compact, the true
value $\theta_{0}$ is the unique maximizer of $\bar{Q}^{\ast}(\theta)$
(and hence $\bar{Q}^{\ast\ast}(\theta)$) on $\Theta$ and $\bar{\beta}(\theta_{0})=\beta_{0}$.
\end{condition}
Condition \ref{cond:lim}(b) is used for the proof of consistency
that is presented in Section \ref{sec:Proof-of-cons-2} below. Given
Condition \ref{cond:lim}(b) and the identifiable uniqueness of the
true parameter vector consistency of the QMLE follows immediately
from, e.g., Lemma 3.1 in Poetscher and Prucha (1997), p. 16.

We combine conditions \ref{cond:lim}(a) and \ref{cond:lim}(b) as
they can be established together.

$\blacksquare$ \textbf{Verification of Condition} \ref{cond:lim}:
To prove that $\bar{Q}^{\ast}(\theta)=\lim_{N\rightarrow\infty}E\left[Q_{N}(\theta)\right]$
exists, it is readily seen that 
\begin{eqnarray}
E\left[Q_{N}(\theta)\right] & = & -\frac{\ln(2\pi)}{2}+\frac{1}{2N}\ln|(I-\lambda W)^{2}\Omega(\theta)^{-1}|\label{AppCon4}\\
 &  & -\frac{1}{2N}\tr\left\{ (I-\lambda W)^{\prime}M_{Z}(\theta)(I-\lambda W)(E\left[YY^{\prime}\right])\right\} \nonumber \\
 & = & -\frac{\ln(2\pi)}{2}+\frac{1}{2N}\ln|(I-\lambda W)^{2}\Omega(\theta)^{-1}|\nonumber \\
 &  & -\frac{1}{2N}\tr\left[(I-\lambda W)^{\prime}M_{Z}(\theta)(I-\lambda W)(I-\lambda_{0}W)^{-1}\left(\Omega_{0}+Z\beta_{0}\beta_{0}^{\prime}Z^{\prime}\right)(I-\lambda_{0}W)^{-1}\right]\nonumber \\
 & = & \bar{Q}_{N}^{(1)}(\theta)+\bar{Q}_{N}^{(2)}(\theta)+\bar{Q}_{N}^{(3)}(\theta),\nonumber 
\end{eqnarray}
with
\begin{eqnarray}
\bar{Q}_{N}^{(1)}(\theta) & = & -\frac{\textrm{ln}(2\pi)}{2}+\frac{1}{2N}\ln|(I-\lambda W)^{2}\Omega(\theta)^{-1}|-\frac{1}{2N}\tr\left[(I-\lambda_{0}W)^{-2}(I-\lambda W)^{2}\Omega(\theta)^{-1}\Omega_{0}\right],\label{AppCon5}\\
\bar{Q}_{N}^{(2)}(\theta) & = & -\frac{1}{2N}\tr\left[\beta_{0}^{\prime}Z^{\prime}(I-\lambda_{0}W)^{-1}(I-\lambda W)M_{Z}(\theta)(I-\lambda W)(I-\lambda_{0}W)^{-1}Z\beta_{0}\right],\nonumber \\
\bar{Q}_{N}^{(3)}(\theta) & = & \frac{1}{2N}\tr\left[(I-\lambda_{0}W)^{-2}(I-\lambda W)^{2}\Omega(\theta)^{-1}Z(Z^{\prime}\Omega(\theta)^{-1}Z)^{-1}Z^{\prime}\Omega(\theta)^{-1}\Omega_{0}\right],\nonumber 
\end{eqnarray}
recalling that $M_{Z}(\theta)=\Omega(\theta)^{-1}-\Omega(\theta)^{-1}Z(Z^{\prime}\Omega(\theta)^{-1}Z)^{-1}Z^{\prime}\Omega(\theta)^{-1}$
and that the matrices $(I-\lambda_{0}W),$$(I-\lambda W)$, $\Omega(\theta)^{-1}$
and $\Omega_{0}$ all commute.

We show that the limits of $\bar{Q}_{N}^{(1)}(\theta)$, $\bar{Q}_{N}^{(2)}(\theta)$
and $\bar{Q}_{N}^{(3)}(\theta)$ exist, in reverse order. Observe
that
\[
\bar{Q}_{N}^{(3)}(\theta)=\frac{1}{2N}\tr\left[\left(\frac{1}{N}Z^{\prime}\Omega(\theta)^{-1}\Omega_{0}(I-\lambda_{0}W)^{-2}(I-\lambda W)^{2}\Omega(\theta)^{-1}Z\right)\left(\frac{1}{N}Z^{\prime}\Omega(\theta)^{-1}Z\right)^{-1}\right].
\]
Both matrices in square brackets are of the form considered in (\ref{AppCon3})
with $p(m_{r},D_{r},\theta)$ and $s(m_{r},D_{r},\theta)$ satisfying
the assumptions of Lemma \ref{lem:Aux1}. Thus their elements, and
in turn the trace, are bounded in absolute value by respective constants
that do not depend on $\theta$ and $N$. Consequently $\sup_{\theta\in\Theta}\bar{Q}_{N}^{(3)}(\theta)\leq\textrm{const}/N\rightarrow0$
as $N\rightarrow\infty$ and $\lim_{N\rightarrow\infty}\bar{Q}_{N}^{(3)}(\theta)=0$.

Second, observe that
\begin{align}
2\bar{Q}_{N}^{(2)}(\theta) & =\beta_{0}^{\prime}\left(\frac{1}{N}Z^{\prime}(I-\lambda_{0}W)^{-2}(I-\lambda W)^{2}\Omega(\theta)^{-1}Z\right)\beta_{0}\nonumber \\
 & -\beta_{0}^{\prime}\left\{ \left(\frac{1}{N}Z^{\prime}(I-\lambda_{0}W)^{-1}(I-\lambda W)\Omega(\theta)^{-1}Z\right)\left(\frac{1}{N}Z^{\prime}\Omega(\theta)^{-1}Z\right)^{-1}\right.\nonumber \\
 & \times\left.\left(\frac{1}{N}Z^{\prime}\Omega(\theta)^{-1}(I-\lambda W)(I-\lambda_{0}W)^{-1}Z\right)\right\} \beta_{0}.\label{eq:Qr2-1}
\end{align}
In light of (\ref{AppCon1}) and (\ref{AppCon2}) and using Lemma
\ref{lem:Aux1} we see that $\sup_{\theta\in\Theta}|\bar{Q}_{N}^{(2)}(\theta)-\bar{Q}^{(2)\ast}(\theta)|\rightarrow_{p}0$,
where
\begin{eqnarray}
\bar{Q}^{(2)\ast}(\theta) & = & \frac{1}{2}\beta_{0}^{\prime}\left(\varUpsilon_{1}(\theta)-\varUpsilon_{2}(\theta)\Upsilon_{3}^{-1}(\theta)\varUpsilon_{2}(\theta)\right)\beta_{0},\nonumber \\
\varUpsilon_{1}(\theta) & = & \sum_{j=1}^{J}\sum_{m=2}^{\bar{M}}\left(\frac{\phi_{W}^{2}(m,\theta)}{\phi_{W}^{2}(m,\theta_{0})\phi_{\Omega}(m,j,\theta)}\ddot{\varkappa}_{m,j}+\frac{\psi_{W}^{2}(m,\theta)}{\psi_{W}^{2}(m,\theta_{0})\psi_{\Omega}(m,j,\theta)}\bar{\varkappa}_{m,j}\right),\label{eq:ups1}\\
\varUpsilon_{2}(\theta) & = & \sum_{j=1}^{J}\sum_{m=2}^{\bar{M}}\left(\frac{\phi_{W}(m,\theta)}{\phi_{W}(m,\theta_{0})\phi_{\Omega}(m,j,\theta)}\ddot{\varkappa}_{m,j}+\frac{\psi_{W}(m,\theta)}{\psi_{W}(m,\theta_{0})\psi_{\Omega}(m,j,\theta)}\bar{\varkappa}_{m,j}\right),\label{eq:ups2}\\
\varUpsilon_{3}(\theta) & = & \sum_{j=1}^{J}\sum_{m=2}^{\bar{M}}\left(\frac{1}{\phi_{\Omega}(m,j,\theta)}\ddot{\varkappa}_{m,j}+\frac{1}{\psi_{\Omega}(m,j,\theta)}\bar{\varkappa}_{m,j}\right),\label{eq:ups3}
\end{eqnarray}
and where $\bar{Q}^{(2)\ast}(\theta)$ is finite and continuous on
$\Theta$ by Lemma \ref{lem:Aux1}.

Third, observe that
\begin{eqnarray*}
\bar{Q}_{N}^{(1)}(\theta) & = & -\frac{\textrm{ln}(2\pi)}{2}-\frac{1}{2N}\textrm{ln}|(I-\lambda_{0}W)^{-2}\Omega_{0}|\\
 &  & +\frac{1}{2N}\textrm{ln}|(I-\lambda_{0}W)^{-2}(I-\lambda W)^{2}\Omega(\theta)^{-1}\Omega_{0}|-\frac{1}{2N}\tr\left[(I-\lambda_{0}W)^{-2}(I-\lambda W)^{2}\Omega(\theta)^{-1}\Omega_{0}\right]\\
 & = & C_{N}+\frac{1}{2}\sum_{j=1}^{J}\sum_{m=2}^{\bar{M}}\frac{R_{m,j}}{N}\textrm{ln}|G(m,j,\theta)|-\frac{1}{2}\sum_{j=1}^{J}\sum_{m=2}^{\bar{M}}\frac{R_{m,j}}{N}\tr[G(m,j,\theta)]
\end{eqnarray*}
with
\begin{eqnarray}
G(m,j,\theta) & = & (I_{m}-\lambda W_{m})^{2}\Omega_{m,j}(\theta)^{-1}(I_{m}-\lambda_{0}W_{m})^{-2}\Omega_{m,j0}\nonumber \\
 & = & \frac{\phi_{W}^{2}(m,\theta)\phi_{\Omega}(m,j,\theta_{0})}{\phi_{W}^{2}(m,\theta_{0})\phi_{\Omega}(m,j,\theta)}I_{m}^{\ast}+\frac{\psi_{W}^{2}(m,\theta)\psi_{\Omega}(m,j,\theta_{0})}{\psi_{W}^{2}(m,\theta_{0})\psi_{\Omega}(m,j,\theta)}J_{m}^{\ast},\nonumber \\
 & = & \frac{\sigma_{\epsilon0,j}^{2}}{\sigma_{\epsilon,j}^{2}}(\frac{m-1+\lambda}{m-1+\lambda_{0}})^{2}I_{m}^{\ast}+\frac{(\sigma_{\epsilon0,j}^{2}+m\sigma_{\alpha0}^{2})}{(\sigma_{\epsilon,j}^{2}+m\sigma_{\alpha}^{2})}(\frac{1-\lambda}{1-\lambda_{0}})^{2}J_{m}^{\ast},\label{eq:Gn}
\end{eqnarray}
and $C_{N}=-\frac{\textrm{ln}(2\pi)}{2}-\frac{1}{2}\sum_{j=1}^{J}\sum_{m=2}^{\bar{M}}\frac{R_{m,j}}{N}\textrm{ln}|(I-\lambda_{0}W_{m})^{-2}\Omega_{m,j0}|$.
Under Assumption \ref{assume:n}, $R_{m,j}/N\rightarrow\omega_{m,j}^{\ast}/m^{\ast}$.
Let $C^{\ast}=\lim_{N\rightarrow\infty}C{}_{N}$ and 
\begin{equation}
\bar{Q}^{(1)\ast}(\theta)=C^{\ast}+\frac{1}{2m^{\ast}}\sum_{j=1}^{J}\sum_{m=2}^{\bar{M}}\omega_{m,j}^{\ast}g(m,j,\theta)\label{eq:EQbar1}
\end{equation}
with
\[
g(m,j,\theta)=\ln|G(m,j,\theta)|-\tr[G(m,j,\theta)],
\]
then clearly $\sup_{\theta\in\Theta}|\bar{Q}_{N}^{(1)}(\theta)-\bar{Q}^{(1)\ast}(\theta)|\rightarrow0$
with $\bar{Q}^{(1)\ast}(\theta)$ finite and continuous on $\Theta$.
In all, $\sup_{\theta\in\Theta}|E\left[Q_{N}(\theta)\right]-\bar{Q}^{\ast}(\theta)|\rightarrow0$,
where
\[
\bar{Q}^{\ast}(\theta)=\bar{Q}^{(1)\ast}(\theta)+\bar{Q}^{(2)\ast}(\theta)=C^{\ast}+\frac{1}{2m^{\ast}}\sum_{j=1}^{J}\sum_{m=2}^{\bar{M}}\omega_{m,j}^{\ast}g(m,j,\theta)+\bar{Q}^{(2)\ast}(\theta),
\]
and where $\bar{Q}^{\ast}(\theta)$ is continuous and finite.

For the un-concentrated likelihood function, 
\begin{align*}
\bar{R}(\theta,\beta) & =\lim_{N\rightarrow\infty}\frac{1}{N}E\left[\ln L_{N}(\theta,\beta)\right]\\
 & =\lim_{N\rightarrow\infty}\frac{1}{N}\left\{ -\frac{N}{2}\ln(2\pi)+\frac{1}{2}\ln|(I-\lambda W)^{2}\Omega(\theta)^{-1}|\right.\\
 & -\frac{1}{2}\tr\left[(I-\lambda_{0}W)^{-2}(I-\lambda W)^{2}\Omega(\theta)^{-1}\Omega_{0}\right]\\
 & \left.-\frac{1}{2}\left((I-\lambda W)(I-\lambda_{0}W)^{-1}Z\beta_{0}-Z\beta\right)^{\prime}\Omega(\theta)^{-1}\left((I-\lambda W)(I-\lambda_{0}W)^{-1}Z\beta_{0}-Z\beta\right)\right\} \\
 & =\lim_{N\rightarrow\infty}\bar{Q}_{N}^{(1)}(\theta)-\lim_{N\rightarrow\infty}\frac{1}{2N}\left(\beta_{0}^{\prime}Z^{\prime}(I-\lambda W)^{2}(I-\lambda_{0}W)^{2}\Omega(\theta)^{-1}Z\beta_{0}\right)\\
 & +\lim_{N\rightarrow\infty}\frac{1}{N}\beta_{0}^{\prime}Z^{\prime}(I-\lambda W)(I-\lambda_{0}W)\Omega(\theta)^{-1}Z\beta-\lim_{N\rightarrow\infty}\frac{1}{2N}\beta^{\prime}Z^{\prime}\Omega(\theta)^{-1}Z\beta\\
 & =\bar{Q}^{(1)\ast}(\theta)-\frac{1}{2}\beta_{0}^{\prime}\varUpsilon_{1}(\theta)\beta_{0}+\beta_{0}^{\prime}\varUpsilon_{2}(\theta)\beta-\frac{1}{2}\beta^{\prime}\varUpsilon_{3}(\theta)\beta,
\end{align*}
where $\varUpsilon_{1}(\theta)$, $\varUpsilon_{2}(\theta)$, $\varUpsilon_{3}(\theta)$
are defined in Equations \eqref{eq:ups1},\eqref{eq:ups2}, and\eqref{eq:ups3}.
Taking the derivative of $\bar{R}(\theta,\beta)$ with respect to
$\beta$, 
\[
\frac{\partial\bar{R}(\theta,\bar{\beta})}{\partial\beta}=\beta_{0}^{\prime}\varUpsilon_{2}(\theta)-\bar{\beta}^{\prime}\varUpsilon_{3}(\theta)=0.
\]
Since $\varUpsilon_{3}(\theta)$ is non-singular by Assumption \ref{assume:z}
and Lemma \ref{lem:Aux1}, 
\begin{equation}
\bar{\beta}(\theta)=\varUpsilon_{3}(\theta)^{-1}\varUpsilon_{2}(\theta)\beta_{0}.\label{eq:bargamma}
\end{equation}
Let $\bar{Q}^{\ast\ast}(\theta)=\bar{R}(\theta,\bar{\beta}(\theta))$
and plug $\bar{\beta}(\theta)$ above back into $\bar{R}(\theta,\beta)$,
\begin{align*}
\bar{Q}^{\ast\ast}(\theta) & =\bar{Q}^{(1)\ast}(\theta)-\frac{1}{2}\beta_{0}^{\prime}\left(\varUpsilon_{1}(\theta)-\varUpsilon_{2}(\theta)\varUpsilon_{3}^{-1}(\theta)\varUpsilon_{2}(\theta)\right)\beta_{0}\\
 & =\bar{Q}^{(1)\ast}(\theta)+\bar{Q}^{(2)\ast}(\theta)=\bar{Q}^{\ast}(\theta).
\end{align*}
Note that the second order derivative 
\[
\frac{\partial^{2}\bar{R}(\theta,\beta)}{\partial\beta\partial\beta^{\prime}}=-\varUpsilon_{3}(\theta)=-\lim_{N\rightarrow\infty}Z^{\prime}\Omega(\theta)^{-1}Z
\]
is negative definite by Assumption \ref{assume:z} and Lemma \ref{lem:Aux1}
uniformly in $\theta$, thus $\bar{\beta}(\theta)$ is the unique
maximizer of $\bar{R}(\theta,\beta)$ over $\beta$. In all, we have
$\bar{Q}^{\ast}(\theta)$ and $\bar{Q}^{\ast\ast}(\theta)$ both exist
and $\bar{Q}^{\ast}(\theta)=\bar{Q}^{\ast\ast}(\theta)$.

\noindent $\blacksquare$ \textbf{Verification of Condition \ref{cond:ind}}
Since $\varUpsilon_{3}(\theta_{0})=\varUpsilon_{2}(\theta_{0})$,
$\bar{\beta}(\theta_{0})=\beta_{0}$ is readily seen. Next we show
that $\theta_{0}$ is the unique global maximizer of $\bar{Q}^{\ast}(\theta)$
on $\Theta$. We first show that $\theta_{0}$ is a global maximizer
of $\bar{Q}^{(2)\ast}(\theta)$. To see this observe that we can rewrite
$Q_{N}^{(2)}(\theta)$ as $Q_{N}^{(2)}(\theta)=-\frac{1}{2N}\tilde{\eta}_{Z}(\theta)^{\prime}\tilde{M}_{Z}(\theta)\tilde{\eta}_{Z}(\theta)$,
where $\tilde{M}_{Z}(\theta)=I-\Omega^{-1/2}Z^{\prime}(Z^{\prime}\Omega(\theta)Z)^{-1}Z\Omega^{-1/2}$
is idempotent and positive semidefinite, and $\tilde{\eta}_{Z}(\theta)=\Omega(\theta)^{-1/2}(I-\lambda W)(I-\lambda_{0}W)^{-1}Z\beta_{0}$.
Thus $Q_{N}^{(2)}(\theta)\leq0$ and consequently also $\bar{Q}^{(2)\ast}(\theta)\leq0$.
Next observe, as is readily checked, that $Q^{(2)\ast}(\theta_{0})=0$.
Therefore $\bar{Q}^{(2)\ast}(\theta)\leqslant\bar{Q}^{(2)\ast}(\theta_{0})$
for all $\theta\in\Theta$.

\noindent To show that $\theta_{0}$ is the unique global maximizer
of $\bar{Q}^{\ast}(\theta)$ it thus suffices to show that $\theta_{0}$
is the unique maximizer of $\sum_{j=1}^{J}\sum_{m=2}^{\bar{M}}\omega_{m,j}^{\ast}g(m,j,\theta)$.
Observe that
\begin{equation}
\sum_{j=1}^{J}\sum_{m=2}^{\bar{M}}\omega_{m,j}^{\ast}g(m,j,\theta)=\sum_{j=1}^{J}\sum_{m=2}^{\bar{M}}\omega_{m,j}^{\ast}\left(\ln|G(m,j,\theta)|-\tr[G(m,j,\theta)]\right)\leqslant-m^{\ast},\label{eq:uniq1}
\end{equation}

\noindent where $m^{\ast}=\sum_{j=1}^{J}\sum_{m=2}^{\bar{M}}\omega_{m,j}^{\ast}m$.\footnote{See Footnote \ref{footnote:max} for details.}
The equality holds if and only if $g(m,j,\theta)=-1$ or equivalently
$G(m,j,\theta)=I_{m}$ for all $m$ and $j$ with $\omega_{m,j}^{\ast}>0.$
Under the two scenarios described by Assumption \ref{assume:id} this
is the case if and only if $\theta=\theta_{0}$, which establishes
that $\theta_{0}$ is the unique maximizer of $\sum_{j=1}^{J}\sum_{m=2}^{\bar{M}}\omega_{m,j}^{\ast}g(m,j,\theta)$.
To see this, observe that in light of (\ref{eq:Gn}) the equality
$G(m,j,\theta)=I_{m}$ only holds if 
\begin{eqnarray}
\left(\frac{m-1+\lambda}{m-1+\lambda_{0}}\right)^{2}\frac{\sigma_{\epsilon0,j}^{2}}{\sigma_{\epsilon,j}^{2}} & = & 1,\label{eq:con1}\\
\frac{(\sigma_{\epsilon0,j}^{2}+m\sigma_{\alpha0}^{2})}{(\sigma_{\epsilon,j}^{2}+m\sigma_{\alpha}^{2})}\left(\frac{1-\lambda}{1-\lambda_{0}}\right)^{2} & = & 1.\label{eq:con2}
\end{eqnarray}
Note that \eqref{eq:con1} and \eqref{eq:con2} are equivalent to
$E\left[\chi_{r}^{w}(\theta)|m_{r}=m,D_{r}=j\right]=0$ and $E\left[\chi_{r}^{b}(\theta)|m_{r}=m,D_{r}=j\right]=0$
respectively, where $\chi_{r}^{w}(\theta)$ and $\chi_{b}^{b}(\theta)$
are defined in \eqref{eq:mwb}. See Equations \eqref{eq:vare_id}
and \eqref{eq:sigmaalpha} in Appendix \ref{sec:prooflemma2}. Thus
mathematically $G(m,j,\theta)=I_{m}$ is equivalent to $E\left[\chi_{r}(\theta)|m_{r}=m,D_{r}=j\right]=0$,
with $\chi_{r}(\theta)=(\chi_{r}^{w}(\theta),\chi_{r}^{b}(\theta))$.
Utilizing Lemma \ref{lem:ID_1}, $\theta=\theta_{0}$ is the only
solution to $E\left[\chi_{r}(\theta)|m_{r}=m,D_{r}=j\right]=0$ and
$E\left[\chi_{r}(\theta)|m_{r}=m^{\prime},D_{r}=j^{\prime}\right]=0$
under Scenarios (i) or (ii).Therefore, $\theta_{0}$ is the unique
global maximizer of $\sum_{j=1}^{J}\sum_{m=2}^{\bar{M}}\omega_{m,j}^{\ast}g(m,j,\theta)$
and thus of $\bar{Q}^{\ast}(\theta)$.

\subsection{Proof of Theorem \ref{theorem:Consistency}(b) \label{sec:Proof-of-cons-2}}

To prove the consistency of the QMLE estimator $\hat{\theta}_{N}$
we utilize Lemma 3.1 of \citet{potscher_basic_1991}. Previously we
have shown that $\theta_{0}$ is the unique maximizer of $\bar{Q}^{\ast}(\theta)$
on $\Theta$, where $\bar{Q}^{\ast}(\theta)$ is finite and continuous.
The compactness of $\Theta$ follows from Assumptions \ref{assume:epsilon},
\ref{assume:alpha}, and \ref{assume:lambda}. To prove consistency
of $\hat{\theta}$, it then suffices to have Condition \ref{cond:conv_theta}.
Since $\bar{\beta}(\theta_{0})=\beta_{0}$, once we have shown that
$\hat{\theta}_{N}\rightarrow_{p}\theta_{0}$, consistency of $\hat{\beta}_{N}(\hat{\theta}_{N})$
follows from Condition \ref{cond:conv_gamma}.
\begin{condition}
\label{cond:conv_theta}As $N\rightarrow\infty$ , $\sup_{\theta\in\Theta}|Q_{N}(\theta)-\bar{Q}^{\ast}(\theta)|\rightarrow_{p}0$.
\end{condition}
\begin{condition}
\label{cond:conv_gamma}As $N\rightarrow\infty$, $\sup_{\theta\in\Theta}|\hat{\beta}_{N}(\theta)-\bar{\beta}(\theta)|\rightarrow_{p}0$.
\end{condition}
\noindent $\blacksquare$ \textbf{Verification of Condition \ref{cond:conv_theta}:}
Verification of Condition \ref{cond:lim} has shown that $\sup_{\theta\in\Theta}|E\left[Q_{N}(\theta)\right]-\bar{Q}^{\ast}(\theta)|\rightarrow0$
as $N\rightarrow\infty$. It remains to show that as $N$ goes to
infinity, $\sup_{\theta\in\Theta}|Q_{N}(\theta)-E\left[Q_{N}(\theta)\right]|\rightarrow_{p}0$
. Upon substitution of $Y=(I-\lambda_{0}W)^{-1}(Z\beta_{0}+U)$ into
(\ref{eq:QN}) we have
\[
Q_{N}(\theta)-E\left[Q_{N}(\theta)\right]=\frac{1}{N}\left(U^{\prime}A_{Q_{N}}(\theta)U+2U^{\prime}A_{Q_{N}}(\theta)Z\beta_{0}-\tr[A_{Q_{N}}(\theta)\Omega_{0}]\right),
\]
where 
\[
A_{Q_{N}}(\theta)=-\frac{1}{2}(I-\lambda_{0}W)^{-1}(I-\lambda W)^{\prime}M_{Z}(\theta)(I-\lambda W)(I-\lambda_{0}W)^{-1}.
\]
The row and column sums in absolute value of $(I-\lambda_{0}W)^{-1}$,
$(I-\lambda W)$, $M_{Z}(\theta)$ and their first derivatives are
all uniformly bounded in absolute value. It now follows from Lemma
\ref{lem:Uniform-convergence} that $Q_{N}(\theta)-E\left[Q_{N}(\theta)\right]\rightarrow_{p}0$
uniformly in $\theta$.

\noindent $\blacksquare$ \textbf{Verification of Condition }\ref{cond:conv_gamma}:

In light of (\ref{eq:deltahathat1}) we have
\begin{eqnarray*}
\hat{\beta}_{N}(\theta) & = & (Z^{\prime}\Omega(\theta)^{-1}Z)^{-1}Z^{\prime}\Omega(\theta)^{-1}(I-\lambda W)Y\\
 & = & \left(N^{-1}Z^{\prime}\Omega(\theta)^{-1}Z\right)^{-1}\left(N^{-1}Z^{\prime}\Omega(\theta)^{-1}(I-\lambda W)(I-\lambda_{0}W)^{-1}Z\right)\beta_{0}\\
 & + & \left(N^{-1}Z^{\prime}\Omega(\theta)^{-1}Z\right)^{-1}\left(N^{-1}Z^{\prime}\Omega(\theta)^{-1}(I-\lambda W)(I-\lambda_{0}W)^{-1}U\right).
\end{eqnarray*}

By \prettyref{lem:Uniform-convergence}, 
\[
\sup_{\theta\in\Theta}\left(N^{-1}Z^{\prime}\Omega(\theta)^{-1}(I-\lambda W)(I-\lambda_{0}W)^{-1}U\right)\rightarrow_{p}0.
\]
Also $\left(N^{-1}Z^{\prime}\Omega(\theta_{N})^{-1}Z\right)^{-1}$
is uniformly bounded in absolute value. By Lemma \ref{lem:Aux1},
\[
\sup_{\theta\in\Theta}\left(\left(N^{-1}Z^{\prime}\Omega(\theta)^{-1}Z\right)^{-1}-\varUpsilon_{3}(\theta)^{-1}\right)\rightarrow0
\]
 and 
\[
\sup_{\theta\in\Theta}\left(N^{-1}Z^{\prime}\Omega(\theta)^{-1}(I-\lambda W)(I-\lambda_{0}W)^{-1}Z-\varUpsilon_{2}(\theta)\right)\rightarrow0.
\]
 In all, we have $\sup_{\theta\in\Theta}|\hat{\beta}_{N}(\theta)-\bar{\beta}(\theta)|\rightarrow_{p}0$.

\subsection{Proof of Theorem \ref{theorem:AsymptoticNormlity}\label{sec:Proof-of-asydis}}

To derive the limiting distribution of the QMLE $\hat{\delta}_{N}=(\hat{\theta}_{N}^{\prime},\hat{\beta}_{N}^{\prime})^{\prime}$
it proves more convenient to work with the unconcentrated log-likelihood
function defined in (\ref{eq:logL}). Applying the mean value theorem,
the first order condition for the QMLE can be written as
\[
0=\frac{1}{N^{1/2}}\frac{\partial\ln L_{N}(\hat{\delta}_{N})}{\partial\delta}=\frac{1}{N^{1/2}}\frac{\partial\ln L_{N}(\delta_{0})}{\partial\delta}+\frac{1}{N}\frac{\partial\ln L_{N}(\check{\delta}_{N})}{\partial\delta\partial\delta}N^{1/2}(\hat{\delta}_{N}-\delta_{0}),
\]
where $\check{\delta}_{N}$ denotes a ``between\textquotedblright{}
value vector. Given that $\hat{\delta}_{N}$ was shown to be consistent,
it follows that also the ``between'' value $\check{\delta}_{N}$
is consistent for $\delta_{0}$. It is not difficult to see that 
\[
\frac{\partial\ln L_{N}(\delta)}{\partial\delta}=\left[\begin{array}{l}
-\tr[(I-\lambda W)^{-1}W]+U(\delta){}^{\prime}\Omega^{-1}WY\\
-\frac{1}{2}\tr[\Omega^{-1}\diag_{r=1}^{R}\left\{ m_{r}J_{m_{r}}^{\ast}\right\} ]+\frac{1}{2}U(\delta){}^{\prime}\Omega^{-1}\diag_{r=1}^{R}\left\{ m_{r}J_{m_{r}}^{\ast}\right\} \Omega^{-1}U(\delta)\\
-\frac{1}{2}\tr[\Omega^{-1}\diag_{r=1}^{R}\{1(D_{r}=1)I_{m_{r}}\}]+\frac{1}{2}U(\delta)^{\prime}\Omega^{-2}\diag_{r=1}^{R}\{1(D_{r}=1)I_{m_{r}}\}U(\delta)\\
\vdots\\
-\frac{1}{2}\tr[\Omega^{-1}\diag_{r=1}^{R}\{1(D_{r}=J)I_{m_{r}}\}]+\frac{1}{2}U(\delta)^{\prime}\Omega^{-2}\diag_{r=1}^{R}\{1(D_{r}=J)I_{m_{r}}\}U(\delta)\\
Z^{\prime}\Omega^{-1}U(\delta)
\end{array}\right]
\]
with $U(\delta)=Y-\lambda WY-Z\beta$, and thus
\begin{equation}
\frac{\partial\ln L_{N}(\delta_{0})}{\partial\delta}=\left[\begin{array}{l}
-\tr[(I-\lambda_{0}W)^{-1}W]+U{}^{\prime}\Omega_{0}^{-1}W(I-\lambda_{0}W)^{-1}(Z\beta_{0}+U)\\
-\frac{1}{2}\tr[\Omega_{0}^{-1}\diag_{r=1}^{R}\left\{ J_{m_{r}}^{\ast}m_{r}\right\} ]+\frac{1}{2}U{}^{\prime}\Omega_{0}^{-1}\diag_{r=1}^{R}\left\{ J_{m_{r}}^{\ast}m_{r}\right\} \Omega_{0}^{-1}U\\
-\frac{1}{2}\tr[\Omega_{0}^{-1}\diag_{r=1}^{R}\{1(D_{r}=1)I_{m_{r}}\}]+\frac{1}{2}U^{\prime}\Omega_{0}^{-2}\diag_{r=1}^{R}\{1(D_{r}=1)I_{m_{r}}\}U\\
\vdots\\
-\frac{1}{2}\tr[\Omega_{0}^{-1}\diag_{r=1}^{R}\{1(D_{r}=J)I_{m_{r}}\}]+\frac{1}{2}U^{\prime}\Omega_{0}^{-2}\diag_{r=1}^{R}\{1(D_{r}=J)I_{m_{r}}\}U\\
Z^{\prime}\Omega_{0}^{-1}U
\end{array}\right].\label{eq:score}
\end{equation}

Furthermore, it is not difficult to see that with $\theta_{2}=\sigma_{\alpha}^{2}$,
$\theta_{2+j}=\sigma_{\epsilon,j}^{2}$, $j=1,...,J$, the elements
of the Hessian matrix are
\begin{eqnarray}
\frac{\partial^{2}\ln L_{N}(\delta)}{\partial\lambda^{2}} & = & -\tr[(I-\lambda W)^{-2}W^{2}]-Y^{\prime}W^{\prime}\Omega(\theta)^{-1}WY,\label{eq:hessian1}\\
\frac{\partial^{2}\ln L_{N}(\delta)}{\partial\lambda\partial\theta_{i}} & = & -U(\delta)^{\prime}\Omega(\theta)^{-1}\frac{\partial\Omega(\theta)}{\partial\theta_{i}}\Omega(\theta)^{-1}WY,\\
\frac{\partial^{2}\ln L_{N}(\delta)}{\partial\theta_{i}\partial\theta_{j}} & = & \frac{1}{2}\tr[\Omega(\theta)^{-2}\frac{\partial\Omega}{\partial\theta_{i}}\frac{\partial\Omega}{\partial\theta_{j}}]\\
 &  & -U(\delta)^{\prime}\Omega(\theta)^{-1}\frac{\partial\Omega(\theta)}{\partial\theta_{i}}\Omega(\theta)^{-1}\frac{\partial\Omega(\theta)}{\partial\theta_{j}}\Omega(\theta)^{-1}U(\delta),\\
\frac{\partial^{2}\ln L_{N}(\delta)}{\partial\theta_{i}\partial\beta} & = & -Z^{\prime}\Omega(\theta)^{-1}\frac{\partial\Omega(\theta)}{\partial\theta_{i}}\Omega(\theta)^{-1}U(\delta),\\
\frac{\partial^{2}\ln L_{N}(\delta)}{\partial\lambda\partial\beta} & = & -Z^{\prime}\Omega(\theta)^{-1}WY,\\
\frac{\partial^{2}\ln L_{N}(\delta)}{\partial\beta\partial\beta^{\prime}} & = & -Z^{\prime}\Omega(\theta)^{-1}Z,\label{eq:hessian2}
\end{eqnarray}
with $i,j=2,3,...,2+J$ and 
\[
\frac{\partial\Omega(\theta)}{\partial\theta_{2}}=\diag_{r=1}^{R}\{J_{m_{r}}^{\ast}m_{r}\},\qquad\frac{\partial\Omega(\theta)}{\partial\theta_{2+j}}=\diag_{r=1}^{R}\{1(D_{r}=j)I_{m_{r}}\}.
\]

Since $Y=(I-\lambda_{0}W)^{-1}(Z\beta_{0}+U)$ and $U(\delta)=(I-\lambda W)(I-\lambda_{0}W)^{-1}(Z\beta_{0}+U)-Z\beta$,
each element of $N^{-1}\partial^{2}\ln L_{N}(\delta)/\partial\delta\partial\delta^{\prime}$
is a linear quadratic form of $U$ or $Z$ in the form of $\frac{1}{N}U^{\prime}A(\theta)U$,
$\frac{1}{N}Z^{\prime}A(\theta)U$, $\frac{1}{N}Z^{\prime}A(\theta)Z$,
and their products with $\beta$: $\frac{1}{N}\beta^{\prime}Z^{\prime}A(\theta)U$,
$\frac{1}{N}Z^{\prime}A(\theta)Z\beta$, $\frac{1}{N}\beta^{\prime}Z^{\prime}A(\theta)Z\beta$,
etc., where $A(\theta)=\diag_{r=1}^{R}\{p(m_{r},D_{r},\theta)I_{m_{r}}^{\ast}+s(m_{r},D_{r},\theta)J_{m_{r}}^{\ast}\}$
satisfies the conditions in \prettyref{lem:Uniform-convergence} and
\prettyref{lem:Aux1}. By these two lemmas, if $\hat{\theta}_{N}\rightarrow_{p}\theta_{0}$,
all three types of linear quadratic forms converge to the limit of
their expected value at $\theta_{0}$ in probability. That is as $N\rightarrow\infty$,
\[
\left|\frac{1}{N}U^{\prime}A(\hat{\theta}_{N})U-\lim_{N\rightarrow\infty}\frac{1}{N}\tr[A(\theta_{0})\Omega_{0}]\right|\rightarrow_{p}0,
\]
\[
|\frac{1}{N}Z^{\prime}A(\theta)U|\rightarrow_{p}0,
\]
\[
\left|\frac{1}{N}Z^{\prime}A(\hat{\theta}_{N})Z-\lim_{N\rightarrow\infty}\frac{1}{N}Z^{\prime}A(\theta_{0})Z\right|\rightarrow_{p}0.
\]
Also, we have $\hat{\beta}_{N}\rightarrow_{p}\beta_{0}$. Thus by
Slutsky's theorem, the products of the linear quadratic forms with
$\hat{\beta}_{N}$ converge in probability to the products of the
expected values with $\beta_{0}$. Therefore, as $\check{\delta}_{N}\rightarrow_{p}\delta_{0}$,

\[
\frac{1}{N}\frac{\partial\ln L_{N}(\check{\delta}_{N})}{\partial\delta\partial\delta^{\prime}}\overset{}{\rightarrow_{p}}\lim_{N\rightarrow\infty}\frac{1}{N}E\left[\frac{\partial^{2}\ln L_{N}(\delta_{0})}{\partial\delta\partial\delta^{\prime}}\right]=-\Gamma_{0},
\]
where the specific structure of $\Gamma_{0}$ is given in Appendix
\ref{sec:Variance-Covariance-Matrix}.

We next show that $N^{1/2}\partial\ln L_{N}(\delta_{0})/\partial\delta\xrightarrow{d}N(0,\Upsilon_{0})$.
Each element of the score function in (\ref{eq:score}) can be written
as a linear quadratic form of $U$ in the form of $U^{\prime}A_{N}(\theta_{0})U+U^{\prime}B_{N}(\theta_{0})Z\beta_{0}+C(\delta_{0})$,
which has zero mean and where the row and column sums of $A_{N}(\theta_{0})$
and $B_{N}(\theta_{0})$ are uniformly bounded in absolute value and
where $C(\delta_{0})$ are constants. Using Theorem \ref{lem:clt},
$N^{-1/2}\partial\ln L_{N}(\delta_{0})/\partial\delta\xrightarrow{d}N(0,\Upsilon_{0})$
with $\Upsilon_{0}=\lim_{N\rightarrow\infty}\frac{1}{N}E[\frac{\partial\ln L_{N}(\delta_{0})}{\partial\delta}\frac{\partial\ln L_{N}(\delta_{0})}{\partial\delta^{\prime}}]$,
whose expression is given in Appendix \ref{sec:Variance-Covariance-Matrix}.
In all, $\sqrt{N}(\hat{\delta}_{N}-\delta_{0})\xrightarrow{d}N(0,\Gamma_{0}^{-1}\Upsilon_{0}\Gamma_{0}^{-1})$
as $N$ goes to infinity.

\subsection{Proof of Theorem \ref{thm:ValidInference}\label{sec:Proof-consistency34}}

The proof of the theorem follows from the definition of $\hat{\Gamma}_{N}$
and $\hat{\Upsilon}_{N}$ in (\ref{eq:Gamma_hat}) and (\ref{eq:Upsilon_hat}),
Lemma \ref{Lemma:consistency_moment34} which is stated below in this
section, Assumptions \ref{assume:n} and \ref{assume:z} and Theorem
\ref{theorem:Consistency} which together imply that $\hat{\Gamma}_{N}\overset{p}{\rightarrow}\Gamma_{0}$
and $\hat{\Upsilon}_{N}\overset{p}{\rightarrow}\Upsilon_{0}$. Since
by Lemma \ref{lem:nd} the matrices $\Upsilon_{0}$ and $\Gamma_{0}$
are full rank, and thus $\Gamma_{0}^{-1}\Upsilon_{0}\Gamma_{0}^{-1}$
is full rank, it follows that $\left(\hat{\Gamma}_{N}^{-1}\hat{\Upsilon}_{N}\hat{\Gamma}_{N}^{-1}\right)^{-1/2}\overset{p}{\rightarrow}\left(\Gamma_{0}^{-1}\Upsilon_{0}\Gamma_{0}^{-1}\right)^{-1/2}$.
The result then follows from the continuous mapping theorem and Theorem
\ref{theorem:AsymptoticNormlity}. 
\begin{lem}
\label{Lemma:consistency_moment34}Suppose Assumptions 1-5 hold, then
$\hat{\mu}_{\alpha}^{(3)}\overset{p}{\rightarrow}\mu_{\alpha0}^{(3)}$,
$\hat{\mu}_{\alpha}^{(4)}\overset{p}{\rightarrow}\mu_{\alpha0}^{(4)}$,
and $\hat{\mu}_{\epsilon,j}^{(3)}\overset{p}{\rightarrow}\mu_{\epsilon0,j}^{(3)}$,
$\hat{\mu}_{\epsilon,j}^{(4)}\overset{p}{\rightarrow}\mu_{\epsilon0,j}^{(4)}$,
for any $j\in\{1,...,J\}$. 
\end{lem}
First note that $f_{\alpha,r}^{(3)},f_{\alpha,r}^{(4)},f_{\epsilon,r}^{(3)}$
and $f_{\epsilon,r}^{(4)}$ can all be rewritten as a linear combination
of finitely many terms in the form of $\psi(m_{r})\sum_{i=1}^{m_{r}}\ddot{u}_{ir}^{p_{1}}\bar{u}_{r}^{p_{2}}$
and some nonstochastic function $f(m_{r},\theta)$, with $\psi(m_{r})$
being a finite function, $p_{1}\geqslant0,p_{2}\geqslant0$ and $p_{1}+p_{2}\leqslant4$,
and $f(m,\theta)$ being continuous in $\theta$. Also note that 
\begin{align*}
\frac{1}{R}\sum_{r=1}^{R}f_{\alpha,r}^{(l)} & =\sum_{j=1}^{J}\frac{R_{j}}{R}\left(\frac{1}{R_{j}}\sum_{r=1}^{R}1(D_{r}=j)f_{\alpha,r}^{(l)}\right)
\end{align*}
 for $l=3$ and 4. Thus our estimates for the third and fourth moments,
$\frac{1}{R}\sum_{r=1}^{R}f_{\alpha,r}^{(l)}$ and $\frac{1}{R_{j}}\sum_{r=1}^{R}1(D_{r}=j)f_{\epsilon,r}^{(l)}$,
$l=3,4$ can all be written as a weighted sum of finitely many terms
of the form 
\begin{equation}
\frac{1}{R_{j}}\sum_{r=1}^{R}1(D_{r}=j)\psi(m_{r})\sum_{i=1}^{m_{r}}\ddot{u}_{ir}^{p_{1}}\bar{u}_{r}^{p_{2}},\label{eq:term1}
\end{equation}
 and 
\begin{equation}
\frac{1}{R_{j}}\sum_{r=1}^{R}1(D_{r}=j)f(m_{r},\theta).\label{eq:term2}
\end{equation}
The former converges to its mean by Lemma \ref{lem:moment34} (a)
and the latter is nonstochastic. Consequently, $\frac{1}{R}\sum_{r=1}^{R}f_{\alpha,r}^{(l)}\rightarrow_{p}\mu_{\alpha0}^{(l)}$
and $\frac{1}{R_{j}}\sum_{r=1}^{R}1(D_{r}=j)f_{\epsilon,r}^{(l)}\rightarrow_{p}\mu_{\epsilon0,j}^{(l)}$
for $l=3,4$ and $j=1,...,J$ as $R$ goes to infinity. For the feasible
counterparts of the terms in \eqref{eq:term1} and \eqref{eq:term2},
note that $\hat{\theta}-\theta_{0}\rightarrow_{p}0$ and $f(m,\theta)$
is continuous in $\theta$. We have
\[
\frac{1}{R_{j}}\sum_{r=1}^{R}1(D_{r}=j)f(m_{r},\hat{\theta})-\frac{1}{R_{j}}\sum_{r=1}^{R}1(D_{r}=j)f(m_{r},\theta_{0})\rightarrow_{p}0
\]
 by the continuous mapping theorem as. It thus remains to show that
\begin{equation}
\frac{1}{R_{j}}\sum_{r=1}^{R}1(D_{r}=j)\psi(m_{r})\sum_{i=1}^{m_{r}}(\hat{\ddot{u}}_{ir}^{p_{1}}\hat{\bar{u}}_{r}^{p_{2}}-\ddot{u}_{ir}^{p_{1}}\bar{u}_{r}^{p_{2}})\rightarrow_{p}0.\label{eq:convp}
\end{equation}

Let $\hat{U}_{r}=(\hat{u}_{1r},...,\hat{u}_{m_{r},r})^{\prime}$,
then\emph{ }
\begin{align*}
\hat{U}_{r} & =(I-\hat{\lambda}W)Y_{r}-Z_{r}\hat{\beta}=(I-\hat{\lambda}W)(I-\lambda_{0}W)^{-1}(Z_{r}\beta_{0}+U_{r})-Z_{r}\hat{\beta}\\
 & =\left((I-\hat{\lambda}W)(I-\lambda_{0}W)^{-1}Z_{r}\beta_{0}-Z_{r}\hat{\beta}\right)+(I-\hat{\lambda}W)(I-\lambda_{0}W)^{-1}U_{r}.
\end{align*}
Note that 
\[
(I-\hat{\lambda}W)(I-\lambda_{0}W)^{-1}=\frac{m_{r}-1+\hat{\lambda}}{m_{r}-1+\lambda_{0}}I_{m_{r}}^{\ast}+\frac{1-\hat{\lambda}}{1-\lambda_{0}}J_{m_{r}}^{\ast},
\]
where $J_{m_{r}}^{\ast}=\iota_{m_{r}}\iota_{m_{r}}^{\prime}/m_{r}$
and $I_{m_{r}}^{\ast}=I_{m_{r}}-J_{m_{r}}^{\ast}$ are two orthogonal
idempotent matrices that generate vectors of group means and vectors
of deviations from the group means. Thus 
\begin{align*}
\hat{\bar{u}}_{r} & =\iota_{m_{r}}^{\prime}\hat{U}_{r}/m_{r}=\bar{z}_{r}(\frac{(1-\hat{\lambda})\beta_{0}}{1-\lambda_{0}}-\hat{\beta})+\frac{1-\hat{\lambda}}{1-\lambda_{0}}\bar{u}_{r}=\bar{z}_{r}\bar{\phi}+\bar{\varphi}\bar{u}_{r},
\end{align*}
where $\bar{\phi}=\frac{(1-\hat{\lambda})\beta_{0}}{1-\lambda_{0}}-\hat{\beta}$,
$\bar{\varphi}=\frac{1-\hat{\lambda}}{1-\lambda_{0}}$, $\bar{z}_{r}=\frac{1}{m_{r}}\sum_{i=1}^{m_{r}}z_{ir}$.
Let $\hat{\ddot{U}}_{r}=(\hat{\ddot{u}}_{1r},...,\hat{\ddot{u}}_{m_{r}r})^{\prime}$,
then 
\[
\hat{\ddot{U}}_{r}=I_{m_{r}}^{\ast}\hat{U}_{r}=\ddot{Z}_{r}\left(\frac{m_{r}-1+\hat{\lambda}}{m_{r}-1+\lambda_{0}}\beta_{0}-\hat{\beta}\right)+\frac{m_{r}-1+\hat{\lambda}}{m_{r}-1+\lambda_{0}}\ddot{U}_{r},
\]
and
\begin{align*}
\ddot{\hat{u}}_{ir} & =\ddot{z}_{ir}\ddot{\phi}_{r}+\ddot{\varphi}_{r}\ddot{u}_{ir},
\end{align*}
where $\ddot{\phi}_{r}=\frac{m_{r}-1+\hat{\lambda}}{m_{r}-1+\lambda_{0}}\beta_{0}-\hat{\beta}$,
$\ddot{\varphi}_{r}=1+\frac{\hat{\lambda}-\lambda_{0}}{m_{r}-1+\lambda_{0}}$,
$\ddot{z}_{ir}=z_{ir}-\bar{z}_{r}$.

In all, 
\[
\hat{\ddot{u}}_{ir}^{p_{1}}\hat{\bar{u}}_{r}^{p_{2}}-\ddot{u}_{ir}^{p_{1}}\bar{u}_{r}^{p_{2}}=(\ddot{z}_{ir}\ddot{\phi}_{r}+\ddot{\varphi}_{r}\ddot{u}_{ir})^{p_{1}}(\bar{z}_{r}\bar{\phi}+\bar{\varphi}\bar{u}_{r})^{p_{2}}-\ddot{u}_{ir}^{p_{1}}\bar{u}_{r}^{p_{2}}.
\]
Given that $p_{1}$ and $p_{2}$ are nonnegative integers with $p_{1}+p_{2}\leqslant4,$
the above equation can be written as a linear combination of terms
of the form $(\ddot{z}_{ir}\ddot{\phi}_{r})^{s_{1}}(\ddot{\varphi}_{r}\ddot{u}_{ir})^{p_{1}-s_{1}}(\bar{z}_{r}\bar{\phi})^{s_{2}}(\bar{\varphi}\bar{u}_{r})^{p_{2}-s_{2}}$,
and $(\ddot{\varphi}_{r}\ddot{u}_{ir})^{p_{1}}(\bar{\varphi}\bar{u}_{r})^{p_{2}}-\ddot{u}_{ir}^{p_{1}}\bar{u}_{r}^{p_{2}}$
with $0\leqslant s_{1}\leqslant p_{1}$, $0\leqslant s_{2}\leqslant p_{2}$
and $s_{1}+s_{2}\geqslant1$. The claim in \eqref{eq:convp} now follows
immediately from Lemma \ref{lem:moment34}(b)-(c).

\section{Variance-Covariance Matrix\label{sec:Variance-Covariance-Matrix}
and Proof of Lemma \ref{lem:nd}}

\subsection{Variance-Covariance Matrix}

Recall that $\Gamma_{0}=$$\lim_{N\rightarrow\infty}-\frac{1}{N}E\left[\frac{\partial^{2}\ln L_{N}(\delta_{0})}{\partial\delta\partial\delta^{\prime}}\right]$
and $\Upsilon_{0}=\lim_{N\rightarrow\infty}\frac{1}{N}E\left[\frac{\partial\ln L_{N}(\delta_{0})}{\partial\delta}\frac{\partial\ln L_{N}(\delta_{0})}{\partial\delta^{\prime}}\right]$.
These matrices are of dimension $(2+J+k_{Z})\times(2+J+k_{Z})$, symmetric,
and underlie the expression for the limiting variance covariance matrix
of the QMLE estimator for $\delta_{0}$. In the following we give
explicit expressions for $\Gamma_{0}$ and $\Upsilon_{0}$. Detailed
derivations are provided in the Online Appendix. We have
\begin{eqnarray}
\Upsilon_{0} & = & \sum_{j=1}^{J}\sum_{m=2}^{\bar{M}}\varphi(m,j)\bar{\Psi}(m,j)\varphi(m,j)^{\prime},\label{eq:Upsilon}
\end{eqnarray}
and 
\begin{equation}
\Gamma_{0}=\sum_{j=1}^{J}\sum_{m=2}^{\bar{M}}\varphi(m,j)\bar{\Psi}_{G}(m,j)\varphi(m,j)^{\prime},\label{eq:Gamma}
\end{equation}
where 
\begin{equation}
\varphi(m,j)=\left(\begin{array}{cccc}
\frac{1}{(m-1+\lambda_{0})\sigma_{\epsilon0,j}^{2}} & -\frac{m}{(1-\lambda_{0})(\sigma_{\epsilon0,j}^{2}+m\sigma_{\alpha0}^{2})} & \frac{1}{(m-1+\lambda_{0})\sigma_{\epsilon0,j}^{2}}\beta_{0}^{\prime} & -\frac{m}{(1-\lambda_{0})(\sigma_{\epsilon0,j}^{2}+m\sigma_{\alpha0}^{2})}\beta_{0}^{\prime}\\
0 & -\frac{m^{2}}{2(\sigma_{\epsilon0,j}^{2}+m\sigma_{\alpha0}^{2})^{2}} & 0 & 0\\
-\frac{1(j=1)}{2\sigma_{\epsilon0,j}^{4}} & -\frac{m1(j=1)}{2(\sigma_{\epsilon0,j}^{2}+m\sigma_{\alpha0}^{2})^{2}} & 0 & 0\\
\vdots & \vdots & \vdots & \vdots\\
-\frac{1(j=J)}{2\sigma_{\epsilon0,j}^{4}} & -\frac{m1(j=J)}{2(\sigma_{\epsilon0,j}^{2}+m\sigma_{\alpha0}^{2})^{2}} & 0 & 0\\
0 & 0 & -\frac{1}{\sigma_{\epsilon0,j}^{2}}I_{k_{Z}} & -\frac{m}{\sigma_{\epsilon0,j}^{2}+m\sigma_{\alpha0}^{2}}I_{k_{Z}}
\end{array}\right),\label{eq:varphim-1}
\end{equation}
which is given in (\ref{eq:varphim}) and repeated here for the convenience
of the reader,
\begin{equation}
\bar{\Psi}_{G}(m,j)=\diag\{2(m-1)\sigma_{\varepsilon0,j}^{4}\frac{\omega_{m,j}^{\ast}}{m^{\ast}},2(\sigma_{\alpha0}^{2}+\frac{\sigma_{\varepsilon0,j}^{2}}{m})^{2}\frac{\omega_{m,j}^{\ast}}{m^{\ast}},\sigma_{\varepsilon0}^{2}\ddot{\varkappa}_{m,j},(\sigma_{\alpha0}^{2}+\frac{\sigma_{\varepsilon0,j}^{2}}{m})\frac{\bar{\varkappa}_{m,j}}{m}\},\label{eq:PsiG}
\end{equation}
\begin{equation}
\bar{\Psi}(m,j)=\left[\begin{array}{cc}
\bar{\Psi}_{11}(m,j) & \bar{\Psi}_{12}(m,j)\\
\bar{\Psi}_{21}(m,j) & \bar{\Psi}_{22}(m,j)
\end{array}\right],\label{eq:Psim}
\end{equation}
with 
\begin{eqnarray*}
\bar{\Psi}_{11}(m,j) & = & \frac{\omega_{m,j}^{\ast}}{m^{\ast}}\left[\begin{array}{cc}
2(m-1)\sigma_{\varepsilon0,j}^{4} & 0\\
0 & 2(\sigma_{\alpha0}^{2}+\frac{\sigma_{\varepsilon0,j}^{2}}{m})^{2}+(\mu_{\alpha0}^{(4)}-3\sigma_{\alpha0}^{4})
\end{array}\right]\\
 &  & +(\mu_{\varepsilon0,j}^{(4)}-3\sigma_{\varepsilon0,j}^{4})\frac{\omega_{m,j}^{\ast}}{m^{\ast}}\left[\begin{array}{cc}
\frac{(m-1)^{2}}{m} & \frac{(m-1)}{m^{2}}\\
\frac{(m-1)}{m^{2}} & \frac{1}{m^{3}}
\end{array}\right],\\
\bar{\Psi}_{21}(m,j) & = & \left[\begin{array}{cc}
0 & 0\\
\frac{m-1}{m}\mu_{\varepsilon0,j}^{(3)}\bar{z}_{m,j}^{\prime}\  & [\mu_{\alpha0}^{(3)}+\frac{1}{m^{2}}\mu_{\epsilon0,j}^{(3)}]\bar{z}_{m,j}^{\prime}
\end{array}\right]=\Psi_{12}^{\prime}(m,j),\\
\bar{\Psi}_{22}(m,j) & = & \left[\begin{array}{cc}
\sigma_{\varepsilon0,j}^{2}\ddot{\varkappa}_{m,j} & 0\\
0 & (\sigma_{\alpha0}^{2}+\sigma_{\varepsilon0,j}^{2}/m)\frac{\bar{\varkappa}_{m,j}}{m}
\end{array}\right].
\end{eqnarray*}
Note that $\bar{\Psi}_{G}(m,j)$ can be obtained by setting $\mu_{\epsilon0,j}^{(4)}-3\sigma_{\epsilon0,j}^{4}=\mu_{\alpha0}^{(4)}-3\sigma_{\alpha0}^{4}=\mu_{\alpha0}^{(3)}=\mu_{\epsilon0,j}^{(3)}=0$
in $\bar{\Psi}(m,j)$. When $\epsilon$ and $\alpha$ are both Gaussian,
$\Upsilon_{0}=\Gamma_{0}$, consistent with what is expected from
the information matrix equality.

\subsection{Proof of the\textcolor{black}{{} Positive Definiteness of $\Upsilon_{0}$
and $\Gamma_{0}$}}

Let $\varphi(m,j)$ , $\bar{\Psi}_{G}(m,j)$ and $\bar{\Psi}(m,j)$
be as defined in \eqref{eq:varphim-1}, \eqref{eq:PsiG} and \eqref{eq:Psim}
respectively. We can partition $\varphi(m,j)$ as $\varphi(m,j)=\left(\begin{array}{cc}
A_{m,j} & B_{m,j}\\
0 & C_{m,j}
\end{array}\right)$, where 
\begin{equation}
A_{m,j}=\left(\begin{array}{cc}
\frac{1}{(m-1+\lambda_{0})\sigma_{\epsilon0,j}^{2}} & -\frac{m}{(1-\lambda_{0})(\sigma_{\epsilon0,j}^{2}+m\sigma_{\alpha0}^{2})}\\
0 & -\frac{m^{2}}{2(\sigma_{\epsilon0,j}^{2}+m\sigma_{\alpha0}^{2})^{2}}\\
-\frac{1(j=1)}{2\sigma_{\epsilon0,j}^{4}} & -\frac{m1(j=1)}{2(\sigma_{\epsilon0,j}^{2}+m\sigma_{\alpha0}^{2})^{2}}\\
\vdots & \vdots\\
-\frac{1(j=J)}{2\sigma_{\epsilon0,j}^{4}} & -\frac{m1(j=J)}{2(\sigma_{\epsilon0,j}^{2}+m\sigma_{\alpha0}^{2})^{2}}
\end{array}\right)\label{eq:Amj}
\end{equation}
is a $(2+J)\times2$ matrix , $B_{m,j}$ is the upper right block,
\begin{equation}
C_{m,j}=[-\frac{1}{\sigma_{\epsilon0,j}^{2}}I_{k_{Z}},-\frac{m}{\sigma_{\epsilon0,j}^{2}+m\sigma_{\alpha0}^{2}}I_{k_{Z}}].\label{eq:Cmj}
\end{equation}
Let $\ell=(\ell_{1},\ell_{2},\ell_{3}^{\prime},\ell_{4}^{\prime})^{\prime}$
be a $(2+2k_{Z})$ dimensional vector, where $\ell_{1}$ and $\ell_{2}$
are scalars and $\ell_{3}$ and $\ell_{4}$ are both $k_{Z}$ dimensional
vectors. To prove Lemma \ref{lem:nd}, we introduce the three lemmas
below.
\begin{lem}
\label{lem:PsiG}Suppose Assumptions 1-5 hold \textup{and} $\omega_{m,j}^{\ast}>0$,
then $\ell^{\prime}\bar{\Psi}_{G}(m,j)\ell=0$ if and only if $\ell_{1}=\ell_{2}=0$,
$\ell_{3}\ddot{\varkappa}_{m,j}\ell_{3}=0$ and $\ell_{4}\bar{\varkappa}_{m,j}\ell_{4}=0$.
\end{lem}
\begin{lem}
\label{lem:Psi}Suppose Assumptions 1-5 hold and assume further that
\textup{$\mu_{\varepsilon0,j}^{(4)}-\sigma_{\varepsilon0,j}^{4}>(\mu_{\epsilon0,j}^{(3)})^{2}/\sigma_{\varepsilon0,j}^{2}$
and} $\omega_{m,j}^{\ast}>0$, then $\ell^{\prime}\bar{\Psi}(m,j)\ell=0$
if and only if $\ell_{1}=\ell_{2}=0$, $\ell_{3}\ddot{\varkappa}_{m,j}\ell_{3}=0$
and $\ell_{4}\bar{\varkappa}_{m,j}\ell_{4}=0$.
\end{lem}
Let $\{(m,j)|\omega_{m,j}^{\ast}>0\}$ be the set of all pairs of
$(m,j)$ such that $\omega_{m,j}^{\ast}>0$, and index its elements
with $p=1,...,\bar{P}$. We therefore have $\omega_{m_{p},j_{p}}^{\ast}>0$
for $p=1,...,\bar{P}$. Note that for all $j=1,...,J$, there exists
some $p$ such that $j_{p}=j$. This is because for each $j$ there
exists some $m$ such that $\omega_{m,j}^{\ast}>0$, observing that
$\omega_{j}^{\ast}=\sum_{m=2}^{\bar{M}}\omega_{m,j}^{\ast}>0$ all
$j$, and $m\leqslant\bar{M}$ is bounded. The set of all $A_{m,j}$
defined in \eqref{eq:Amj} with $\omega_{m,j}^{\ast}>0$ is $\{A_{m,j}|\omega_{m,j}^{\ast}>0\}=\{A_{m_{1},j_{1}},...,A_{m_{\bar{P}},j_{\bar{P}}}\}$.
The Lemma below states that the column by column concatenation of
all matrices in this set has full row rank.
\begin{lem}
\label{lem:rank}Suppose Assumptions 1-6 hold, then the matrix $\Phi=[A_{m_{1},j_{1}}...,A_{m_{\bar{P}},j_{\bar{P}}}]$
has full row rank.
\end{lem}
Lemma \ref{lem:PsiG} follows easily from \eqref{eq:PsiG}, observing
that $\omega_{m,j}^{\ast}>0$, $\sigma_{\epsilon0,j}^{2}>0$. The
proofs of Lemma \ref{lem:Psi} and Lemma \ref{lem:rank} are given
in the Online Appendix.

We can now utilize the above lemmas to prove that under the maintained
assumptions $\Gamma_{0}$ is positive definite, and that the matrix
$\Upsilon_{0}$ is positive definite for $\mu_{\varepsilon0,j}^{(4)}-\sigma_{\varepsilon0,j}^{4}>(\mu_{\epsilon0,j}^{(3)})^{2}/\sigma_{\varepsilon0,j}^{2}$.
We present a proof for the positive definiteness of $\Upsilon_{0}$.
The proof for the positive definiteness of $\Gamma_{0}$ is analogous.

Let $\alpha=(\alpha_{1}^{\prime},\alpha_{2}^{\prime})$ be a $(2+J+k_{Z})$
vector, where $\alpha_{1}$ is a $(2+J)\times1$ vector and $\alpha_{2}$
is a $k_{Z}\times1$ vector. To show that $\Upsilon_{0}$ is positive
definite is equivalent to showing that $\alpha^{\prime}\Upsilon\alpha=0$
if and only if $\alpha_{1}=0$ and $\alpha_{2}=0$. Observe that 
\begin{align*}
\alpha^{\prime}\Upsilon_{0}\alpha & =\sum_{j=1}^{J}\sum_{m=2,\omega_{m,j}^{\ast}>0}^{\bar{M}}\alpha^{\prime}\varphi(m,j)\bar{\Psi}(m,j)\varphi(m,j)^{\prime}\alpha\\
 & =\sum_{j=1}^{J}\sum_{m=2,\omega_{m,j}^{\ast}>0}^{\bar{M}}\ell_{m,j}^{\prime}\bar{\Psi}(m,j)\ell_{m,j},
\end{align*}
where 
\[
\ell_{m,j}=\varphi(m,j)^{\prime}\alpha=\left(\begin{array}{c}
A_{m,j}^{\prime}\alpha_{1}\\
B_{m,j}^{\prime}\alpha_{1}+C_{m,j}^{\prime}\alpha_{2}
\end{array}\right),
\]
with $A_{m,j}$ and $C_{m,j}$ defined in Equations \eqref{eq:Amj}
and \eqref{eq:Cmj}. It thus suffices to show that $\ell_{m,j}^{\prime}\bar{\Psi}(m,j)\ell_{m,j}=0$
for all $m,j$ with $\omega_{m,j}^{\ast}>0$ if and only if $\alpha_{1}=0$
and $\alpha_{2}=0$. Given Lemma \ref{lem:Psi}, if $\omega_{m,j}^{\ast}>0$
and $\ell_{m,j}^{\prime}\bar{\Psi}(m,j)\ell_{m,j}=0$ then $A_{m,j}^{\prime}\alpha_{1}=0$.
Lemma \ref{lem:rank} indicates that for $A_{m,j}^{\prime}\alpha_{1}=0$
to hold for all $m$ and $j$, we must have $\alpha_{1}=0$. With
$\alpha_{1}=0$, 
\begin{align*}
\ell_{m,j} & =\left(\begin{array}{c}
0\\
C_{m,j}^{\prime}\alpha_{2}
\end{array}\right)=-\left(\begin{array}{c}
0\\
\frac{1}{\sigma_{\epsilon0,j}^{2}}\alpha_{2}\\
\frac{m}{\sigma_{\epsilon0,j}^{2}+m\sigma_{\alpha0}^{2}}\alpha_{2}
\end{array}\right).
\end{align*}
Utilizing Lemma \ref{lem:Psi} again and noting that then $\frac{1}{\sigma_{\epsilon0,j}^{2}}>0$
and $\frac{m}{\sigma_{\epsilon0,j}^{2}+m\sigma_{\alpha0}^{2}}>0$,
we have $\alpha_{2}^{\prime}\ddot{\varkappa}_{m,j}\alpha_{2}=0$ and
$\alpha_{2}^{\prime}\bar{\varkappa}_{m,j}\alpha_{2}=0$ for all $m$
and $j$. Consequently, 
\[
\alpha_{2}^{\prime}\sum_{j=1}^{J}\sum_{m=2}^{\bar{M}}(\ddot{\varkappa}_{m,j}+\bar{\varkappa}_{m,j})\alpha_{2}=0.
\]
 This gives $\alpha_{2}=0$ as $\sum_{j=1}^{J}\sum_{m=2}^{\bar{M}}(\ddot{\varkappa}_{m,j}+\bar{\varkappa}_{m,j})$
is positive definite under Assumption \ref{assume:z}. In all, $\alpha^{\prime}\Upsilon_{0}\alpha=0$
if and only if $\alpha=(\alpha_{1}^{\prime},\alpha_{2}^{\prime})=0$
hence $\Upsilon_{0}$ is positive definite. The proof of the positive
definiteness of $\Gamma_{0}$ follows similarly.
\end{document}